\documentclass[sigconf]{acmart}
\usepackage[ruled,vlined,linesnumbered]{algorithm2e}
\usepackage{amsmath,amssymb,amsfonts}
\usepackage{graphicx}
\usepackage{textcomp}
\usepackage{tabularray}
\usepackage{xcolor}
\usepackage{subfigure}
\usepackage{stfloats}
\usepackage{multicol}
\usepackage{booktabs}
\newtheorem{theorem}{Theorem}
\usepackage{amsthm}
\usepackage{tabularx}
\usepackage{enumitem}
\usepackage{xcolor}

\setcopyright{none}
\AtBeginDocument{%
  }

\setcopyright{acmlicensed}
\copyrightyear{2018}
\acmYear{2018}
\acmDOI{XXXXXXX.XXXXXXX}
\acmConference[WWW '26]{Make sure to enter the correct
  conference title from your rights confirmation email}{April 13–17, 2026}{Dubai, United Arab Emirates.}

\acmISBN{978-1-4503-XXXX-X/2018/06}

\begin{document}
\begin{sloppypar}

\title{Cardinality is Not Enough: Super Host Detection via Segmented Cardinality Estimation}

%\settopmatter{authorsperrow=4}

\author{Yilin Zhao}
\affiliation{
\institution{Central South University}
\city{Changsha}
\country{China}
}
\email{yilinzhao@csu.edu.cn}

\author{Jiawei Huang}
\affiliation{
\institution{Central South University}
\city{Changsha}
\country{China}
}
\email{jiaweihuang@csu.edu.cn}

\author{Xianshi Su}
\affiliation{
\institution{Central South University}
\city{Changsha}
\country{China}
}
\email{suxianshi@csu.edu.cn}

\author{Weihe Li}
\affiliation{
\institution{The University of Edinburgh}
\city{Edinburgh}
\country{United Kingdom}
}
\email{weihe.li@ed.ac.uk}

\author{Xin Li}
\affiliation{
\institution{Central South University}
\city{Changsha}
\country{China}
}
\email{xinli@csu.edu.cn}

\author{Yan Liu}
\affiliation{
\institution{Central South University}
\city{Changsha}
\country{China}
}
\email{yliu20974@gmail.com}

\author{Jiacheng Xie}
\affiliation{
\institution{Central South University}
\city{Changsha}
\country{China}
}
\email{jiachengxie@csu.edu.cn}

\author{Qichen Su}
\affiliation{
\institution{Central South University}
\city{Changsha}
\country{China}
}
\email{qichensu@csu.edu.cn}

\author{Jin Ye}
\affiliation{
\institution{Guangxi University}
\city{Nanning}
\country{China}
}
\email{qichensu@csu.edu.cn}

\author{Wanchun Jiang}
\affiliation{
\institution{Central South University}
\city{Changsha}
\country{China}
}
\email{jiangwc@csu.edu.cn}

\author{Jianxin Wang}
\affiliation{
\institution{Central South University}
\city{Changsha}
\country{China}
}
\email{jxwang@mail.csu.edu.cn}

\iffalse
\author{
  \IEEEauthorblockN{
    Yilin Zhao\textsuperscript{1},
    Jiawei Huang\textsuperscript{1}, 
    Xianshi Su\textsuperscript{1},
    Weihe Li\textsuperscript{2},
    Xin Li\textsuperscript{1},
    Yan Liu\textsuperscript{1},
    Jiacheng Xie\textsuperscript{1},
    Qichen Su\textsuperscript{1},
    Jin Ye\textsuperscript{3}, 
    Wanchun Jiang\textsuperscript{1},
    Jianxin Wang\textsuperscript{1}
  }\\
  \IEEEauthorblockA{\textsuperscript{1}School of Computer Science and Engineering, Central South University, Changsha, China}\\
  \IEEEauthorblockA{\textsuperscript{2}School of Informatics, The University of Edinburgh, Edinburgh, United Kingdom}\\
  \IEEEauthorblockA{\textsuperscript{3}School of Computer, Electronics and Information, Guangxi University, Nanning, China}\\
  \IEEEauthorblockA{\{yilinzhao, jiaweihuang, suxianshi, xinli, jiachengxie, qichensu, jiangwc\}@csu.edu.cn},
  \IEEEauthorblockA{\{weihe.li\}@ed.ac.uk}, 
  \IEEEauthorblockA{\{yliu20974\}@gmail.com},
  \IEEEauthorblockA{\{yejin\}@gxu.edu.cn},
  \IEEEauthorblockA{\{jxwang\}@mail.csu.edu.cn}
}
\fi
  
\settopmatter{printacmref=true} % Removes citation information below abstract

\copyrightyear{2026}
\acmYear{2026}
\setcopyright{cc}
\setcctype{by}
\acmConference[WWW '26]{Proceedings of the ACM Web Conference 2026}{April 13--17, 2026}{Dubai, United Arab Emirates}
\acmBooktitle{Proceedings of the ACM Web Conference 2026 (WWW '26), April 13--17, 2026, Dubai, United Arab Emirates}
\acmPrice{}
\acmDOI{10.1145/3774904.3792071}
\acmISBN{979-8-4007-2307-0/2026/04}

\begin{abstract}
Accurately detecting super host that establishes connections to a large number of distinct peers is significant for mitigating web attacks and ensuring high quality of web service. Existing sketch-based approaches estimate the number of distinct connections called flow cardinality according to full IP addresses, while ignoring the fact that a malicious or victim super host often communicates with hosts within the same subnet, resulting in high false positive rates and low accuracy. Though hierarchical-structure based approaches could capture flow cardinality in subnet, they inherently suffer from high memory usage. To address these limitations, we propose SegSketch, a segmented cardinality estimation approach that employs a lightweight halved-segment hashing strategy to infer common prefix lengths of IP addresses, and estimates cardinality within subnet to enhance detection accuracy under constrained memory size. Experiments driven by real-world traces demonstrate that, SegSketch improves F1-Score by up to $8.04\times$ compared to state-of-the-art solutions, particularly under small memory budgets. 
\end{abstract}



\renewcommand{\shortauthors}{Yilin Zhao et al.}
\maketitle

\newcommand\webconfavailabilityurl{https://doi.org/10.5281/zenodo.17284567}
\ifdefempty{\webconfavailabilityurl}{}{
\begingroup\small\noindent\raggedright\textbf{Resource Availability:}\\
The source code of this paper has been made publicly available at \url{\webconfavailabilityurl}.\cite{seg}
\endgroup
}

\section{Introduction}
In recent years, super spreader and super receiver (collectively called super host) have become increasingly prevalent in web attacks \cite{network, spreadsketch}. A super spreader is typically a host that connects to a large number of distinct destinations and launches various types of attacks such as IP scanning \cite{port, internet}, spam distribution \cite{spam, report}, and Carpet Bombing DDoS attacks \cite{13, 7, 11, fast, 41, dollm, gmcb}. For instance, the Mirai botnet infected nearly 65,000 IoT devices within the first 20 hours of its outbreak and later stabilized at approximately 200,000 $\sim$ 300,000 active infections \cite{infection}. In contrast, a super receiver is a host receiving traffic from numerous distinct sources, such as traditional DDoS attacks from botnets \cite{modern, anomaly, 2025report}. These super-host attacks degrade web service quality and inflate operations costs \cite{fs}, but remain difficult to detect under constrained monitoring resources and high traffic load \cite{geometric, spike}.
    
Owing to achieving a good trade-off among accuracy, speed and memory usage, sketch has emerged as the canonical solution for approximate flow estimation, especially in identifying super hosts \cite{spreadsketch, spike, geometric, kjoin, couper, supersketch}. These approaches typically employ cardinality estimators (e.g., Linear Counting \cite{lc} or HyperLogLog \cite{hll} estimators) to identify super hosts. For example, SpreadSketch \cite{spreadsketch} finds super hosts by estimating the number of flows (i.e., flow cardinality) from a host to different destinations. However, these approaches struggle to distinguish malicious hosts from benign ones, which may also have large flow cardinality. For instance, some web servers or DNS resolvers also exhibit high flow cardinality due to their huge number of legitimate connections with diverse peers.
%For instance, some analytics services or metadata catalogs also exhibit high access cardinality due to their large number of legitimate queries with diverse clients.

%    

Essentially, almost all state-of-the-art theories and practices on super-host detection have focused on single-dimension statistics of flow cardinality in isolation from other informative traffic patterns, causing the false positive problem and low detection accuracy. In this paper, through empirical study, we reveal that the flows sent to or from a super host usually have the same network address in source or destination IP addresses, since attacking flows are usually generated by a large number of zombie hosts within the same subnet or are sent to many victim hosts located in the same subnet. %, which is a pattern commonly observed in clustered database deployments where contiguous IP addresses are assigned to adjacent nodes or partitions. 
These facts indicate that the number of flows (called subnet cardinality) with the same subnet address can be used to assist super-host detection. For instance, contemporary web services are typically deployed as geographically co-located clusters where hosts share a common network address. Consequently, subnet-oriented attacks are far more devastating than per-IP assaults, disabling the whole set of web-service instances serving a particular application.
     
However, the IP prefix length is unknown in advance and can vary across networks. The intuitive solution to estimate subnet cardinality is employing the hierarchical structure with multiple layers to cover all possible IP prefix lengths (e.g., /8, /16, /24) and respectively count corresponding subnet cardinality \cite{sigcomm17, sigmod04, vldb, alenex, tkdd08, itc09}. However, this hierarchical structure introduces prohibitive memory overhead, making it impractical in real-world deployments.

To tackle this challenge, we propose a novel super-host detection approach that accurately estimates the subnet cardinality under constrained memory. Our design introduces a halved-segment hashing strategy to map flows with the same prefix into compact region of bitmap, enabling lightweight estimation of subnet length. Moreover, we incorporate the bitmap-based cardinality estimation to quantify the number of distinct connections sharing the same subnet address. Therefore, our approach accurately identifies flows of super host characterized by same network address and high subnet cardinality, while maintaining much lower memory overhead compared to approaches based on the hierarchical structure.

% In summary, our contributions are as follows:
% \begin{itemize}
%     \item   We propose SegSketch, a memory-efficient sketch that integrates cardinality estimation with a segment-wise guided hashing mechanism to identify the super hosts. It achieves a good balance between memory overhead and detection accuracy for super hosts characterized by both same IP prefix and high subnet cardinality (\S\ref{sec:design}). 
%     \item 	We construct mathematical models to theoretically analyze %the performance of SegSketch, %including %the memory complexity under segment-wise bitmap organization, the time complexity of per-packet updates and halved-segment hashing
%     the error bound of subnet cardinality estimation associated with SegSketch (\S\ref{sec:theoretical analysis}).
%     \item 	We conduct evaluation based on real-world datasets. The results show that SegSketch achieves higher detection accuracy and more efficient memory utilization than state-of-the-art approaches. Furthermore, we implement SegSketch using P4\cite{p4} and deploy it on a programmable switch with small hardware resource usage (\S\ref{sec:evaluation}, \S\ref{sec:P4 implementation}).
% \end{itemize}

In summary, our contributions are as follows:

% \vspace{-4mm}
% \begin{tabularx}{0.95\linewidth}{@{}p{1.5em}@{}X@{}}
% \bigdot & We propose SegSketch, a memory-efficient sketch that integrates cardinality estimation with a segment-wise guided hashing mechanism to identify the super hosts. It achieves a good balance between memory overhead and detection accuracy for super hosts characterized by both same IP prefix and high subnet cardinality (§3). \\[0.8ex]
% \bigdot & We construct mathematical models to theoretically analyze the error bound of subnet cardinality estimation associated with SegSketch (§4). \\[0.8ex]
% \bigdot & We conduct evaluation based on real-world datasets. The results show that SegSketch achieves higher detection accuracy and more efficient memory utilization than state-of-the-art approaches. Furthermore, we implement SegSketch using P4~\cite{p4} and deploy it on a programmable switch with small hardware resource usage (§5, §6). \\
% \end{tabularx}
% \vspace{-3mm}
\begin{itemize}[leftmargin=2em,labelsep=0.5em,topsep=-0.5ex,itemsep=0ex,parsep=0pt] % 缩进4个字符宽度
\item We propose SegSketch, a memory-efficient sketch that integrates cardinality estimation with halved-segment hashing to identify super hosts. It achieves a good balance between memory overhead and detection accuracy for super hosts characterized by same IP prefix and high subnet cardinality (§3).
\item We establish a mathematical model to analyze the error bound for subnet cardinality estimation of SegSketch. Furthermore, we prove that using host-address hashing consistently yields smaller error than using full-address hashing (§4).
%We establish a theoretical model of SegSketch with error bounds for subnet cardinality estimation, and prove that host-address hashing yields lower error than full-address hashing (§4).
%We establish a mathematical model to analyze the error bound for subnet cardinality estimation of SegSketch. Furthermore, we theoretically demonstrate that using host-address hashing consistently yields smaller estimation error of subnet cardinality than using full-address hashing (§4).
%\item We construct a theoretical model to analyze the error bound of subnet cardinality estimation of SegSketch (§4). 
\item %We conduct evaluation based on real-world datasets. The results show that SegSketch achieves higher detection accuracy and more efficient memory utilization, with up to $8.97\times$ F1-Score improvement over state-of-the-art approaches under the same memory budget. Furthermore, we implement SegSketch using P4~\cite{p4} and deploy it on a programmable switch with small hardware resource usage (§5, §6).
We conduct trace-driven evaluation, and deploy SegSketch on a programmable switch using P4 \cite{p4} with low hardware overhead, requiring only 1.77\% of SRAM. SegSketch achieves higher detection accuracy and more efficient memory utilization, reaching up to $8.04\times$ improvement in F1-Score (§5, §6).
\end{itemize}

%The remainder of this paper is structured as follows. Section~\ref{sec:motivation} outlines the motivation behind our work. Section~\S\ref{sec:design} details the design of the proposed method. In Section~\S\ref{sec:theoretical analysis}, we provide the mathematical analyses. Section~\S\ref{sec:evaluation} presents the trace-driven performance evaluation, while Section~\S\ref{sec:P4 implementation} describes the P4-based implementation. Section~\ref{sec:related work} reviews related studies. Finally, Section~\ref{sec:conclusion} concludes the paper.
\vspace{-2mm}
\section{Motivation}
\label{sec:motivation}

\subsection{Existing Solutions and Limitations}\label{AA}

A super host sends or receives a huge number of flows with distinct IP addresses. If the estimated flow cardinality of a host exceeds a given threshold, it will be identified as a super host. To address the challenge of limited memory size, fast processing speed and high accuracy requirement, sketch-based approaches use compact data structures to estimate flow cardinality \cite{spreadsketch, couper, supersketch, extended, geometric, cds}. These approaches typically hash IP addresses into a bitmap. Since the same IP address will only be hashed into one bit in the bitmap, one flow will be recorded in the same bit, no matter how many packets it delivers. Due to hash collisions, however, the flow cardinality may be under-estimated. To address this issue, the advanced algorithms such as Linear Counting \cite{lc}, HyperLogLog \cite{hll}, and Multi-Resolution Bitmap \cite{mrb}, employ various mathematical techniques, such as Bernoulli trials, to estimate the true flow cardinality. %of the flow. 

However, most super-host detection approaches only focus on the flow cardinality estimation of full IP address, without awareness of the fact that connections of a super host usually share the same subnet address. We measure the percentage of super hosts whose connections are within a single subnet across three attack datasets: SDN-DDoS traffic dataset \cite{sdn}, MACCDC2012 dataset \cite{maccdc}, and UNSW-NB15 dataset \cite{unsw}. The results show that nearly every super host establishes connections with a single subnet, with SDN-DDoS exhibiting the highest ratio at 100\%, followed by MACCDC2012 at 98.11\%, and UNSW-NB15 at 95.45\%.

The unawareness of this feature results in low accuracy of super host detection. Specifically, as illustrated in Figure~\ref{fig:ss}, host 1 with IP address $src_1$ connects to many hosts within a single subnet, while host 2 with IP address $src_2$ connects to hosts across different subnets. All the destination IP addresses of each host are hashed in one bitmap. According to the number of hashed bits, the flow cardinality can be estimated. As a result, both hosts have equal estimated flow cardinality and are identified as super hosts. 
%Host 1 is the real attacker that lanuchs IP scanning in one subnet. However, host 2 is the benign host that connects to many destinations in the whole network. 
However, host 1 is the real attacker that launches IP scanning in one subnet, while host 2 is the benign host that connects to many destinations in the whole network. This result indicates that, without the subnet information, current flow cardinality estimators easily cause the false positive problem in detection of real attacker.

\vspace{-1mm}  
 \begin{figure}[htbp]
\vspace{-3mm}
\centering  
\includegraphics[width=0.83\linewidth]{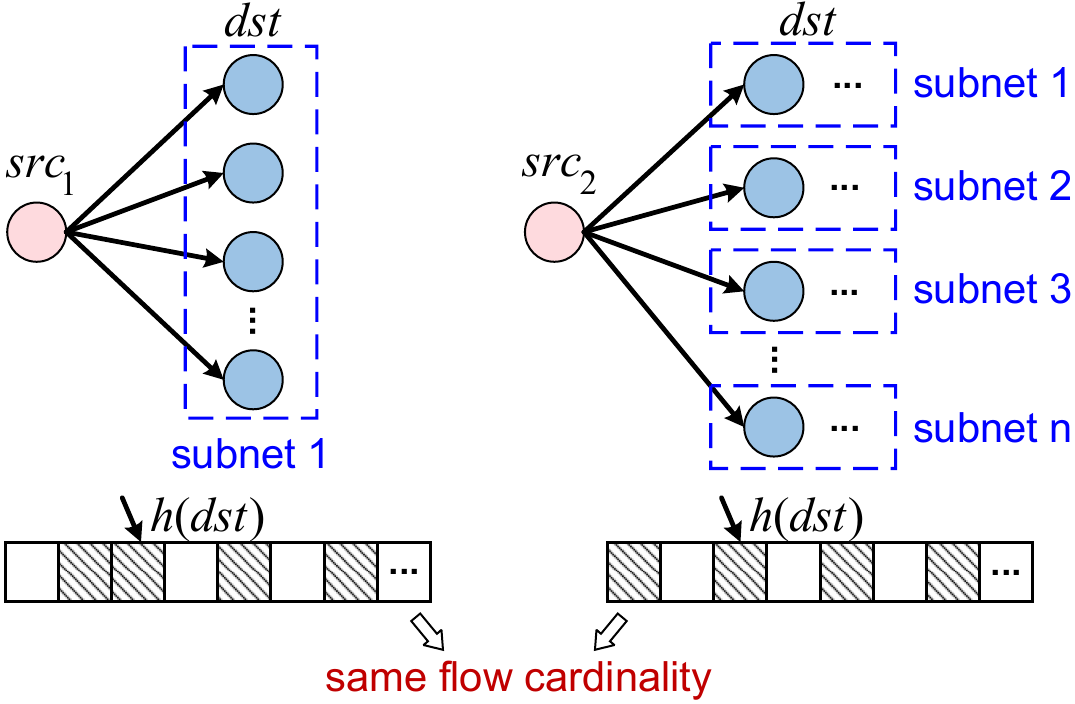}
\vspace{-2mm}
\caption{Sketch-based flow cardinality estimation approach.}
\vspace{-1mm}
 \label{fig:ss}
\end{figure}

To address this issue, the hierarchical cardinality estimators \cite{sigcomm17, sigmod04, vldb, alenex, tkdd08, itc09} have been proposed. For example, Randomized Hierarchical Heavy Hitter (RHHH) \cite{sigcomm17} has a hierarchical structure with multiple layers. As shown in Figure~\ref{fig:hhh}, each layer uses one cardinality estimator to track the flow cardinality for a specific subnet length. Specifically, for an incoming packet, the host address with different lengths in its destination address is hashed into the corresponding layers. For example, layer 1 estimates subnet cardinality with the last 24 bits of address. Hosts whose flow cardinality exceeds a given threshold in the layer are classified as super hosts. 
   
\vspace{-2mm}
\begin{figure}[htbp]
\centering  
\includegraphics[width=0.74\linewidth]{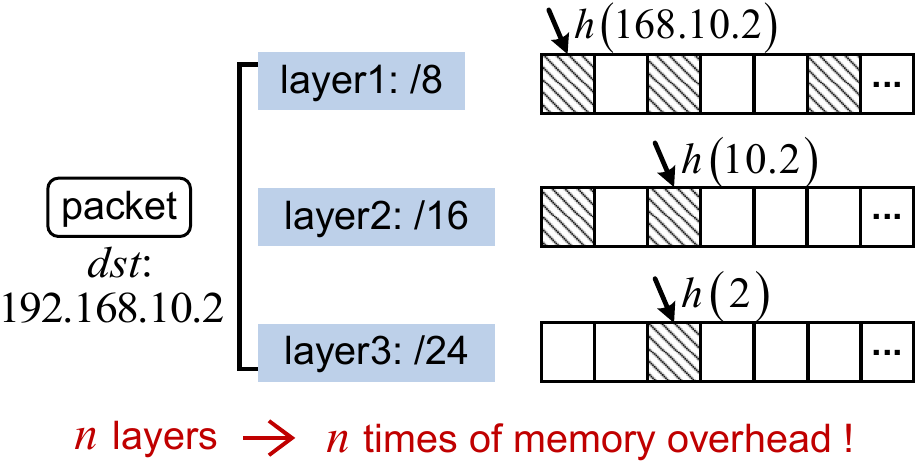}
\vspace{-1mm}
	\caption{Hierarchical approach for cardinality estimation.}
   \vspace{-2mm}
 \label{fig:hhh}
\end{figure}

Unfortunately, since the subnet length is unknown in advance and can vary across networks, the hierarchical cardinality estimator typically faces a dilemma between detection accuracy and memory usage. Increasing the number of layers can improve detection accuracy by detecting more subnet lengths, but this comes with a dramatic increase in memory usage, making it difficult to be deployed on network devices with small on-chip memory.
 
%\vspace{-1mm}

\subsection{Empirical Study}\label{AA}
To gain deep insights, we adopt two single-layer cardinality estimation approaches including SpreadSketch \cite{spreadsketch} and Couper \cite{couper}, together with the multi-layer cardinality estimator RHHH to conduct empirical tests on a combined dataset of UNSW-NB15 \cite{unsw} and CAIDA2016 \cite{caida}. This mixed dataset includes realistic normal and attack flows, such as IP scanning and worm propagation.

\vspace{-3mm}
\begin{figure}[htbp]
    \centering
    % Case1-1
    \subfigure[Precision]{
        \includegraphics[width=0.475\linewidth]{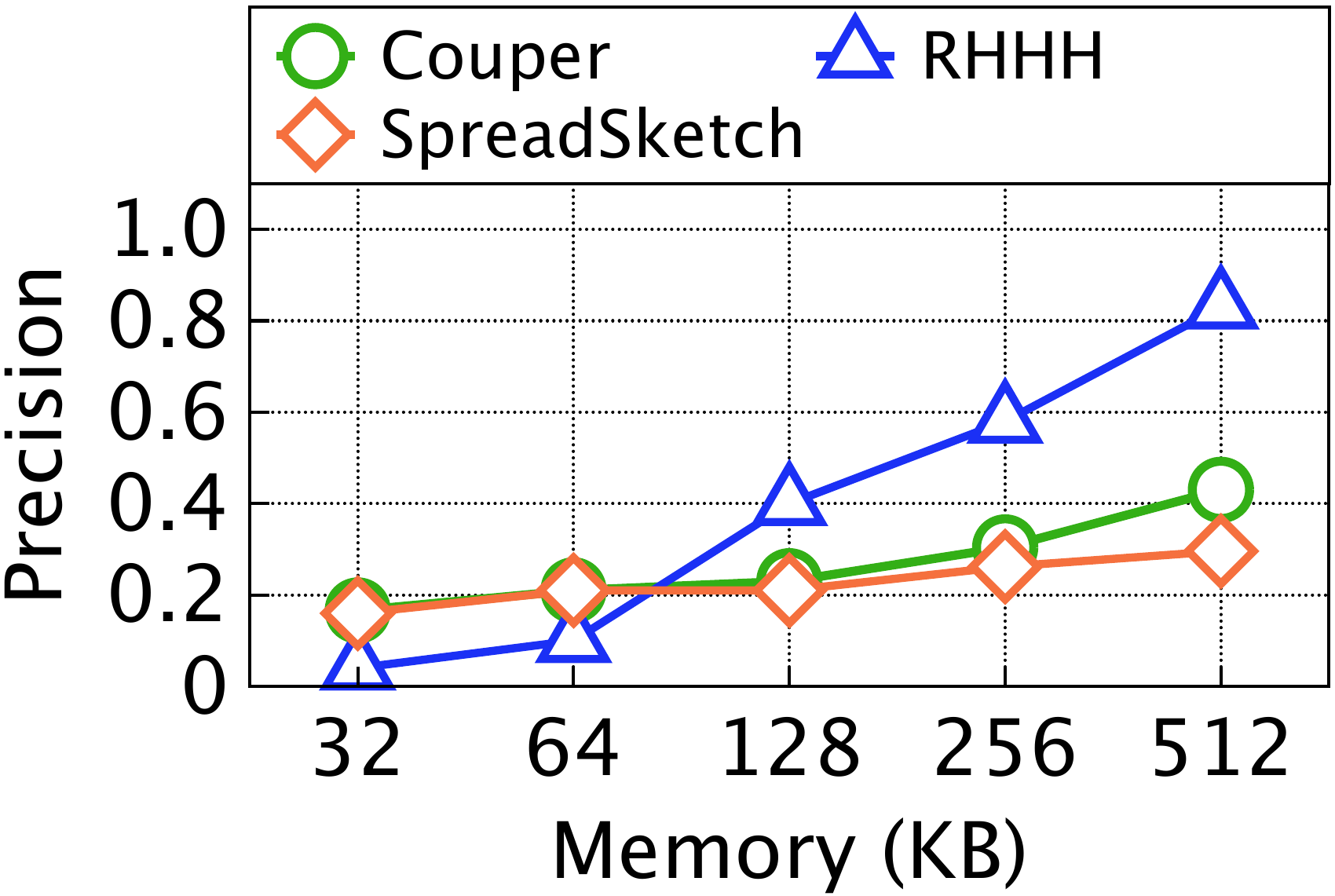}
     \label{fig:a}
    }
    \subfigure[Recall]{
        \includegraphics[width=0.475\linewidth]{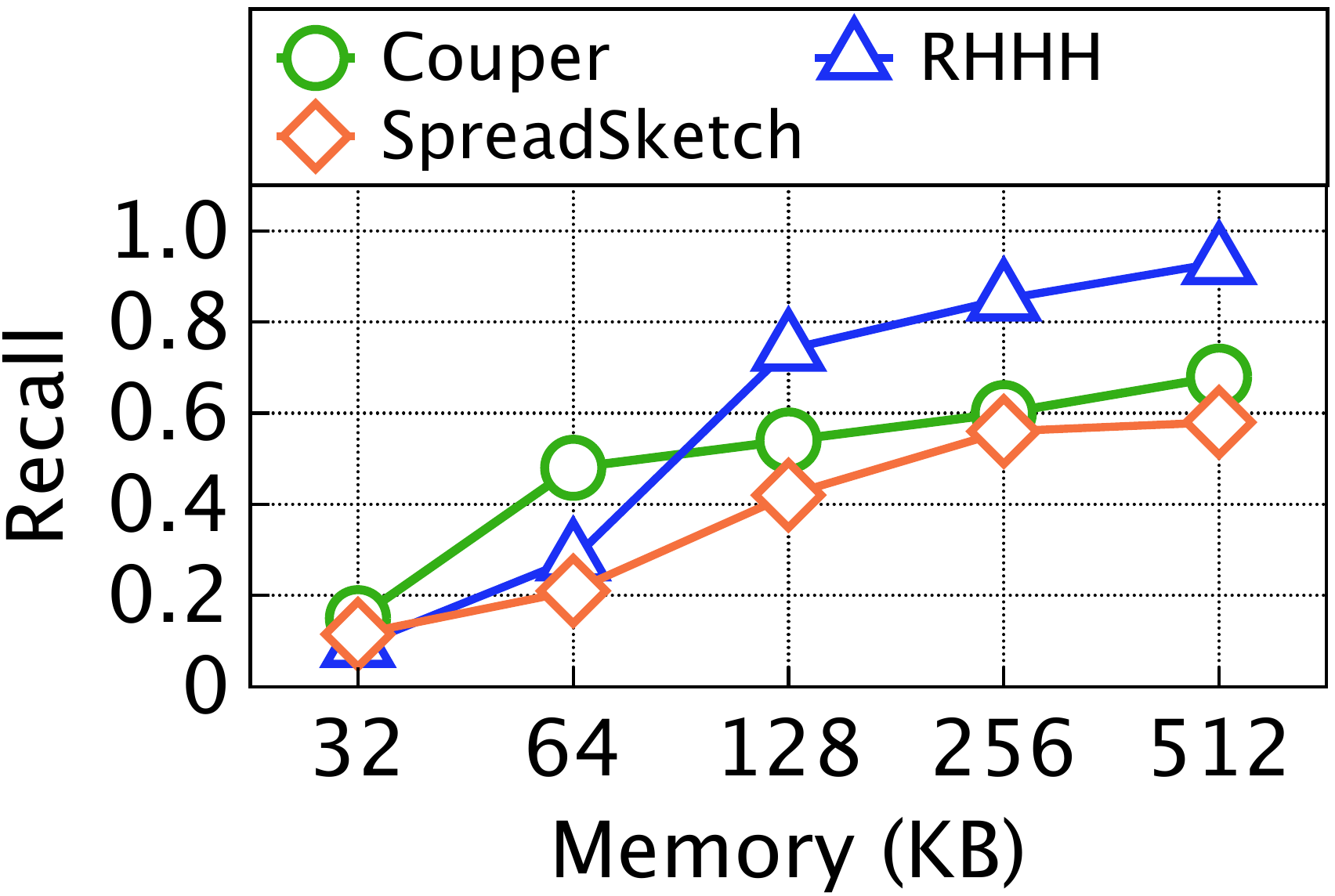}
     \label{fig:b}
    }        
    \\
    % Case2
    \subfigure[F1-Score]{
        \includegraphics[width=0.475\linewidth]{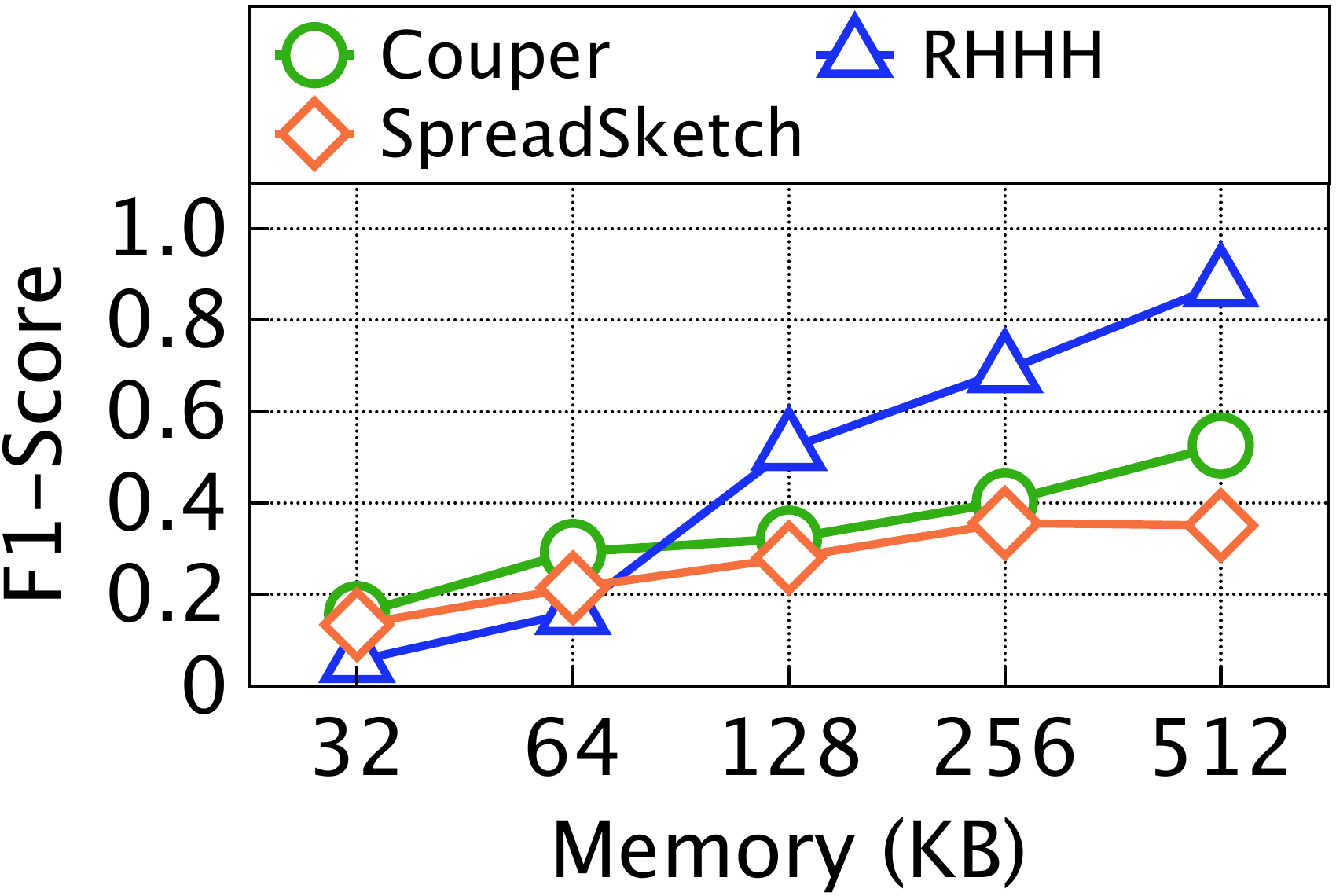}
    \label{fig:c}

    }
    % Case3
    \subfigure[ARE]{
        \includegraphics[width=0.475\linewidth]{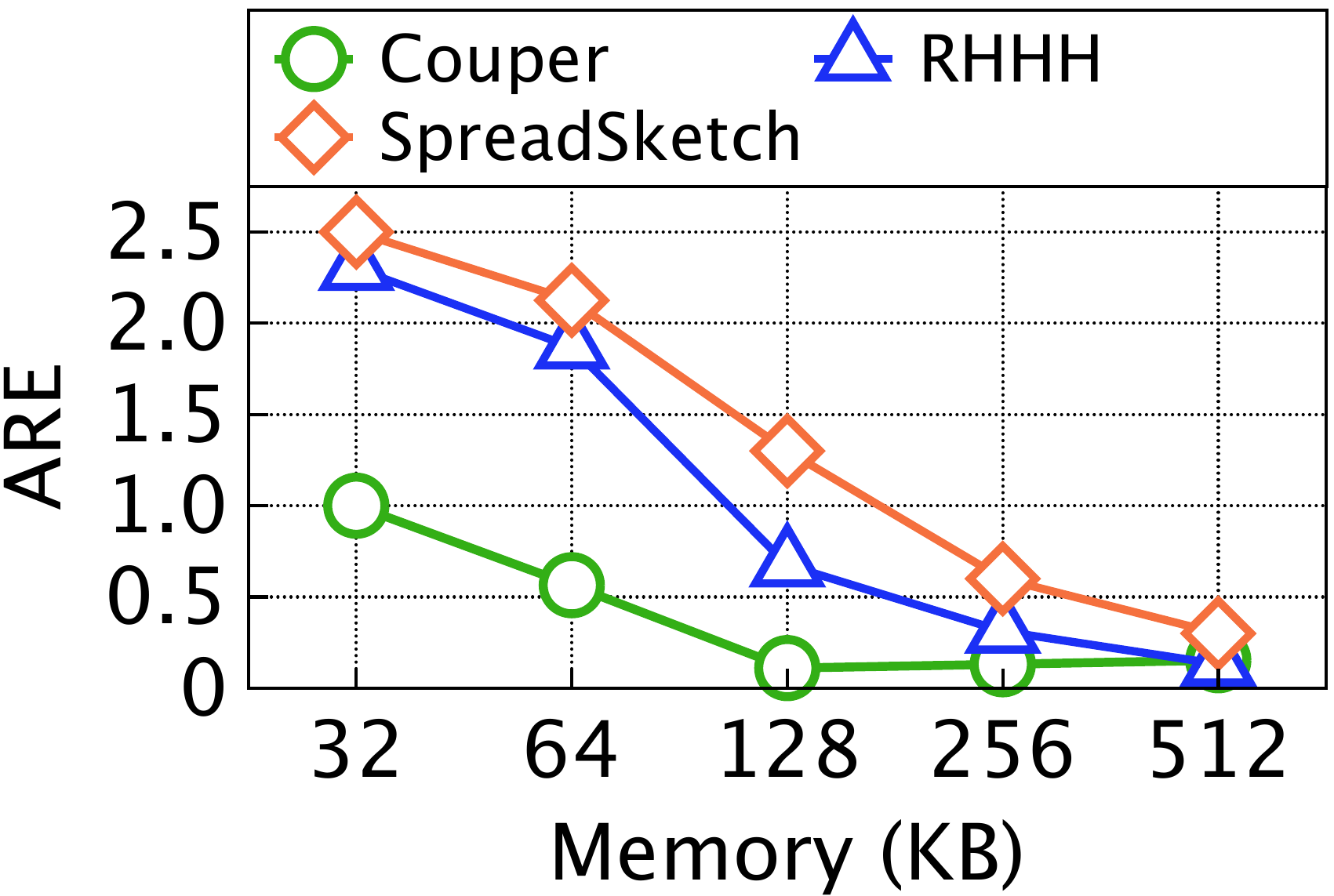}
    \label{fig:d}
    }
  \vspace{-3mm}
    \caption{Super spreader detection performance.}
    \vspace{-2mm}
    \label{fig:mawii}
\end{figure}
    
The experimental results in Figure~\ref{fig:mawii} show that all these approaches demonstrate relatively low precision in identifying attack flows across all memory sizes, with RHHH exhibiting higher accuracy compared to the two sketch-based cardinality estimation methods. This is due to the inability of sketch-based cardinality estimation methods to distinguish whether the high-cardinality flows are within a single subnet or distributed over the entire network. On the other hand, despite its higher precision, RHHH struggles to achieve optimal performance under tight memory budgets due to its high overhead introduced by the multi-layer structure.

   %\vspace{-1mm}
   
Based on the above investigation, we reveal that, under the small memory limitation, it is hard for existing sketch-based or hierarchical cardinality estimation approaches to identify super hosts with huge number of connections sharing the same subnet addresses. These findings motivate us to design a novel approach to accurately identify super hosts based on the common IP prefix identification.

\section{Design of SegSketch}
\label{sec:design}

To address the above challenge, we propose a segmented cardinality estimation scheme called SegSketch that utilizes a compact structure and efficiently identifies super hosts via common IP prefix identification and cardinality estimation.

\subsection{Data Structure}\label{AA}
The data structure of SegSketch is shown in Figure~\ref{fig:data structure}, which consists of $r$ rows, each containing $c$ buckets. Each row is associated with a pairwise-independent hash function $h_1, ..., h_r$. $B(i, j)$ represents the bucket at the $i^{th}$ row and $j^{th}$ column, where $1 \leq i \leq r$ and $1 \leq j \leq c$. Each bucket consists of three parts: (i) the host key; (ii) the subnet bitmap aiming for estimating the common IP prefix length; and (iii) the host bitmap used to estimate the subnet cardinality. 
          
\vspace{-1mm}
\begin{figure}[!h]
\vspace{-2mm}
    \centering
    \includegraphics[width=0.94\linewidth]{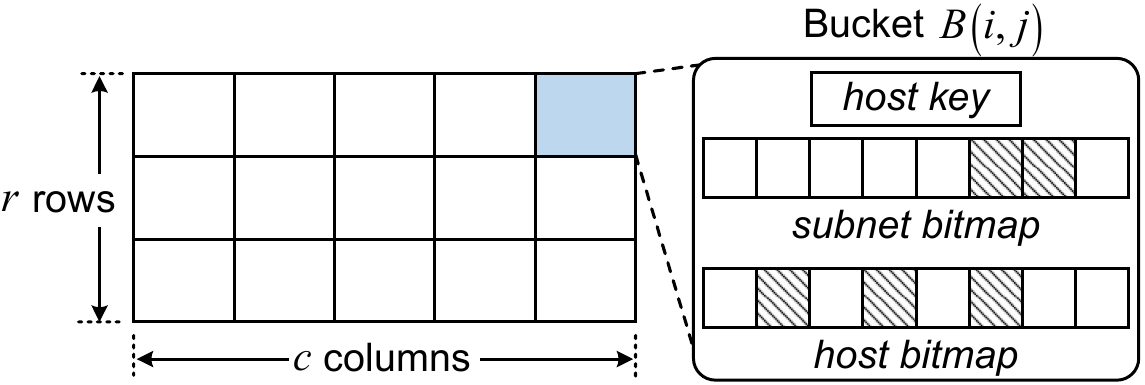}
   \vspace{-2mm}
    \caption{Data structure.}
    \vspace{-2mm}
    \label{fig:data structure}
\end{figure}
\vspace{-2mm}

\subsection{Key Design}\label{AA}
The subnet cardinality is an important feature to assist super-host detection. However, it is hard to get the prior knowledge of subnet information at the measurement nodes, making the subnet cardinality estimation impractical. Different from the hierarchical cardinality estimator requiring large memory usage, the core idea of SegSketch is to leverage a lightweight halved-segment hashing strategy to infer the subnet address length, thereby enabling accurate estimation of flow cardinality within the subnet.

{
\renewcommand{\baselinestretch}{0.9}\selectfont  % 调整为你想要的值，比如 0.85～0.95
\begin{algorithm}
\caption{Halved-Segment Hashing}
\label{alg:HalvedSegmentHashing}
\KwIn{IP address, subnet bitmap}
\KwOut{$k*G$} %\tcp{the subnet length} %with segment width $G$

$V \gets$ number of segments of the IP address\;
$L[1..V] \gets$ segment\ set\;

$left \gets 0$\;
$right \gets$ size of the subnet bitmap $- 1$\;
$k=1$;

\While{$left < right$}{
    $h \gets hash(L[k]) \bmod 2$\;
    $mid \gets (left + right) / 2$\;
    
    \eIf{$h == 1$}{
    %\tcp{the 2-value hash result chooses the right half of the subnet bitmap}
        \If{find(subnet bitmap, left, mid) }{
        %\tcp{search the left half to check if any position is set to 1}
            subnet bitmap$[mid] \gets 1$\;
            \Return{$k$}\;
        }
        $left \gets mid$\;
    }{
        \If{find(subnet bitmap, mid, right)}{
            subnet bitmap$[mid] \gets 1$\;
            \Return{$k$}\;
        }
        $right \gets mid$\;
    }
    
    $k \gets k + 1$\;
}
subnet bitmap$[left] \gets 1$\;
\Return{$k$}\;
\end{algorithm}
}

\textbf{Halved-segment hashing strategy}. As shown in Algorithm 1, the halved-segment hashing strategy firstly divides each IP address into $V$ segments (denoted as $L[1..V]$) (Lines 1 $\sim$ 2), each consisting of a fixed length of $G$ bits. Once a packet arrives, the 2-value hash operation is triggered to hash $G$ bits in the $1^{st}$ segment of its IP address to 1 or 0 (Line 7). Then the subnet bitmap is halved into two parts and only one half part is selected according to the hash result (Lines 8 $\sim$ 9, 14). This operation will be recursively conducted from the $1^{st}$ to the $(V-1)^{th}$ segment of IP address (Lines 6, 13, 18 $\sim$ 19). For all packets of a super host, if one segment of IP addresses across all packets is same, the hash results of this segment will also be the same, leading to selection of the same half in the current subnet bitmap (Line 20). Otherwise, the hash results of this segment across all packets are different, and two half parts of the current subnet bitmap will be selected (Lines 10 $\sim$ 11, 15 $\sim$ 16). Then, this binary search recursive process will stop (Lines 12, 17).
    
\vspace{-2mm}
 \begin{figure}[htbp]
    \centering
    \includegraphics[width=0.82\linewidth]{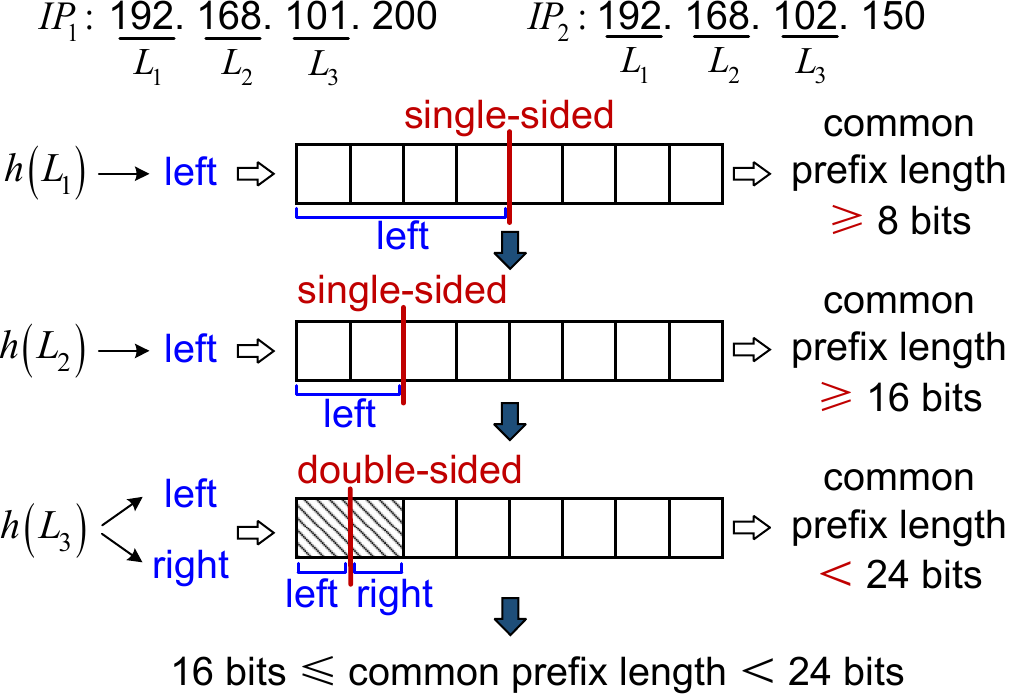}
   \vspace{-1mm}
    \caption{Example of common prefix length estimation using halved-segment hashing with IP segment width $G = 8$ bits.}
    \label{fig:8bit}
\end{figure}
 \vspace{-2mm}

Figure~\ref{fig:8bit} illustrates an example of the halved-segment hashing strategy. We assume the 32-bit destination IP is segmented into four 8-bit segments $L_1\sim L_4$. If the 2-value hash results of the first segment $L_1$ across all packets are same, then only one half of the bitmap will be hashed (i.e., left side). For the second segment $L_2$, if the hash results of all packets are still same, then the hashed part of bitmap is halved again (i.e., left side again). For the third segment $L_3$, however, if the hash results of all packets are different, then both halves of current part %hashed part of bitmap 
will be hashed (i.e., both sides). This result indicates that the third segment $L_3$ is different across all packets. Thus, the common prefix length falls in the range [16, 24) bits.

\textbf{Subnet cardinality estimation}. After the estimation of the common prefix length of IP address, the subnet address length of packets to or from the candidate super host is obtained. With the subnet address length, we then estimate the number of distinct IPs (i.e., subnet cardinality) within the subnet to identify super host using the Linear Counting estimator.
\setcounter{figure}{5}
\vspace{-7mm}
\begin{figure}[htbp]
    \centering
    \includegraphics[width=0.93\linewidth]{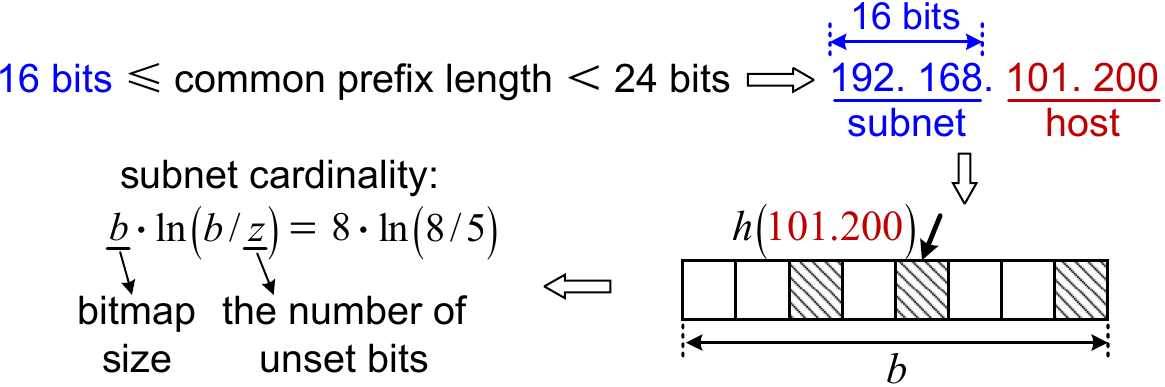}
   \vspace{-1mm}
    \caption{Example of subnet cardinality estimation.}
    \vspace{-3mm}
    \label{fig:BCC}
\end{figure}
%\vspace{-1mm}

As shown in Figure~\ref{fig:BCC}, the inferred subnet address length of packets falls in [16, 24). To ensure that the inferred host address covers the actual one, we take the low bound 16 bits of the range as the subnet address length and the remaining 16 bits are considered as the host address. Subsequently, we hash the host address into the host bitmap and use Linear Counting to estimate the flow cardinality as \( b \cdot \ln \left( {b}/{z} \right) \), where $b$ denotes the size of the host bitmap, and $z$ represents the number of unset bits in the host bitmap. Finally, since we estimate the cardinality according to the host address within the same subnet, rather than the whole IP address, we can effectively distinguish malicious hosts from benign ones, which may also have very large flow cardinality covering a wide range of subnets.

\subsection{Basic Operations}\label{AA}

SegSketch supports two basic operations: (i) the update operation selectively inserts each incoming packet into the sketch; (ii)
the query operation estimates subnet cardinality of an input host key.

\textbf{Update}. For each arrival packet with a source–destination pair $(x, y)$, SegSketch uses $x$ and $y$ as the host key for super spreader detection and super receiver detection, respectively. This section uses super spreader detection as an example to describe how SegSketch works. The host key is mapped using $r$ independent hash functions $h(\cdot)$ to select an available bucket, denoted as $\mathit{B(i, h_i(x))}$ for $i = 1, 2, \dots, r$. 

\begin{figure*}[htbp]
    \centering
    % Case1-1
    \subfigure[$x$ already exists in a hashed bucket.]{
        \includegraphics[width=0.31\linewidth]{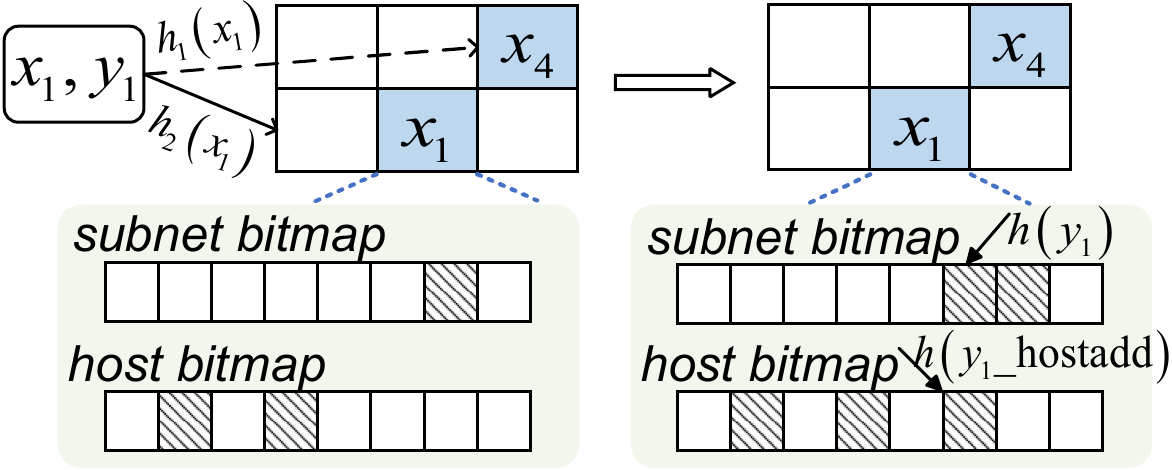}
        \vspace{-2mm}
        \label{fig:a}
    }
    \subfigure[$x$ is not stored and at least one hashed bucket is empty.]{
        \includegraphics[width=0.31\linewidth]{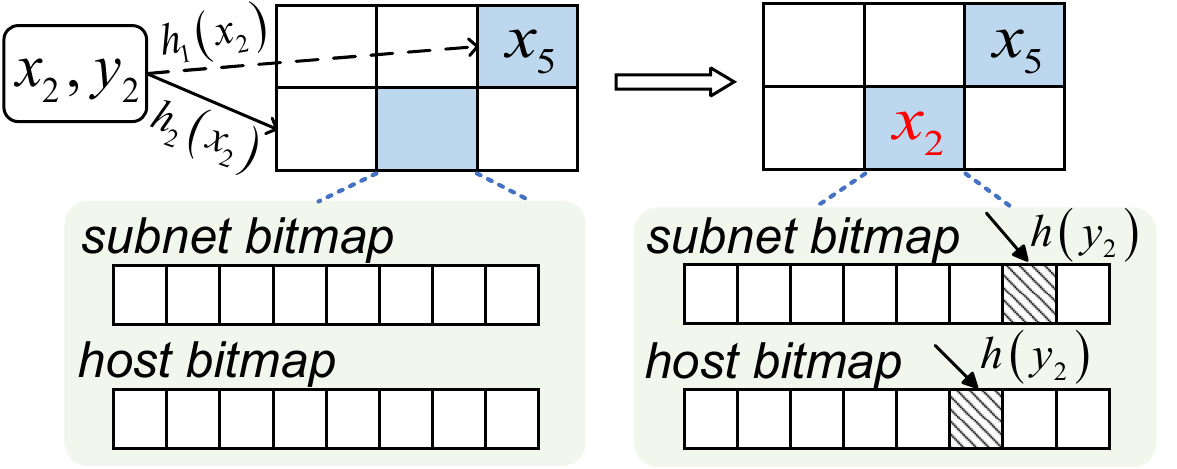}
        \vspace{-2mm}
        \label{fig:b}
    }        
    \subfigure[$x$ is not stored and all hashed buckets are not empty.]{
        \includegraphics[width=0.31\linewidth]{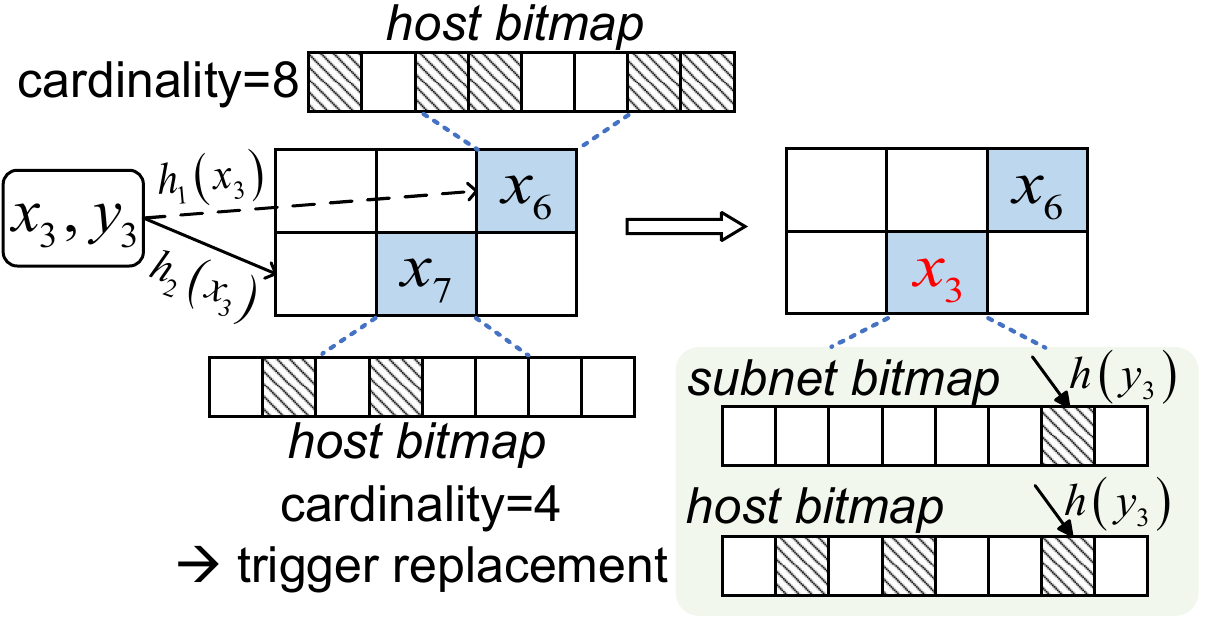}
        \vspace{-2mm}
        \label{fig:c}
    }
    \vspace{-3mm}
    \caption{Examples of the update procedure for super spreader detection.}
    \vspace{-2mm}
    \label{fig:abc}
\end{figure*}

\textbf{Case 1: $x$ already exists in a bucket}. If the host key stored in $B(i, h_i(x))$ is $x$, then the subnet bitmap is updated by adopting the halved-segment hashing strategy to $y$, and the host bitmap is updated through hashing the host address of $y$.

As shown in Figure~\ref{fig:abc}(a), upon the arrival of packet $(x_1, y_1)$, a bucket with host key $x_1$ is found. The subnet bitmap is updated by applying the halved-segment hashing to $y_1$, and the inferred common prefix length is used to extract the host address of $y_1$, which is then hashed to update the host bitmap.

\textbf{Case 2: $x$ is not found and at least one hashed bucket is empty}. If $x$ is not found in any of the $r$ hashed buckets and at least one of them is empty, it is inserted into the first empty bucket. The sketch sets the host key to $x$, and initializes both the subnet bitmap and host bitmap. The subnet bitmap is initialized via halved-segment hashing on $y$, while the host bitmap is updated by hashing $y$, as the subnet prefix cannot yet be determined on the first arrival.

As shown in Figure~\ref{fig:abc}(b), when a packet $(x_2, y_2)$ arrives, none of the hashed buckets contains host $x_2$, and then it is inserted into the first empty bucket among the hashed buckets. The host key is set to $x_2$, the subnet bitmap is initialized via the halved-segment hashing of $y_2$, and a bit in the host bitmap is set by hashing the full $y_2$.

\textbf{Case 3: $x$ is not stored and all hashed buckets are not empty}. If $x$ does not appear in any of $r$ hashed buckets and all hashed buckets are not empty, the sketch performs a probability-based replacement strategy on the bucket with the smallest subnet cardinality among hashed buckets. The replacement probability is defined as
\begin{equation}
\Pr = \frac{1}{b \cdot \ln \left( b/z \right) + 1}, 
\end{equation} 
where $b \cdot \ln(b/z)$ estimates the subnet cardinality of the host stored in the bucket. This design ensures that a host already stored in the bucket with higher subnet cardinality is less likely to be replaced, while a newly arriving host $x$ that sends packets more frequently has a greater chance of replacing it through repeated replacement attempts.
If $x$ successfully triggers the replacement operation, the corresponding bucket is updated as follows: (i) the host key field is set to $x$; (ii) the subnet bitmap is cleared to eliminate residual prefix information from the previous host that may interfere with the subnet inference of the new host, and is reinitialized via the halved-segment hashing on $y$; and (iii) the entire $y$ is hashed to update the host bitmap. Otherwise, $x$ is discarded.

%As shown in Figure~\ref{fig:abc}(c), a packet $(x_3, y_3)$ arrives and attempts to replace the host of the bucket containing the smallest subnet cardinality based on the replacement probability. Suppose the replacement condition is met, the bucket's host key is updated to $x_3$, the subnet bitmap is cleared and reinitialized via halved-segment hashing on $y_3$, and the host bitmap is updated by hashing $y_3$.

As shown in Figure~\ref{fig:abc}(c), upon the arrival of packet $(x_3, y_3)$, $x_3$ attempts to replace the host  of the bucket that contains the smallest subnet cardinality based on the replacement probability. Suppose the replacement condition is met, the bucket's host key is updated to $x_3$, the subnet bitmap is cleared and initialized via halved-segment hashing on $y_3$, and the host bitmap is updated by hashing $y_3$.
 
%As shown in Figure~\ref{fig:abc}(c), upon the arrival of packet $(x_3, y_3)$, $x_3$ attempts to replace the host in one of the hashed buckets based on the replacement probability. Suppose the replacement condition is met for $x_7$, the bucket's host key is updated to $x_3$, the subnet bitmap is cleared and reinitialized via halved-segment hashing on $y_3$, and the host bitmap is updated by hashing the entire $y_3$.

\textbf{Query}. The query operation is invoked
for a given input host key $x$, and the corresponding estimated subnet cardinality will be returned. The subnet cardinality is estimated through the Linear Counting algorithm.

\subsection{Identification of Super Hosts}\label{AA}
At the end of each measurement epoch, all buckets of SegSketch are checked to recover super hosts. 
A recorded host is reported if its estimated subnet cardinality exceeds the detection threshold $T(p)$, where $p$ is the inferred prefix length of the host's connected subnet. 
Specifically, the threshold is defined as:  
\begin{equation}
T(p) = \theta \cdot 2^{32 - p},
\end{equation}
where $\theta\in(0,1]$ is a scaling factor. 
This design ensures that the threshold increases with the number of IPs increasing in the associated subnet,
reflecting the fact that a host connected to a larger subnet will contact more distinct IPs within it. Our method elastically estimates subnet cardinality for subnet-based attacks and degrades to per-host cardinality estimation for distributed attacks with no common or short prefixes.

%will contact more distinct IPs within it.

%This design ensures that the threshold increases with subnet size, reflecting that a host must contact more distinct IPs within a larger subnet to be identified as a super host.

\section{Theoretical Analysis}
\label{sec:theoretical analysis}

\makeatletter
\thm@headfont{\sc}
\makeatother

\begin{theorem}
%For any host $x$ connecting to a targeted subnet, the error between its true subnet cardinality $C(x)$ and the estimated subnet cardinality $\hat{C}(x)$ is bounded.
%obtained via Linear Counting is bounded. Specifically, for a expected error $\varepsilon$ of the estimated cardinality introduced by the Linear Counting algorithm, we 
%For any host $x$ connecting to a targeted subnet, the probability that the error between its true and estimated subnet cardinality $C(x)$ and $\hat{C}(x)$ exceeds the expected error $\varepsilon$ introduced by the Linear Counting algorithm is bounded as follows:
For any host $x$ connecting to a targeted subnet, the error between its true and estimated subnet cardinality $C(x)$ and $\hat{C}(x)$ may exceed the expected error $\varepsilon$ introduced by the Linear Counting algorithm. The probability of this event is bounded as follows:
\begin{align}
\Pr \left\{ |\hat{C}(x) - C(x)| \geq \varepsilon \right\}
\;\le\;
\varepsilon^{-1}\,
M\left[\,1 - \bigl(1 - M^{-1}\bigr)^{U}\,\right],
\label{eq:unified-bound}
\end{align}
%where $M$ is the effective bitmap domain size determined by the hashing strategy, and $U$ is the number of flows that are not in the targeted subnet but are misclassified into it.
where $\varepsilon$, $M$ and $U$ are three hash-strategy-dependent variables that satisfy Table~\ref{tab:math}. In Table~\ref{tab:math}, $N(x)$ denotes the flow cardinality of host $x$, $l$ represents the true subnet prefix length, and $G$ is the IP segment width under the halved-segment hashing strategy.
\end{theorem}

\vspace{-2mm}
\renewcommand{\arraystretch}{0.9} 
\begin{table}[htbp]
\centering
\caption{Variable values under different hashing strategies.}
\vspace{-2mm}
\label{tab:math}
\begin{tblr}{
  width = \linewidth, 
  colspec = {
    X[0.7, l]  % 第一列（变量名）压缩
    X[1.1, l]  % 第二列（Full-address）
    X[1.1, l]  % 第三列（Host-address）
  },
  cell{1}{1} = {r=2}{},   % 第一列跨两行
  cell{1}{2} = {c=2}{},   % 第二行前两列合并
  cells = {c},            % 单元格内容居中
  rowsep = 0.9pt,
  colsep = 2pt,           % 减小列与列之间x轴间距（默认6pt）
  hline{1,3,6} = {-}{},
  hline{2} = {2-3}{},
}
\textbf{Variables} & \textbf{Hashing Strategies} & \\
                   & \textbf{Full-address}       & \textbf{Host-address} \\
$\varepsilon$      & $N(x)-M\left[1 - (1 - M^{-1})^{N(x)}\right]$ 
                   & $C(x)+U-M\left[1 - (1 - M^{-1})^{C(x)+U}\right]$ \\
$M$                & $2^{32}$                    
                   & $2^{32 - \lfloor l/G \rfloor  G}$ \\
$U$                & $N(x)-C(x)$                 
                   & $(N(x)-C(x))\left(\tfrac{1}{2}\right)^{\lfloor l/G \rfloor}$
\end{tblr}
\end{table}

\iffalse
\begin{table}
\centering
\caption{Variable values under different hashing strategies}aption{Variable values under different hashing strategies}
\vspace{-2mm}
\begin{tblr}{
  width = \linewidth, 
  colspec = {
    Q[c] 
    Q[l] 
    Q[l]
  },
  cells = {c},
  rowsep = 0.9pt,
  cell{1}{1} = {r=2}{},
  cell{1}{2} = {c=2}{},
  hline{1,3,6} = {-}{},
  hline{2} = {2-3}{},  % 第二条横线扩展到第2列和第3列
}
\textbf{Variables} & \textbf{Hashing Strategies} & \\
                   & \textbf{Full-address}       & \textbf{Host-address} \\
$\varepsilon$      & $N(x)-M\left[\,1 - \bigl(1 - M^{-1}\bigr)^{N(x)}\,\right]$
                   & $C(x)+U-M\left[\,1 - \bigl(1 - M^{-1}\bigr)^{C(x)+U}\,\right]$ \\
$M$                & $2^{32}$
                   & $2^{32 - \lfloor l/G \rfloor \cdot G}$ \\
$U$                & $N(x)-C(x)$
                   & $(N(x)-C(x))\left(\tfrac{1}{2}\right)^{\lfloor l/G \rfloor}$
\label{tab:math}
\end{tblr}
\end{table}
\fi

\vspace{-4mm}
\begin{proof}
Please refer to Appendix A.1 for details. 
\end{proof}
\vspace{-4mm}

\begin{theorem}
The expected estimation error of subnet cardinality through hashing the full IP addresses is larger than that obtained by hashing only the host addresses, that is,
\begin{align}
\mathbb{E}\!\left[ |\hat{C}_{{full}}(x) - C(x)| \right] >\mathbb{E}\!\left[ |\hat{C}_{{host}}(x) - C(x)|\right],
\end{align}
where $\hat{C}_{{full}}(x)$ and $\hat{C}_{{host}}(x)$ are the estimated subnet cardinality through hashing full and host addresses, respectively, and ${C}(x)$ is the true subnet cardinality.% $Q = N(x) - C(x)$ represents the number of distinct flows out of the targeted subnet, and $r = \left\lfloor \frac{l}{G} \right\rfloor \cdot G<32$ denotes the estimated prefix length with $l$ and $G$ representing the true subnet length and IP segment width. %Moreover, the smaller the segment width $G$, the larger the error gap between the two hashing strategies.
\end{theorem}

\vspace{-2mm}
\vspace{-2mm}
\begin{proof}
Please refer to Appendix A.2 for details.
\end{proof}

%Moreover, we plot the trend of $\mathbb{E}\!\left[\hat{C}_{\mathrm{host}}(x) - C(x)\right]$ under different prefix lengths $l$ and $\mathbb{E}\!\left[\hat{C}_{\mathrm{full}}(x) - C(x)\right]$ with respect to segment width $G$, as illustrated in Figure~\ref{fig:math}. The monotonic decrease of $\mathbb{E}\!\left[\hat{C}_{\mathrm{host}}(x) - C(x)\right]$ with smaller segment widths indicates that host-address hashing achieves lower estimation error and increasing advantage over full-address hashing as $G$ becomes smaller.
%It can be observed that $F(G)$ decreases monotonically with increasing $G$, which confirms that smaller segment widths lead to larger error gaps between full-address and host-address hashing.
  
%\vspace{-5mm}

%The monotonic decrease of $\mathbb{E}[\hat{C}_{{host}}(x) - C(x)]$ as $G$ decreases shows that estimating subnet cardinality though host-address hashing achieves lower estimation error and increasing advantage over full-address hashing with finer segmentation.

% \vspace{-7mm}
\section{Trace-Driven Evaluation}
\label{sec:evaluation}
\subsection{Experimental Setup}\label{AA}

\textbf{Datasets}. 
%We utilize two combined datasets constructed by mixing real-world traffic traces, where the UNSW-NB15\cite{unsw} dataset is respectively merged with the MAWI2021\cite{mawi} dataset and the CAIDA2016\cite{caida} dataset to conduct super spreader detection\footnote{We take super spreader detection as an example. For super receiver detection, the destination IP is used as the flow ID, and the halved-segment hashing approach is applied to the source IPs to identify common prefixes and estimate, for each destination, the number of distinct sources residing within the same subnet.}. Each packet in the above datasets provides source and destination IP addresses (4 bytes each). For super spreader detection, we set source IP as the flow ID.
We utilize two datasets constructed by merging malicious and benign traffic traces: the UNSW-NB15 \cite{unsw} dataset is mixed with MAWI2021 \cite{mawi} and CAIDA2016 \cite{caida}, respectively, to evaluate the performance of super-host detection. (1) The UNSW-NB15 dataset is generated by UNSW Canberra using the IXIA PerfectStorm tool, containing benign traffic and nine types of attacks, including super spreader attacks (e.g., worm propagation, reconnaissance) and super receiver attacks (e.g., DDoS). (2) The MAWI dataset is anonymized traffic trace captured by MAWI on the trans-Pacific backbone link. (3) The CAIDA dataset consists of anonymized IP traces collected by CAIDA through the Equinix-Chicago monitor. 

%; the SDN-DDoS\cite{sdn} dataset is mixed with MAWI2021\cite{mawi} and CAIDA2016\cite{caida}, respectively, to evaluate super receiver detection.
% While we use super spreader detection as an example, the methods also apply to super receiver detectioin.

%The UNSW-NB15 dataset is created by UNSW Canberra using the IXIA PerfectStorm tool to generate mixed traffic of benign behaviors and nine types of attacks, including super spreader attacks (e.g., worm propagation and reconnaissance) and super receiver attacks (e.g., DDoS).

%\textbf{2) SDN-DDoS dataset:} The SDN-DDoS dataset is generated in a Software-Defined Networking (SDN) environment using two network topologies with 12 switches and 24 hosts, orchestrated by a Ryu controller. It contains over 1 million labeled traffic flows comprising both benign and multiple DDoS attack scenarios. 

%(\romannumeral 2) The MAWI dataset is anonymized traffic trace captured by MAWI on the trans-Pacific backbone link.

%(\romannumeral 3) The CAIDA dataset consists of anonymized IP traces collected by CAIDA through the Equinix-Chicago monitor. 

Since the MAWI and CAIDA datasets do not contain labeled super host attacks, we use the super host attacks from UNSW-NB15 as ground truth, just as same as previous works \cite{dollm, gmcb, robust}. %that construct datasets by combining attack flows with benign traces. 
Using UNSW-NB15 alone limits evaluation to a single benign background, while MAWI and CAIDA offer diverse traffic characteristics, evaluating our method's adaptability to different distributions.
%Experiments using UNSW-NB15 alone are shown in Appendix B.1.}

%\textcolor{blue}{We combine MAWI and CAIDA datasets with UNSW-NB15 separately, using the latter's super host attacks as ground truth, as done in previous works \cite{dollm, gmcb, robust}. This mixing introduces diverse background traffic characteristics, allowing evaluating our method's adaptability to varying attack-to-background ratios and traffic distributions.}

\textbf{Implementation}. SegSketch is compared with Couper \cite{couper}, SpreadSketch \cite{spreadsketch}, and RHHH \cite{sigcomm17}. 
All of these schemes are implemented in C++. 
In the experiments, we focus on super spreader detection by default, where the source IP is used as the host key (4 bytes). For super receiver detection, the destination IP will be used as the host key (4 bytes). %and a halved-segment hashing approach is applied to the source IPs to identify common prefixes and estimate, for each destination, the number of distinct sources within the same subnet. 
The default number of rows is 3 in SegSketch. The threshold scaling factor is set as $\theta=0.5$ \cite{characteristics}, with the experiments testing the impact of different $\theta$ on detection performance shown in Appendix B.1.
The hash functions of SegSketch are implemented using Bob Hash \cite{bob} with different random seeds. 
For other schemes, we use their recommended parameter settings.

\textbf{Methodology}. (1) Couper \cite{couper} proposes a two-layer design. %The first layer uses bitmaps for estimating cardinality of small flows, while the second layer applies Multi-Resolution Bitmaps or HyperLogLog estimators to track large flows.
The first layer uses bitmaps to estimate the number of distinct IPs contacted by low-cardinality hosts, while the second layer tracks high-cardinality hosts using Multi-Resolution Bitmaps or HyperLogLog estimators. Couper detects super hosts by inserting hosts into a candidate table whenever their estimated cardinalities increase. At the end of each measurement epoch, it reports those whose estimated cardinalities exceed a predefined threshold. %However, it does not check whether the connections of a host share the same subnet, and may misidentify benign hosts with many distinct connections as super hosts.

(2) SpreadSketch \cite{spreadsketch} stores candidate host key, Multi-Resolution Bitmap for cardinality estimation, and the maximum number of leading zero bits in the IP hash string for replacement decision. %Upon packet arrival, it updates the corresponding Multi-Resolution Bitmap or compares the leading zero count of the incoming packet’s hash string with the maximum recorded in the bucket to decide whether to replace. 
After each epoch, it traverses all buckets and reports any host whose estimated cardinality exceeds the threshold as a super host. %However, it does not distinguish whether the flows of a host are within the same subnet, which can lead to false positives in super-host detection when benign hosts connect to many different IP addresses.

(3) RHHH \cite{sigcomm17} maintains a multi-level prefix hierarchy, where each layer tracks flows with a specific prefix length, and uses Linear Counting to estimate cardinality within the corresponding subnet. 

\textbf{Platform}. The experiments are conducted on a Mac laptop with an Intel Core i7 CPU running at 2.30GHz, providing four physical cores and supporting simultaneous multithreading. It is equipped with 32GB of RAM. Each CPU core incorporates a multi-level cache hierarchy that includes a 64KB L1 cache (divided into separate instruction and data caches), a private 256KB L2 cache, and a shared 6MB L3 cache. Ubuntu 18.04.3 is used as the operating system. %for running the experiments.

\textbf{Metrics}. 
%We define a flow as a Carpet Bombing DDoS attack source (or traditional DDoS attack victim destination) if it is ultimately retained in the sketch with spatial skewness shown in the bitmap contributed by segment-wise binary hashing technique, and its cardinality exceeds a threshold $\theta D$ ($\theta S$), where $D$ (or $S$) is the total number of distinct destination (or source) IPs in the monitoring epoch, and $\theta$ is a tunable ratio ($0<\theta<1$). By default, $\theta$ is initialized to ***. 
We evaluate whether SegSketch achieves a favorable balance between accuracy and speed. To assess measurement accuracy, we employ Precision, Recall, F1-Score, and Average Relative Error (ARE). Processing speed is quantified by throughput. 
\begin{itemize}[leftmargin=2em,labelsep=0.5em,topsep=-0.5ex,itemsep=0ex,parsep=0pt]
  %\item {False positive rate}: the proportion of benign flows incorrectly identified as super hosts.
  \item {Precision}: the fraction of the number of true super hosts reported over the number of all super hosts reported.
  \item {Recall}: the fraction of the number of true super hosts reported over the number of real super hosts.
  \item {F1-Score}: \( \frac{2 \times \textit{Precision} \times \textit{Recall}}{\textit{Precision} + \textit{Recall}} \), the harmonic mean of Precision and Recall.
  \item {Average Relative Error (ARE)}: \( \frac{1}{\Psi} \sum_{x_i \in \Psi} \frac{|C(x_i) - \hat{C}(x_i)|}{C(x_i)} \), where \( \Psi \) is the set of reported super hosts, $ C(x_i)$ and $\hat{C}(x_i)$ are the real and estimated subnet cardinalities of host $x_i$, respectively. %\( C(f_i) \) is the real subnet cardinality of super host \( f_i \), and \( \hat{C}(f_i) \) is the estimated subnet cardinality of super host \( f_i \).
  \item {Throughput}: the ratio of the number of inserted packets to the entire monitoring time. Million packets inserted per second (Mpps) is used as the unit. 
\end{itemize}
 \vspace{-1mm}

\subsection{Accuracy of Super Host Detection }\label{AA}
\textbf{Super spreader detection}. 
%To evaluate the performance of different algorithms in detecting super spreaders, we conduct experiments on the aforementioned mixed MAWI and CAIDA dataset (each combined with the UNSW-NB15 dataset), with results presented in Figure~\ref{fig:mawi} and Figure~\ref{fig:caidaa}. The results show that SegSketch outperforms SpreadSketch, Couper, and RHHH in most scenarios. As shown in Figure~\ref{fig:mawi}, with the minimum memory of 32KB, SegSketch outperforms SpreadSketch, Couper, and RHHH in terms of Precision by 353.15\%, 444.72\% and 545.68\%, Recall by 160.21\%, 244.82\% and 100.40\%, and F1-Score by 245.13\%, 257.79\% and 529.03\%, respectively, while simultaneously reducing ARE by 88.57\%, 18.18\%, and 65.38\% compared to these methods. 
We evaluate super spreader detection on the mixed MAWI and CAIDA datasets, each combined with UNSW-NB15. The results on the mixed MAWI dataset are provided in Appendix B.2. 
As shown in Figure~\ref{fig:caidaa}, SegSketch consistently outperforms SpreadSketch, Couper, and RHHH. Under 32KB memory, it achieves Precision improvements of $4.63\times$, $4.42\times$, and $21.50\times$, Recall gains of $2.01\times$, $1.31\times$, and $2.84\times$, and F1-Score increases of $2.73\times$, $2.18\times$, and $8.04\times$, while reducing ARE by 86.08\%, 68.00\%, and 87.20\%, compared to SpreadSketch, Couper, and RHHH, respectively.
The performance gains of SegSketch arise from its integration of common prefix inference and subnet cardinality estimation. By adopting a lightweight halved-segment hashing strategy, SegSketch avoids explicit bitwise recording and comparison on IPs, significantly reducing memory overhead and per-packet processing cost.

%The performance improvements mainly stem from SegSketch's ability to identify super spreaders by combining the identification of the common prefix length of flows and the estimation of subnet cardinality. Moreover, leveraging the lightweight and time-efficient halved-segment hashing strategy for common prefix length estimation, SegSketch avoids explicitly recording and matching each bit of IP addresses, thus reducing memory usage and per-packet comparison overhead.

%The inferior performance of SpreadSketch and Couper primarily arises from their inability to distinguish whether the connections of each host belong to the same subnet. Specifically, they treat all distinct IPs equally, regardless of whether they belong to the same subnet. As a result, benign hosts communicating with diverse peers are likely to be falsely identified as super hosts, leading to high false positive rates and reduced precision. Additionally, since they estimate the flow cardinality rather than the subnet cardinality, these approaches incur higher ARE.

\vspace{-3mm}
\begin{figure}[htbp]
    \centering
    % Case1-1
    \subfigure[Precision]{
        \includegraphics[width=0.465\linewidth]{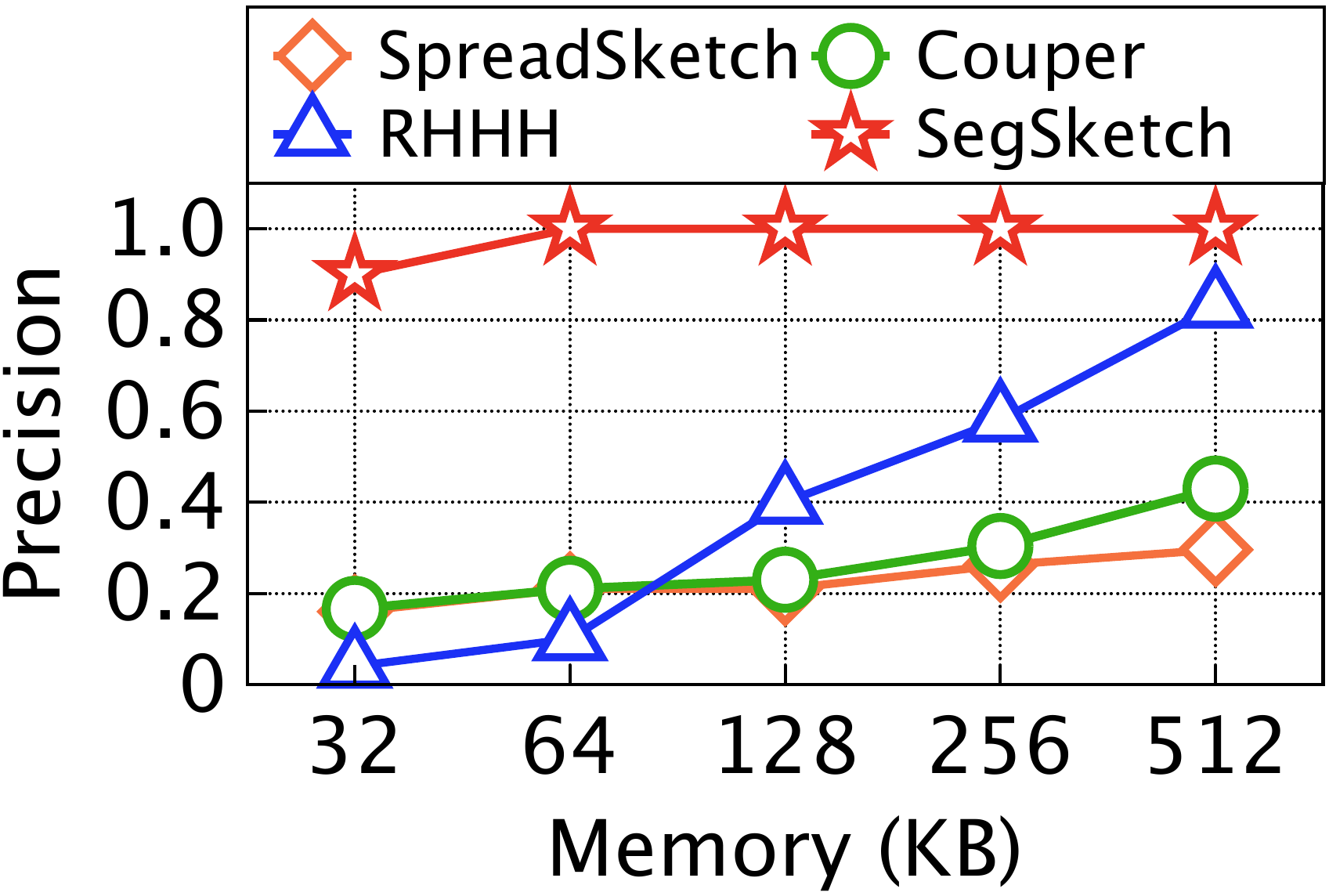}
     \label{fig:a}
    }
    \subfigure[Recall]{
        \includegraphics[width=0.465\linewidth]{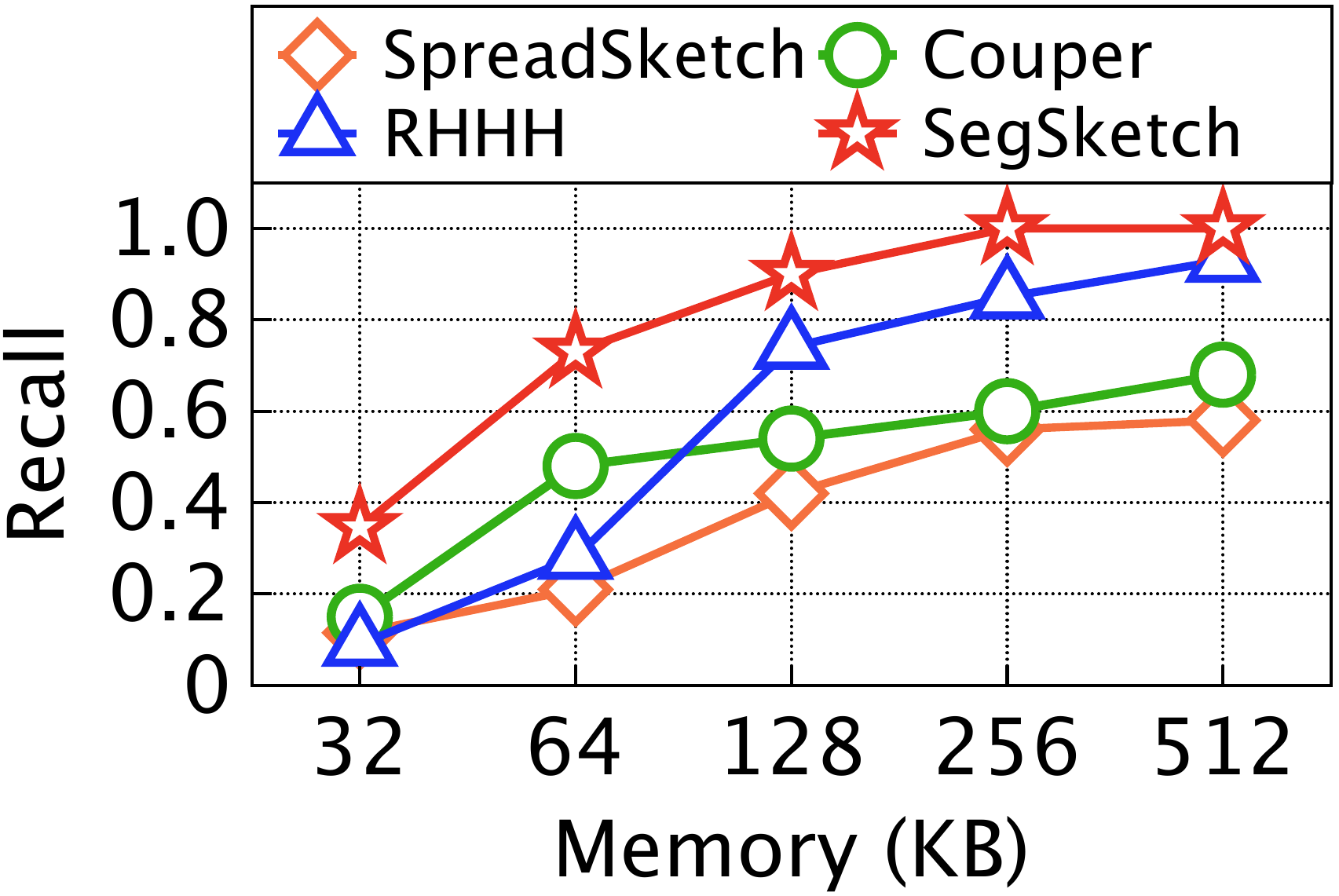}
     \label{fig:b}
    }
    \\
    % Case2
    \subfigure[F1-Score]{
        \includegraphics[width=0.465\linewidth]{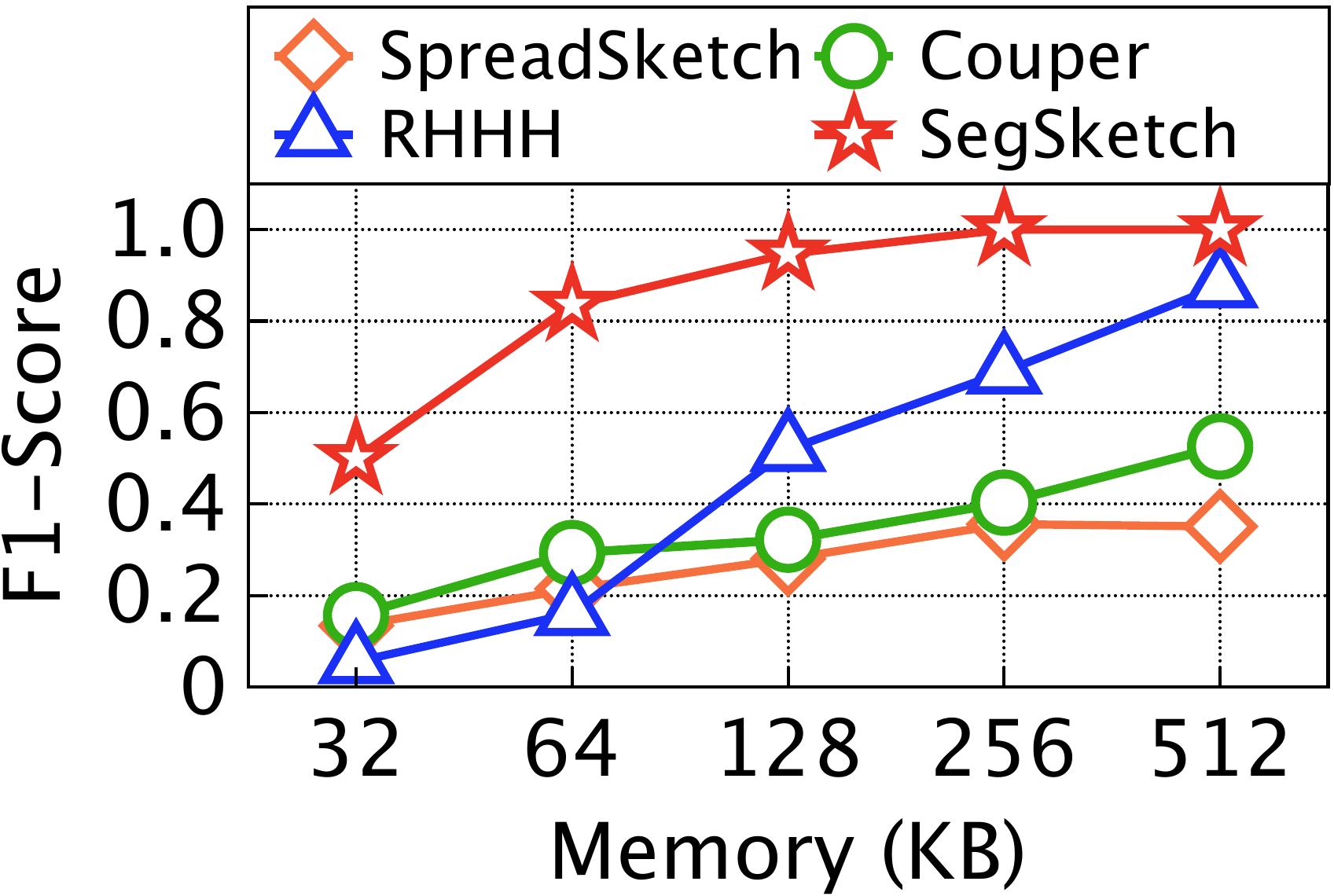}
    \label{fig:c}

    }
    % Case3
    \subfigure[ARE]{
        \includegraphics[width=0.465\linewidth]{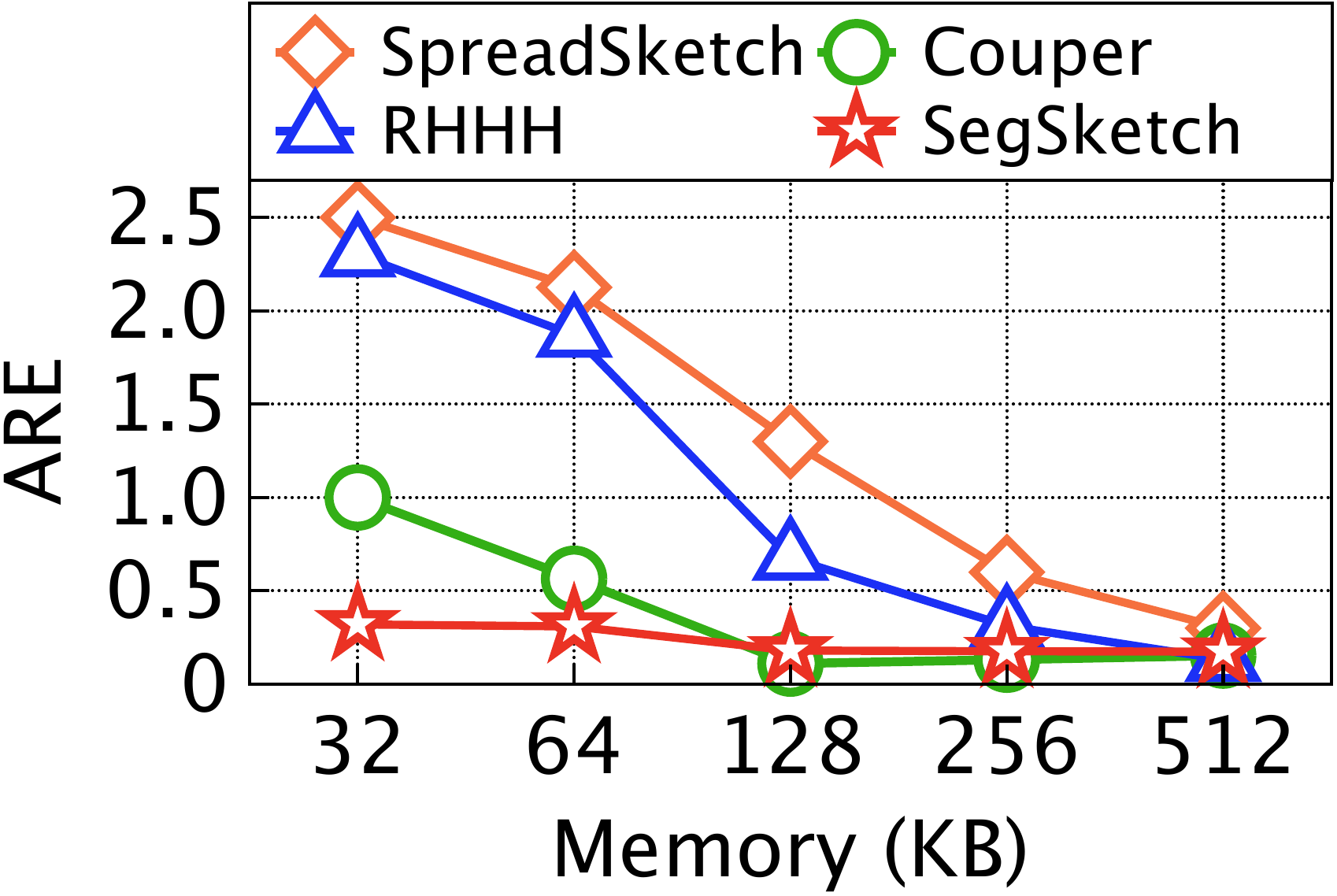}
    \label{fig:d}
    }
  \vspace{-3mm}
    \caption{Super spreader detection performance on the dataset mixing UNSW-NB15 with CAIDA2016.}
   \vspace{-1mm}
    \label{fig:caidaa}
\end{figure}

The inferior performance of SpreadSketch and Couper stems from their inability to distinguish whether a host’s connections reside within the same subnet. By treating all destination IPs equally, they often misclassify benign hosts communicating with a large number of diverse peers as super hosts, resulting in high false positives and reduced precision. Moreover, their estimation of flow cardinality, rather than subnet cardinality, leads to higher ARE.

%In contrast, RHHH identifies hosts with common prefixes using multiple prefix-level estimators. However, due to the unknown subnet lengths in advance, it has to set up multiple estimators for each host to track all possible prefix lengths simultaneously, which incurs significant memory overhead, hurting its accuracy in memory-constrained scenarios. Furthermore, under constrained memory, RHHH often faces a memory allocation dilemma. Specifically, increasing the number of estimators for different prefix lengths within a fixed total memory limits the size of each estimator, leading to more frequent hash collisions within each estimator and inaccurate subnet cardinality estimation. Alternatively, reducing the number of estimators to save memory results in coarse common prefix detection. As a result, choosing a shorter common prefix than the actual subnet aggregates unrelated addresses and overestimates the subnet cardinality, while the opposite incurs longer common prefixes, leading to underestimation.

RHHH identifies hosts with common prefixes using multiple prefix-level estimators. However, lacking prior knowledge of subnet lengths, it must maintain estimators for all possible prefix lengths per host, resulting in substantial memory overhead and reduced accuracy under limited memory. Increasing the number of estimators shrinks their individual sizes, causing frequent hash collisions within each estimator and inaccurate cardinality estimation, while reducing the number of estimators impairs prefix resolution. Choosing a shorter prefix aggregates unrelated addresses and overestimates cardinality, while the opposite leads to underestimation.

%In contrast, RHHH is capable of identifying flows with common prefixes by maintaining multiple prefix-level estimators. However, because subnet lengths are unknown in advance, it must track all predefined prefix lengths simultaneously, which under a fixed memory budget creates an inherent dilemma. On the one hand, allocating memory across many prefix levels limits the number and size of estimators in each layer.
%This leads to frequent hash collisions: between different flows mapping to the same bucket, which contaminates the bitmaps and incurs inaccurate subnet cardinality estimation, and within the same flow, where different connections are hashed to the same bits of the bitmap, resulting in underestimation of subnet cardinality. On the other hand, reducing the number of prefix levels to save memory yields coarse prefix resolution. In this case, actual subnet lengths may be missed entirely: choosing a shorter prefix than the real subnet aggregates unrelated addresses and overestimates the subnet cardinality, while a longer prefix fragments the subnet and undercounts its distinct connections. Consequently, Enhanced RHHH struggles to achieve satisfactory performance under constrained memory.

\textbf{Super receiver detection}. 
%To evaluate the performance of different approaches in detecting super receivers, we conduct experiments on the mixed MAWI and CAIDA dataset, with results shown in Figure~\ref{fig:MAWI} and Figure~\ref{fig:CAIDA}. Unlike super spreader detection, all four algorithms are modified to estimate the cardinality of source IPs. The results show that SegSketch still outperforms all the other methods. Specifically, as shown in Figure~\ref{fig:MAWI}, in the memory scenario of 32 KB, SegSketch exceeds the performance of SpreadSketch, Couper, and RHHH in F1-Score by 897.05\%, 374.12\%, and 707.14\%, respectively. Similar to super spreader detection, the improvements mainly stem from SegSketch’s ability to identify super receivers by combining the examination of the common prefix length of flows and the estimation of subnet cardinality. These experimental results further confirm the effectiveness of SegSketch in the super receiver detection task.
To evaluate the performance of different approaches in detecting super receivers, we conduct experiments on the mixed MAWI and CAIDA datasets. The results on the mixed MAWI are in Appendix B.2. All four algorithms are adapted to estimate the cardinality of source IPs. SegSketch still surpasses the baselines. As shown in Figure~\ref{fig:CAIDA}, with 32KB memory, SegSketch improves F1-Score by $5.08\times$, $4.17\times$, and $2.77\times$ over SpreadSketch, Couper, and RHHH, respectively. Similar to super spreader detection, the gains come from SegSketch’s combination of common prefix inference and subnet cardinality estimation. These results validate the effectiveness of SegSketch in detecting super receivers.

\vspace{-2mm}
\begin{figure}[htbp]
    \centering
    % Case1-1
    \subfigure[Precision]{
        \includegraphics[width=0.465\linewidth]{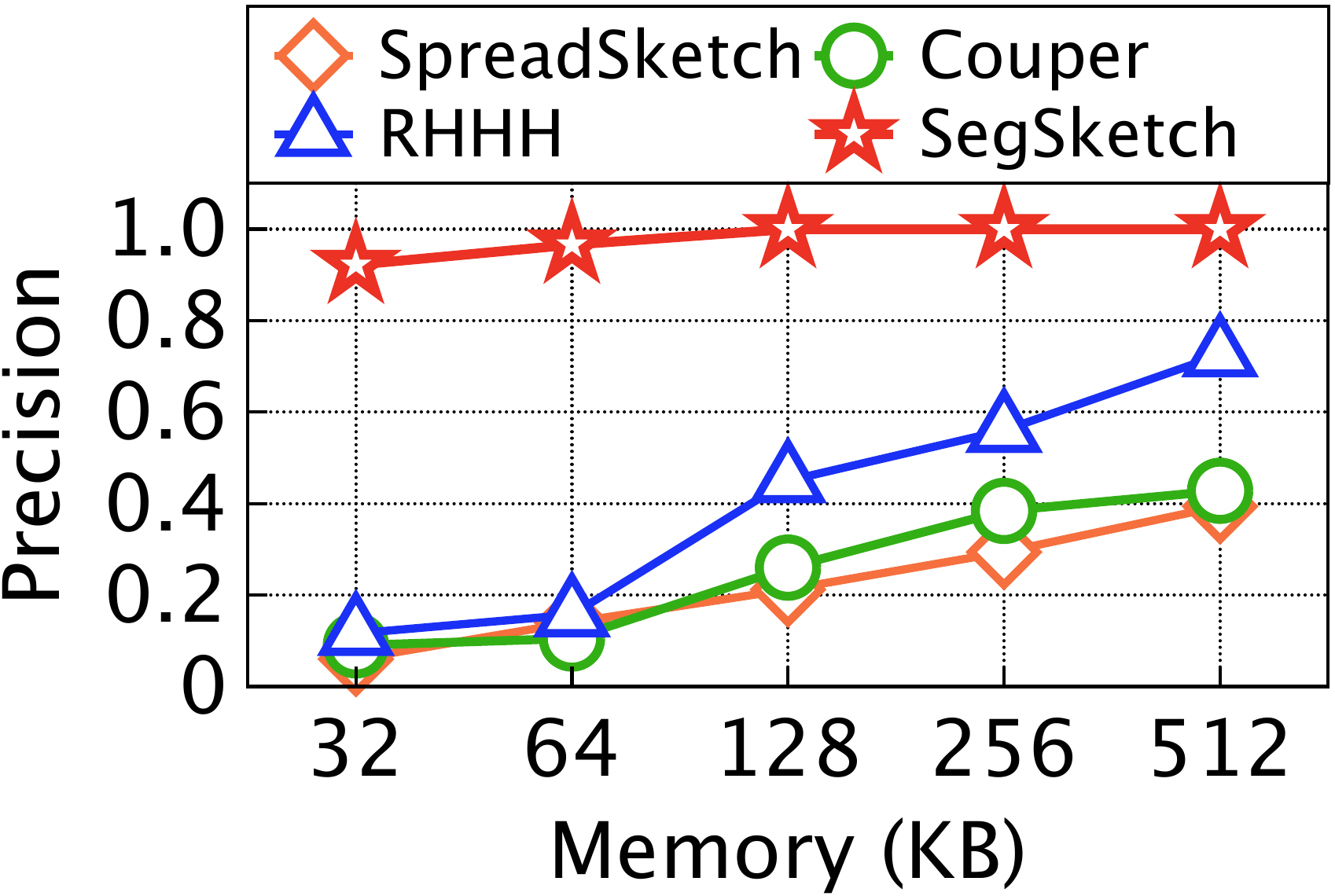}
     \label{fig:a}
    }
    \subfigure[Recall]{
        \includegraphics[width=0.465\linewidth]{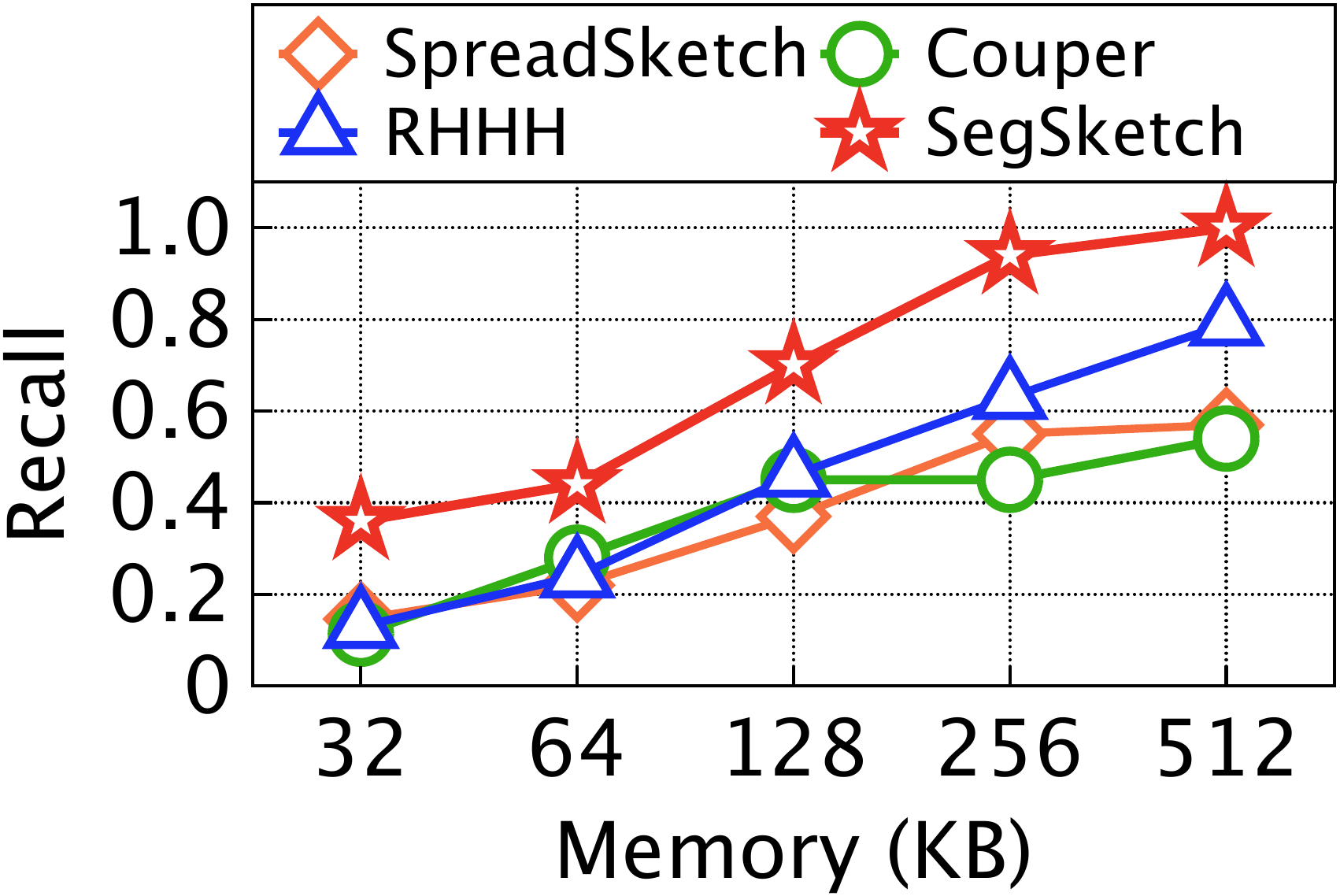}
     \label{fig:b}
    }
    \\
    % Case2
    \subfigure[F1-Score]{
        \includegraphics[width=0.465\linewidth]{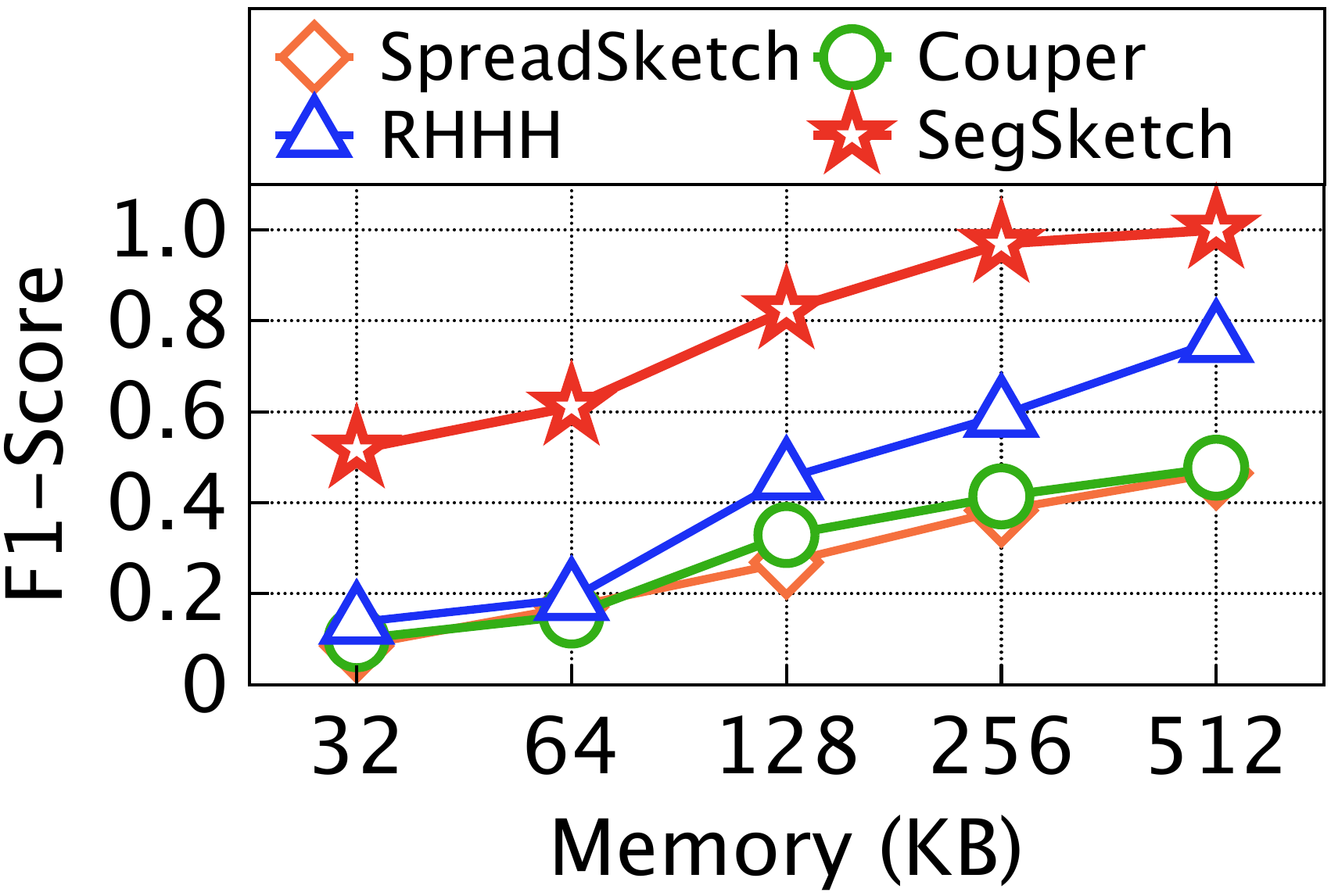}
    \label{fig:c}

    }
    % Case3
    \subfigure[ARE]{
        \includegraphics[width=0.465\linewidth]{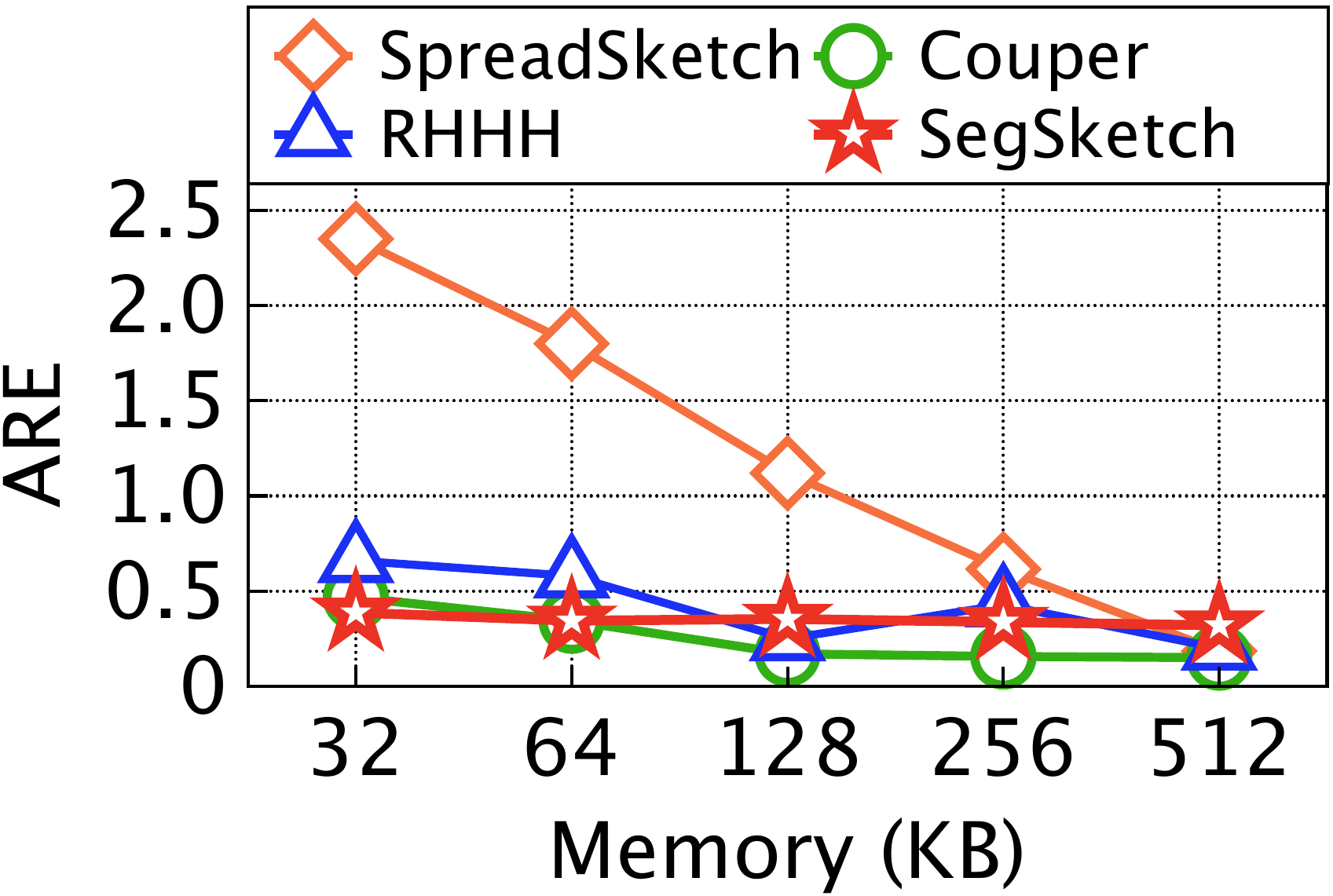}
    \label{fig:d}
    }
    \vspace{-2mm}
    \caption{Super receiver detection performance on the dataset mixing UNSW-NB15 with CAIDA2016.}
   \vspace{-2mm}
    \label{fig:CAIDA}
\end{figure}

\vspace{-2mm}
\subsection{Scalability}\label{AA}

To evaluate robustness under varying super host ratios, we test four methods with 64KB memory on the CAIDA2016 dataset, injected with different ratios of super spreaders from UNSW-NB15. The evaluated ratios of super spreaders to benign hosts are approximately 1:20, 1:25, 1:33, and 1:50. Results are shown in Figure~\ref{fig:Pof}.

%Note that we only present results for super spreader detection, since the distinction between detecting super spreaders and super receivers lies merely in whether the host key is the source or destination IP, and thus the performance is equivalent.

%\vspace{-1mm}
  
 %\setcounter{figure}{12}

%The results show that as the ratio of super spreaders increases, the Precision, Recall, and F1-Score of all methods gradually decline and the ARE increases. This degradation is caused by more super spreaders entering the sketch, which leads to heavier hash collisions among the buckets. As more super spreaders compete for the limited number of available buckets, the likelihood of erroneous evictions in the sketch rises. 
%However, despite the performance degradation, SegSketch still achieves the best results among all these methods, maintaining an F1-Score of 0.79 and an ARE of 0.32 even when the ratio of super spreaders to benign hosts reaches 1:20. This robustness is attributed to SegSketch’s halved-segment hashing strategy in the subnet bitmap and its approach of hashing only the host address in the host bitmap. These strategies significantly enhance both memory efficiency and the accuracy of subnet cardinality estimation. The results confirm that SegSketch can effectively handle high ratios of super hosts, maintaining high detection accuracy even under limited memory conditions.

%\vspace{-1mm}
The results show that as the ratio of super spreaders increases, Precision, Recall, and F1-Score of all methods decline, while ARE rises due to intensified hash collisions and more frequent evictions. However, SegSketch consistently outperforms other schemes, sustaining an F1-Score of 0.79 and an ARE of 0.32 even at 1:20 ratio. This robustness is attributed to halved-segment hashing in the subnet bitmap and host-address hashing in the host bitmap, which together improve memory efficiency and enhance the accuracy of subnet cardinality estimation. These results demonstrate that SegSketch effectively maintains high detection accuracy under constrained memory even in scenarios with a high ratio of super hosts. The 1:33 ratio is used across the other experiments.

 \setcounter{figure}{10}
\vspace{-2mm}
\begin{figure}[htbp]
\vspace{-2mm}
    \centering
 
    % Case2
    \subfigure[F1-Score]{
        \includegraphics[width=0.465\linewidth]{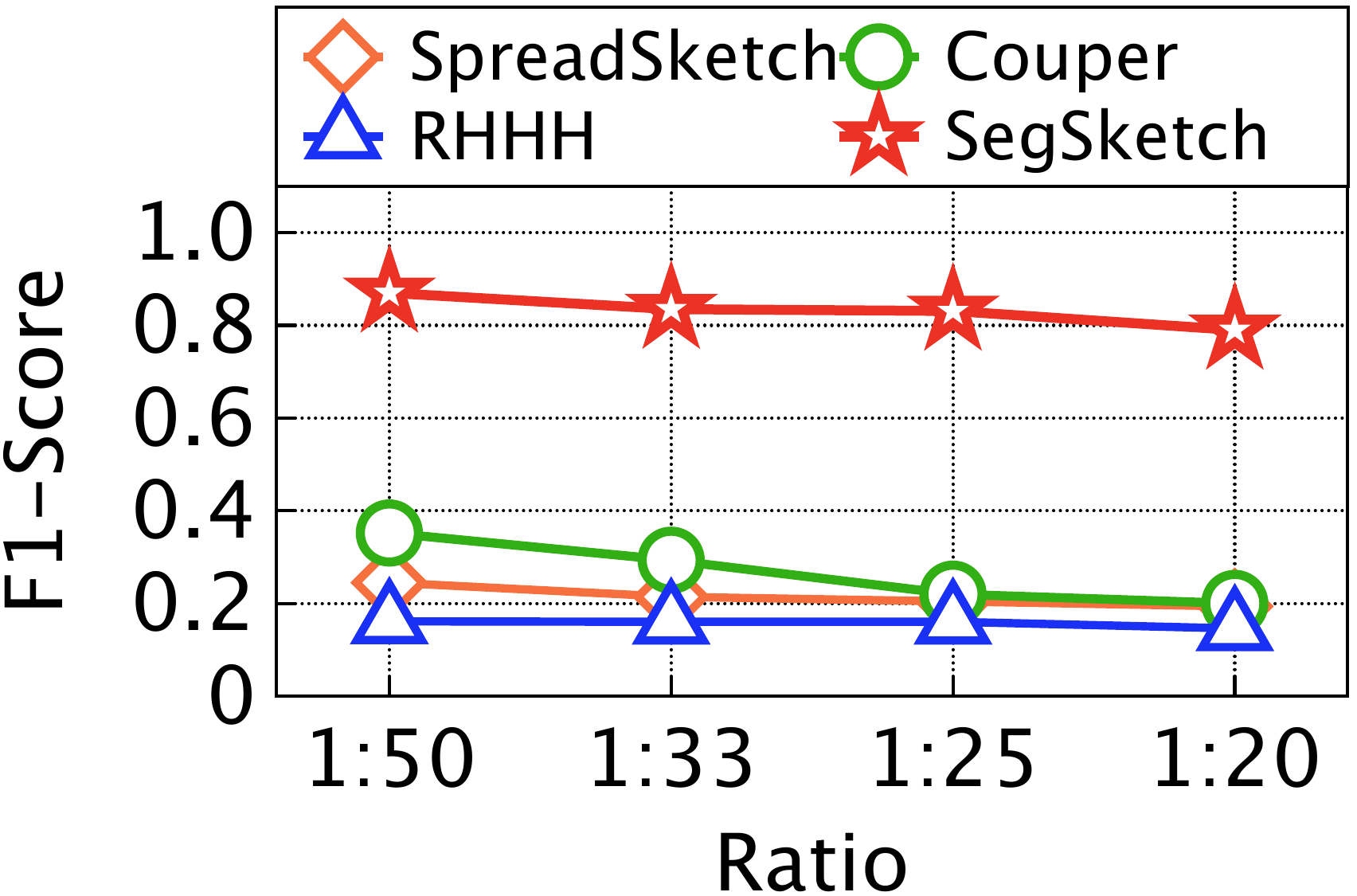}
    \label{fig:c}

    }
    % Case3
    \subfigure[ARE]{
        \includegraphics[width=0.465\linewidth]{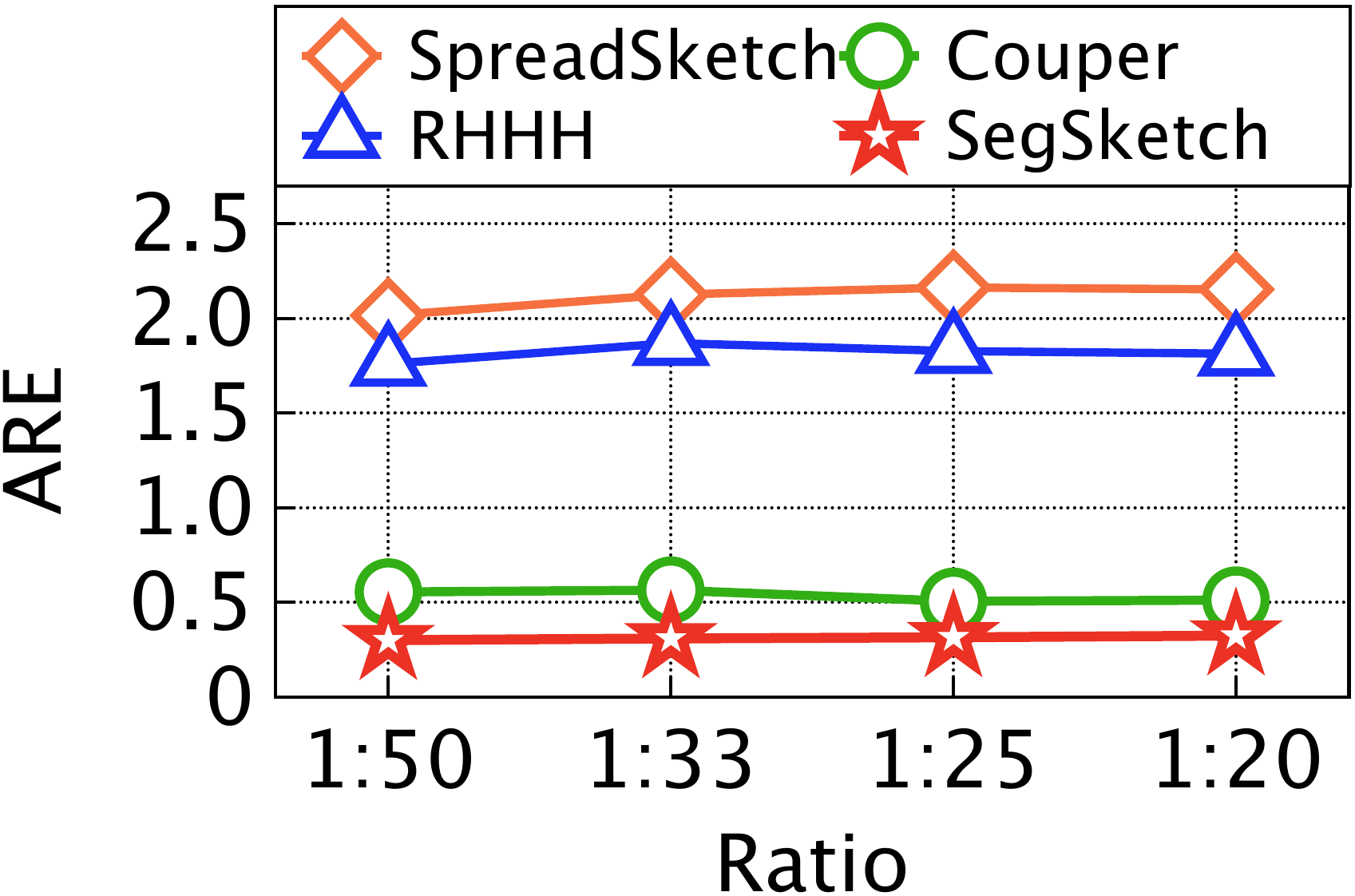}
    \label{fig:d}
    }
    \vspace{-3mm}
    \caption{Impact of varying ratios of
super spreaders.}
    \vspace{-4mm}
    \label{fig:Pof}
\end{figure}

\subsection{Parameter Evaluation}\label{AA}
%We perform experiments to investigate how critical parameters of SegSketch impact its performance, including the IP segment width $G$ used in the halved-segment hashing strategy and the size of the bitmap $BC$ employed for subnet cardinality estimation. 

%To validate the impact of key parameters on SegSketch's performance, we conduct experiments with various configurations on the mixed MAWI and CAIDA datasets. We examine how IP segment width $G$ affects prefix length estimation error and detection accuracy, and how host bitmap size impacts detection accuracy.

%These experiments assess the influence of different IP segment widths $G$ on common prefix length estimation and super host detection, and the impact of varying sizes of the host bitmap used for subnet cardinality estimation on super host detection. 

%We likewise present only the super spreader results, as the difference from super receiver detection is solely the choice of source or destination IP as the host key, yielding similar performance.%Finally, we also validated the impact of different ratios of super hosts in the dataset on SegSketch's detection performance.

%\textbf{Common Prefix Length Estimation Accuracy.} 

\textbf{Varying the IP segment width}. %To validate the impact of different segment widths on common prefix length estimation accuracy, we conduct experiments using SegSketch configured with segment widths $G$ of 2, 4, 6, and 8 bits on the mixed MAWI and CAIDA datasets. The experimental results, shown in Figure~\ref{fig:prefixll}, illustrate the error between the true and estimated common prefix lengths. Across both mixed datasets, when $G=2$ and $G=4$, the prefix estimation error remains very low at approximately 0.3, indicating highly accurate common prefix length inference under fine-grained segmentation. Although the estimation error increases as $G$ grows to 6 and 8, the error values are still low, remaining below 1.5. These results confirm that the halved-segment hashing achieves robust and reliable common prefix length estimation across a wide range of segment widths. It maintains accurate estimation even with moderately large segment widths, avoiding the high hash overhead associated with excessively small $G$. 
To assess the impact of segment width on prefix length estimation, we evaluate SegSketch with $G$ set to 2, 4, 6, and 8 bits using the mixed MAWI and CAIDA datasets. As shown in Figure~\ref{fig:prefixll}, when $G=2$ or $G=4$, the average absolute estimation error remains near 0.3, demonstrating highly accurate inference under fine-grained segmentation. Although the error increases with $G=6$ and $G=8$, it remains below 1.5 across both datasets. This confirms that halved-segment hashing provides robust and accurate prefix estimation across a broad range of segment widths.
We further evaluate the effect of varying $G$ on super host detection, with the results in Appendix B.3. Given the trade-off between accuracy and computational overhead, we set $G=4$.

\vspace{-2mm}
\begin{figure}[htbp]
    \centering  
    \includegraphics[width=0.63\linewidth]{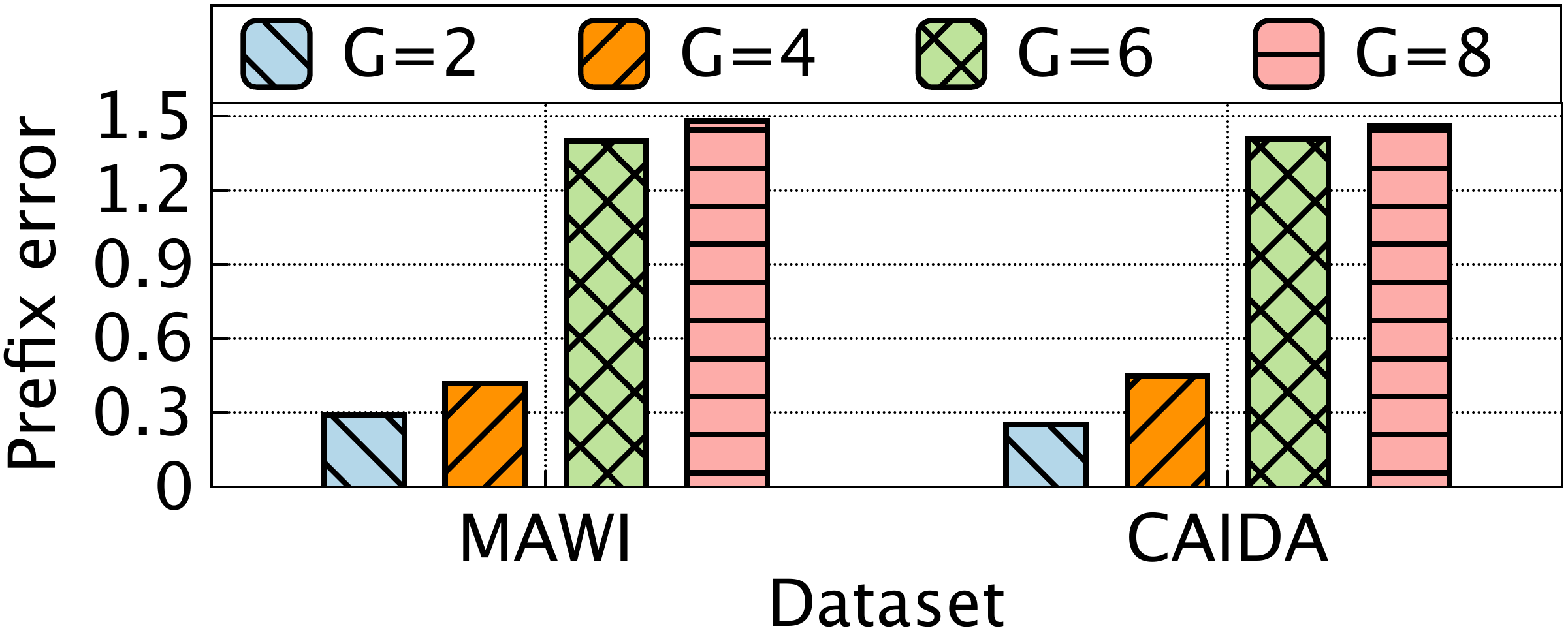}
    \vspace{-2mm}
    \caption{Common prefix length estimation error.}
     %\vspace{-4mm}
    \label{fig:prefixll}
\end{figure}
\vspace{-2mm}

\textbf{Varying the size of the host bitmap}. 
To evaluate the impact of various host bitmap sizes on super host detection, we configure SegSketch with host bitmap sizes ranging from 0.2KB to 1KB and conduct experiments on the mixed CAIDA dataset.  

%As illustrated in Figure~\ref{fig:size}, when the total memory is small, configurations with relatively large (0.75KB and 1KB) host bitmaps cause insufficient sketch buckets to store all true super hosts. This leads to intense hash collisions among hosts, and thus true super hosts are frequently evicted, causing a notably low F1‑Score under small total memory conditions. For example, with a total sketch memory of 32KB, SegSketch configured with a 0.5KB bitmap achieves an F1‑Score of 0.78, outperforming the 0.75KB and 1KB configurations by 45.52\% and 85.71\%, respectively. 
%Conversely, a small bitmap size, such as 0.25KB, provides insufficient bitmap cells relative to cardinality of the host. This results in severe hash collisions among connections of the same host, degrading cardinality estimation accuracy and yielding the highest ARE. Inaccurate cardinality estimation further causes erroneous bucket replacements between hosts, preventing accurate identification of true super hosts. Consequently, even with the maximum sketch memory of 512KB, the configuration of 0.25KB bitmap fails to achieve an F1-score of 1.0. 
%These observations demonstrate that both too small and too large host bitmap sizes are suboptimal, emphasizing the necessity of balanced configuration. Therefore, we set each host bitmap as 0.5KB.
\vspace{-3mm}
 \begin{figure}[htbp]
    \centering
   
    % Case2
    \subfigure[F1-Score]{
        \includegraphics[width=0.47\linewidth]{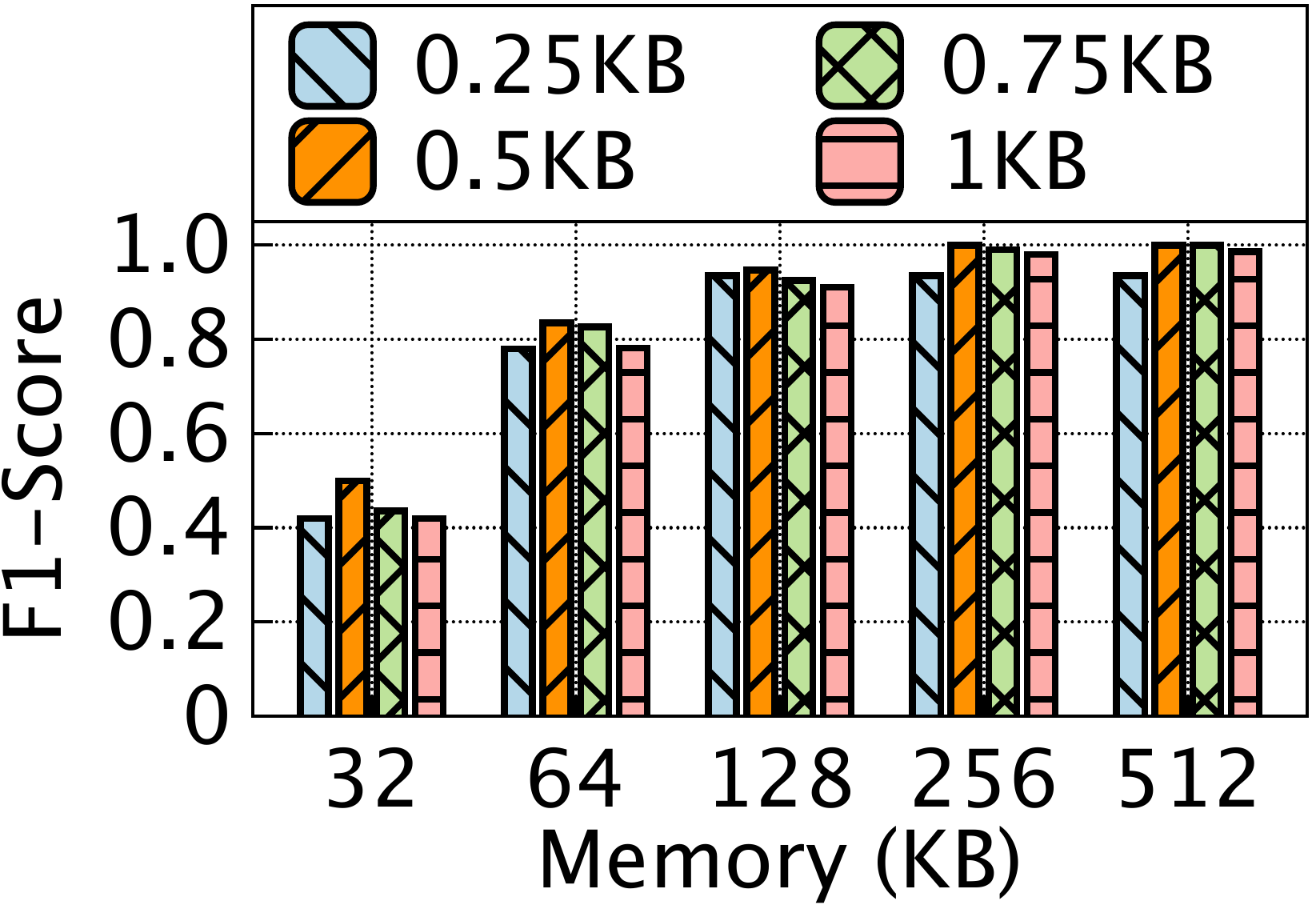}
    \label{fig:c}

    }
    % Case3
    \subfigure[ARE]{
        \includegraphics[width=0.47\linewidth]{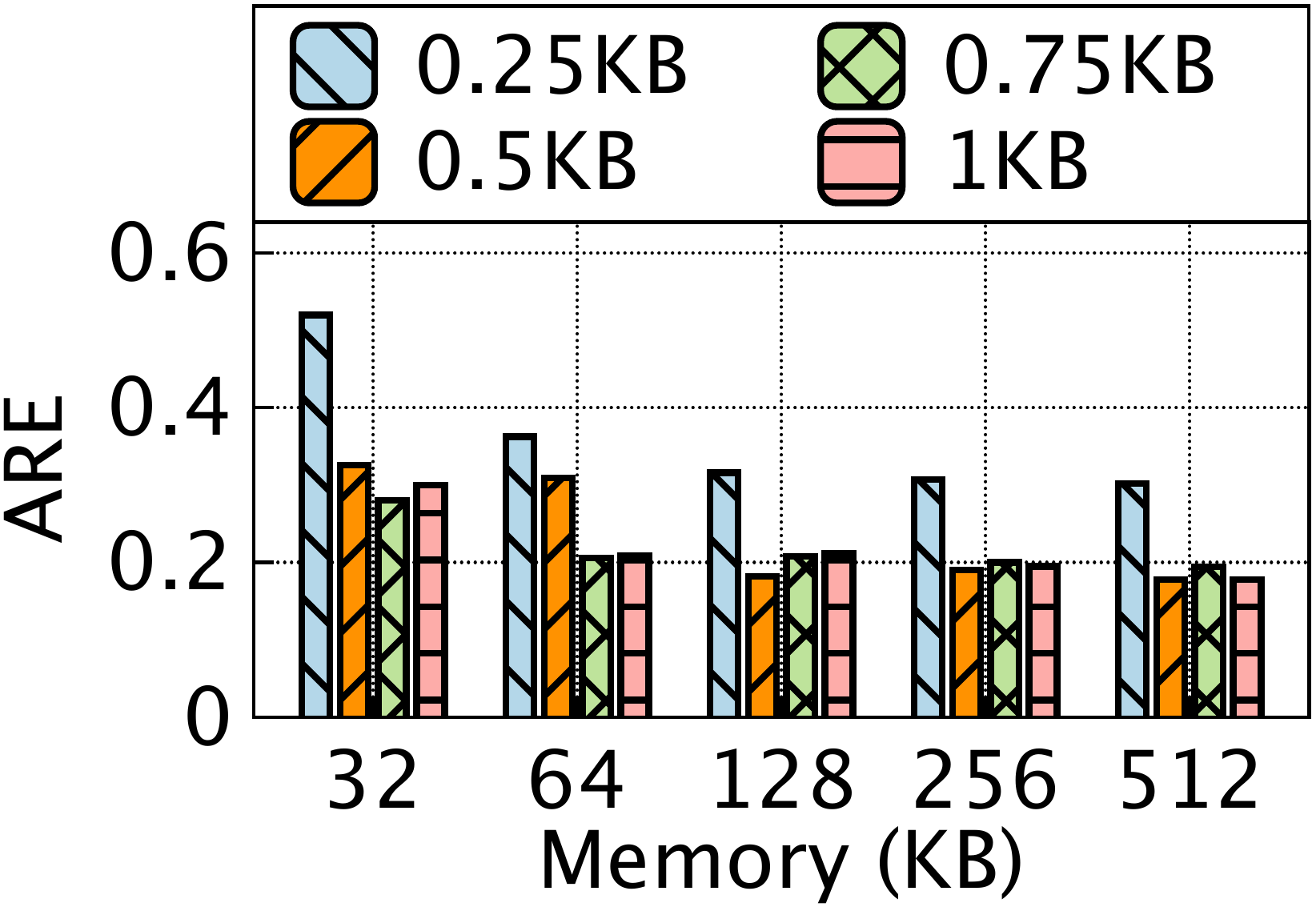}
    \label{fig:d}
    }
   \vspace{-3mm}
    \caption{Impact of varying host bitmap size.}
    \label{fig:size}
\end{figure}
\vspace{-2mm}

As illustrated in Figure~\ref{fig:size}, large host bitmap sizes (e.g., 0.75KB and 1KB) cause insufficient sketch buckets under limited total memory, leading to intense hash collisions among hosts and wrong evictions of true super hosts. For instance, with a total memory of 32KB, a 0.5KB bitmap achieves an F1-Score of 0.5, outperforming the 0.75KB and 1KB settings by 14.67\% and 19.04\%, respectively. 
Conversely, a small bitmap size (e.g., 0.25KB) provides too few bitmap cells relative to a host's cardinality, causing severe hash collisions among connections of the same host and degrading cardinality estimation accuracy. This leads to erroneous replacements and hinders correct super host identification, such that even with 512KB memory, the F1-Score remains below 1.0.
These results indicate that both excessively small and large bitmap sizes are suboptimal. Therefore, we set each host bitmap as 0.5KB.

 % \vspace{-3mm}

\subsection{Throughput}\label{AA}

%We compare the processing throughput of the four methods across varying memory sizes using the mixed MAWI and CAIDA datasets. Due to space constraints, results on the mixed MAWI dataset are detailed in Appendix B.2. The results presented in Figure~\ref{fig:th-caida} clearly show that SegSketch consistently outperforms all the other methods in terms of throughput across all the datasets. For instance, in the mixed MAWI dataset, SegSketch achieves the highest throughput of 29 Mpps, even under the minimal memory condition of 32 KB. The superior performance of SegSketch can be primarily attributed to its lightweight halved-segment hashing strategy and its compact bitmap-based design for subnet cardinality estimation. These structures significantly reduce memory overhead compared to RHHH, which relies on hierarchical estimators, SpreadSketch, which uses multi-level bitmap scanning, and Couper, which incorporates dual-layer estimators. These results clearly demonstrate that SegSketch not only enhances detection accuracy but also maintains high processing efficiency, making it highly suitable for real-time deployment in high-speed network environments.

We evaluate the throughput of all four methods across varying memory sizes using the mixed MAWI and CAIDA datasets. The results for the mixed MAWI are provided in Appendix B.4. As shown in Figure~\ref{fig:th-caida}, SegSketch consistently achieves the highest throughput% across all memory settings
, reaching 28 Mpps even under the minimal memory condition of 32KB. This performance is attributed to its lightweight halved-segment hashing and compact bitmap-based cardinality estimation, reducing overhead compared to RHHH’s hierarchical estimators, SpreadSketch’s multi-level scanning in  Multi-Resolution Bitmap, and Couper’s dual-layer design. These results confirm that SegSketch offers both high detection accuracy and processing efficiency, making it well suited for deployment in high-speed networks.
 
\vspace{-3mm}
\begin{figure}[htbp]
    \centering  
    \includegraphics[width=0.63\linewidth]{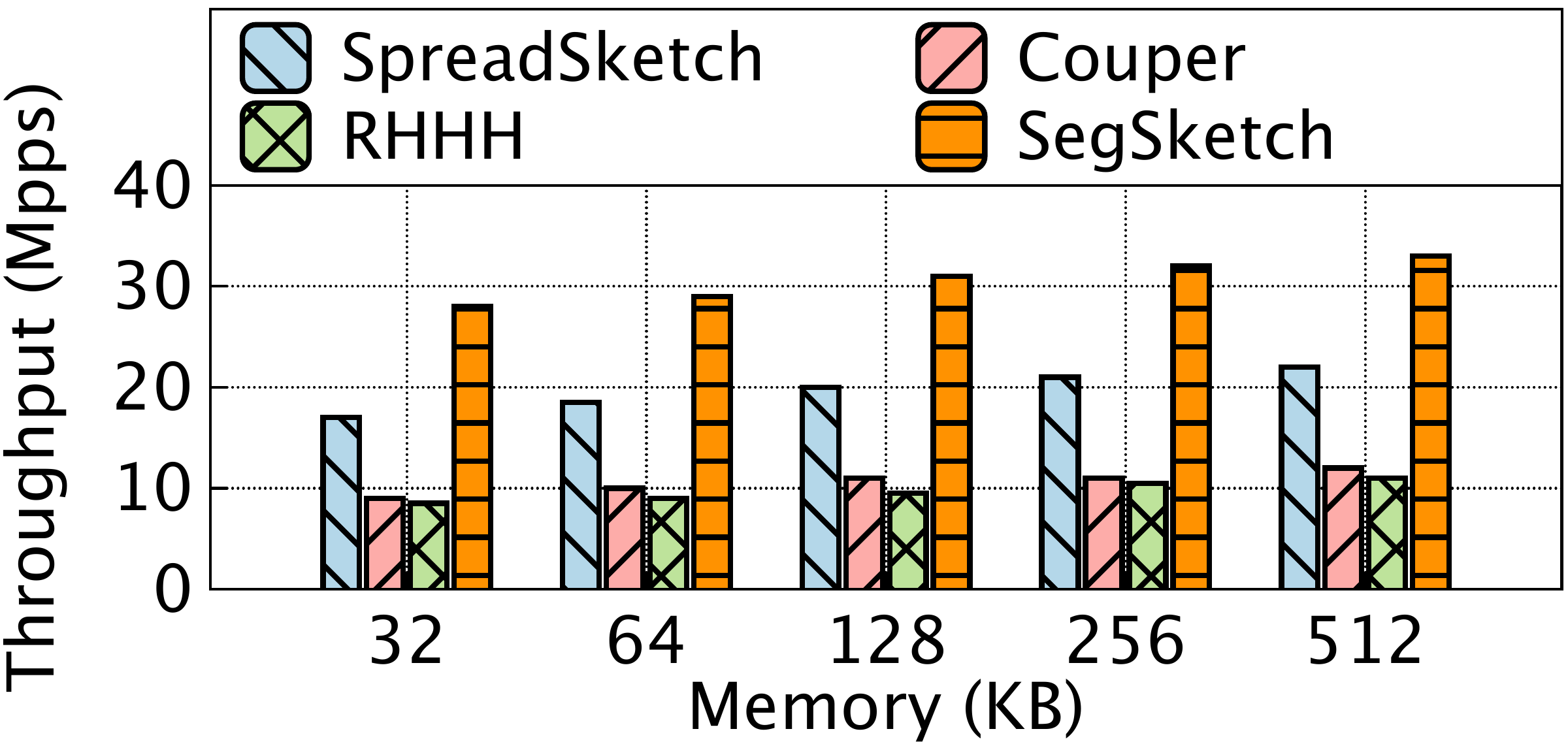}
    \vspace{-2mm}
    \caption{Throughput on the mixed CAIDA dataset.}
     %\vspace{-4mm}
    \label{fig:th-caida}
\end{figure}
\vspace{-3mm}

\vspace{-1mm}
\section{Performance on P4}
\label{sec:P4 implementation}

We implement the prototypes of SpreadSketch, Couper, RHHH and SegSketch using P4 \cite{p4} and deploy them on the Wedge 100BF-32X programmable switch \cite{tofino}. 
Due to the lack of support for complex instructions on the Tofino architecture, certain data structures such as heaps cannot be efficiently implemented in the data plane, and thus we integrate the control plane as assistance. For the halved-segment hashing strategy of SegSketch on P4, we derive a mask from the hash value of the IP segment and use it to extract the corresponding region of bits from the bitmap. A value greater than zero indicates that at least one bit within the specified region has been set, otherwise the region remains unmarked.

%We compare the P4 resource usage of SegSketch with SpreadSketch, Couper and RHHH in Table~\ref{tab:p4}. Specifically, the usage of the hash distribution unit is reduced to 11.11\% in SegSketch, lower than the other three methods. Couper consumes 22.22\% of the hash distribution unit due to its heavy reliance on hash operations for indexing two layers of register arrays. For RHHH, since multiple layers of cardinality estimators are simultaneously maintained for each candidate host, the number of hash operations is significantly higher than the other three schemes. For the on-chip memory (e.g., SRAM) usage across different algorithms, it is observed that both SegSketch and SpreadSketch exhibit lower SRAM consumption, which can be attributed to their use of highly compact registers. Furthermore, the usage of Gateways, VLIW instructuions and Stateful ALU of SegSketch is the lowest. Overall, SegSketch achieves resource-efficient deployment at programmable switches.

 %\vspace{-2mm}

\vspace{-2mm}
\setlength{\tabcolsep}{2pt} 
\renewcommand{\arraystretch}{0.9}  % 将行距缩小为原来的90%
\begin{table}[!htp]
  \caption{P4 usage percentage comparison.}
 \vspace{-1mm}
  \label{tab:p4}
  \begin{tabular}{lcccc}
    \toprule
    \textbf{Resource} & \textbf{SpreadSketch} & \textbf{Couper} & \textbf{RHHH} & \textbf{SegSketch} \\
    \midrule
    Match Crossbar    & 2.66\% & 6.17\% & 9.07\% & 4.12\% \\
    Hash Dist Unit         & 12.50\% & 18.06\% & 33.33\% & 11.11\% \\
    Gateways          & 6.77\% & 8.85\% & 15.10\% & 4.69\% \\
    VLIW Instructions &3.39\%  &  4.43\%      &7.55\%  &2.34\% \\
    SRAM             &1.67\% & 12.60\% & 3.65\% & 1.77\% \\
    TCAM             & 0.35\% & 0.35\% & 1.74\% & 0.69\% \\
    Stateful ALU      & 12.50\% & 12.50\% & 25.00\% & 8.33\% \\
    \bottomrule
  \end{tabular}
\end{table}
\renewcommand{\arraystretch}{1.0} % 恢复默认行距
\setlength{\tabcolsep}{6pt}       % 恢复默认列距
\vspace{-3mm}

Table~\ref{tab:p4} compares the P4 resource usage of SegSketch with SpreadSketch, Couper, and RHHH. SegSketch reduces the hash distribution unit usage to 11.11\%, lower than other methods. Couper incurs 18.06\% due to frequent hash operations across two register layers, while RHHH shows the highest due to maintaining multiple estimators per host. For on-chip memory, SegSketch and SpreadSketch exhibit lower SRAM usage owing to compact register usage. SegSketch also exhibits the lowest consumption of gateways, VLIW instructions, and stateful ALUs. These results highlight the resource efficiency of SegSketch for deployment on programmable switches.

%\vspace{-1mm}
\section{Related Work}
\label{sec:related work}
\subsection{Sketch-Based Cardinality Estimation}

%Existing sketch-based cardinality estimation methods fall into two main categories. The first focuses on single-host estimation  \cite{lc, mrb, hll}, employing compact data structures such as bitmaps or register arrays to estimate the cardinality of a single source or destination. The second targets multi-host estimation  \cite{gskt, vhs, rerskt, spreadsketch, cds, vbf, kjoin}, typically combining multiple single-host estimators (e.g., Linear Counting or Multi-Resolution Bitmaps  \cite{mrb}). 

Existing sketch-based cardinality estimation methods fall into two categories. The first focuses on single-host cardinality estimation \cite{lc, mrb, hll}, using compact structures such as bitmaps or register arrays to estimate the cardinality of a single source or destination. The second targets multi-host cardinality estimation \cite{gskt, vhs, rerskt, spreadsketch, cds, vbf, kjoin} by combining multiple single-host cardinality estimators such as Multi-Resolution Bitmaps \cite{mrb}.
   
%Specifically, SpreadSketch \cite{spreadsketch} maintains a candidate host key, a Multi-Resolution Bitmap, and the maximum number of leading zeros in the hash string of the IP per bucket. Replacement is determined by comparing leading-zero counts. CDS  \cite{cds} and VBF  \cite{vbf} use algebraically reversible hashes based on the Chinese Remainder Theorem to recover high-cardinality host keys without explicit storage. To address inefficiency in uniform memory allocation, where most hosts have low cardinality and only a few are large, Couper  \cite{couper} adopts a two-layer structure leveraging the coupon collector's principle. Small flows are retained in the first layer, while potential large flows are promoted to a second layer.

Specifically, SpreadSketch \cite{spreadsketch} stores candidate host key, Multi-Resolution Bitmap, and the maximum number of leading zeros in the IP hash string per bucket, with replacement decided by comparing leading-zero counts. CDS \cite{cds} and VBF \cite{vbf} employ algebraically reversible hashing based on the Chinese Remainder Theorem to recover high-cardinality hosts. To mitigate inefficiency from uniform memory allocation under skewed cardinality distribution, Couper \cite{couper} introduces a two-layer structure guided by the coupon collector’s principle, retaining small flows in the first layer and transferring large ones to the second layer.

However, these approaches struggle to distinguish malicious super hosts from benign ones that connect to many diverse IPs across the whole network, as they ignore whether connections share the same subnet address, leading to high false positive rates.

\vspace{-2mm}
\subsection{Hierarchical Methods for HHH Detection}

Hierarchical Heavy Hitters (HHHs) are flows sharing a common IP prefix whose cumulative packet counts exceed a given threshold. %The work in~ \cite{alenex} proposes an algorithm that maintains a Space Saving~ \cite{space} instance at each common prefix level, updating relevant counters in all levels per packet and estimating frequencies by traversing the hierarchy to subtract sub-prefix contributions. 

%The work of  \cite{alenex} proposes an algorithm that maintains a Space Saving \cite{space} instance per prefix level, updating the relevant counters in all layers per packet and estimating frequencies by traversing the hierarchy to subtract sub-prefix contributions. While effective, this incurs $O(H)$ update cost per packet. Other methods, such as  \cite{tkdd08} with $O(H \log(N/\varepsilon))$ and  \cite{itc09} with $O(H^{3/2}/\varepsilon)$ complexity, also suffer from high processing overhead. To address this issue, RHHH \cite{sigcomm17} reduces update complexity to $O(1)$ by randomly selecting a prefix level to update for each packet and scaling the estimates by the inverse sampling probability. 
%, making it practical for high-speed environments.

The work of \cite{alenex} maintains a Space Saving \cite{space} instance per prefix level, updating all $H$ layers per packet and estimating frequencies by subtracting sub-prefix contributions, which incurs $O(H)$ update cost per packet. Other methods, such as \cite{tkdd08} with $O(H\cdot \log(N/\varepsilon))$ and \cite{itc09} with $O(H^{3/2}/\varepsilon)$ complexity, where $\varepsilon$ is the allowed relative estimation error for each flow’s frequency, also suffer from high processing overhead. RHHH \cite{sigcomm17} reduces the complexity to $O(1)$ by randomly selecting a prefix level to update for each packet.

Although these hierarchical methods can detect super hosts with common IP prefixes by replacing counters with single-host cardinality estimators, this requires maintaining estimators for all possible prefix lengths since subnet lengths are unknown beforehand and can vary across networks. This leads to rapidly increasing memory overhead, making them unsuitable for on-chip deployment.

\vspace{-2mm}
\section{Conclusion}
\label{sec:conclusion}
%We propose SegSketch, a novel sketch for accurate and efficient super-host detection that captures both high cardinality and common subnet structure. SegSketch adopts a lightweight halved-segment hashing strategy to infer subnet prefix lengths and employs a bitmap-based estimator to measure cardinality within subnets. Theoretical analysis quantifies its estimation error. We implement SegSketch on a Barefoot Tofino switch using P4 and validate its performance through trace-driven experiments, showing superior accuracy, memory efficiency, and throughput compared to state-of-the-art methods.
    
We propose SegSketch, a novel sketch for accurate and efficient super-host detection that accounts for both high cardinality and common subnet address characteristics. SegSketch employs a lightweight halved-segment hashing strategy to infer subnet prefix lengths and utilizes bitmap-based cardinality estimators to estimate cardinality within subnet. Theoretical analysis quantifies the cardinality estimation error of SegSketch. We implement SegSketch on a Barefoot Tofino switch using P4 and conduct trace-driven experiments, confirming that SegSketch outperforms state-of-the-art solutions in detection accuracy, memory efficiency, and throughput.

\vspace{-2mm}
\section*{Acknowledgment}
We thank the anonymous reviewers for their valuable comments. This work was supported by the National Natural Science Foundation of China (U24A20245, 62132022) and the Science and Technology Innovation Program of Hunan Province (2024RC1005). This work was carried out in part using computing resources at the High Performance Computing Center of Central South University. Jiawei Huang is the corresponding author.

\bibliographystyle{ACM-Reference-Format}
\bibliography{reference}

\appendix

\section{Theoretical Analysis}

\subsection{Cardinality Estimation Error Bound}
\begin{proof}

%To compute the expectation term, we analyze the process of hashing $U$ distinct flows into a bitmap of size $M$. There are $\binom{M}{t}$ ways to choose which $t$ out of $M$ bits will be set. Once the $t$ bits are selected, it should be ensured that each of them receives at least one hashed flow. The probability that all $t$ bits receive at least one of the $U$ flows is given by the inclusion-exclusion principle as
%\begin{align}
%    Pr = 
%    \sum_{i=0}^{t} (-1)^i \binom{t}{i}
%    \left( 1 - \frac{i}{t} \right)^{U}.
%\end{align}
%Furthermore, the probability that all $U$ hashed values fall into these $t$ bits is $\left(\frac{t}{M}\right)^{U}$. Thus, the probability that $U$ distinct flows occupy exactly $t$ bits in the bitmap is
%\begin{align}
%    P(U, t) 
%= \binom{M}{t} \frac{t^{U}}{M^{U}}
%\sum_{i=0}^{t}
%(-1)^{i}\binom{t}{i}
%\left(1-\frac{i}{t}\right)^{U},
%\quad
%t \le M,\ t \le U,
%\end{align}
%where $M$ is the effective bitmap size determined by the hash strategy, and $U$ is the number of flows that are not in the targeted subnet but are misclassified into it.

For a flow, the probability that it is hashed to a specific bit $j\in\{1,\dots,M\}$ in a bitmap of size $M$ is $M^{-1}$. Conversely, the probability that the flow is not hashed to this bit is $1-M^{-1}$. Assuming $R$ independent and uniformly hashed flows, the probability that none of them is hashed to bit $j$ is $(1-M^{-1})^{R}$.
Therefore, the probability that bit $j$ is set by at least one flow is
\begin{equation}
\Pr(Z_j=1)=1-(1-M^{-1})^{R}.
\end{equation}

By linearity of expectation, the expected number of set bits $T=\sum_{j=1}^{M} Z_j$ in a bitmap of size $M$ hashed by $R$ flows is
\begin{equation}
\mathbb{E}[T]=\sum_{j=1}^{M} \mathbb{E}[Z_j]
= M\Bigl[\,1-(1-M^{-1})^{R}\,\Bigr].
\label{eq:tong}
\end{equation}
%which is the estimated cardinality using Linear Counting.

Considering the estimation error $\varepsilon$ introduced by the Linear Counting algorithm, using the Markov inequality, we have
\begin{align}
\Pr \left\{ |\hat{C}(x) - C(x)| \geq \varepsilon \right\}
\leq {\varepsilon}^{-1}{\mathbb{E}\!\left[|\hat{C}(x)-C(x)|\right]}.
\label{eq:mark}
\end{align}

Let $U$ denote the number of flows that are not in the targeted subnet but are misclassified into it, and these flows induce estimation error between $\hat{C}(x)$ and $C(x)$. Based on equation (\ref{eq:tong}), the expected estimation error of subnet cardinality is
\begin{align}
    \mathbb{E}[|\hat{C}(x)-C(x)|] 
= M\left[\,1 - \bigl(1 - M^{-1}\bigr)^{U}\,\right].
\label{eq:expect}
\end{align}

Substituting equation \eqref{eq:expect} into equation \eqref{eq:mark} yields the bound in \eqref{eq:unified-bound}

as follows:
\begin{equation}
\Pr \left\{ |\hat{C}(x) - C(x)| \geq \varepsilon \right\}
\;\le\;
\varepsilon^{-1}\,
M\left[\,1 - \bigl(1 - M^{-1}\bigr)^{U}\,\right].
\label{eq:unified-bound}
\end{equation}

\iffalse
\begin{table*}
\centering
\caption{Variable values under different hashing strategies}
\vspace{-2mm}
\begin{tblr}{
  width = \linewidth, 
  colspec = {
    Q[c] 
    Q[l] 
    Q[l]
  },
  cells = {c},
  rowsep = 0.9pt,
  cell{1}{1} = {r=2}{},
  cell{1}{2} = {c=2}{},
  hline{1,3,6} = {-}{},
  hline{2} = {2-3}{},  % 第二条横线扩展到第2列和第3列
}
\textbf{Variables} & \textbf{Hashing Strategies} & \\
                   & \textbf{Full-address}       & \textbf{Host-address} \\
$\varepsilon$      & $N(x)-M\left[\,1 - \bigl(1 - M^{-1}\bigr)^{N(x)}\,\right]$
                   & $C(x)+U-M\left[\,1 - \bigl(1 - M^{-1}\bigr)^{C(x)+U}\,\right]$ \\
$M$                & $2^{32}$
                   & $2^{32 - \lfloor l/G \rfloor \cdot G}$ \\
$U$                & $N(x)-C(x)$
                   & $(N(x)-C(x))\left(\tfrac{1}{2}\right)^{\lfloor l/G \rfloor}$
\label{tab:math}
\end{tblr}
\end{table*}
\fi

Instantiation of the two cases is as follows:

\textbf{Case1: full-address hashing}.
Let $M = 2^{32}$, $U = N(x)-C(x)$, and $\varepsilon = \varepsilon_{{full}}= %\delta \cdot N(x)$.
N(x)-M\left[\,1 - \bigl(1 - M^{-1}\bigr)^{N(x)}\,\right]$. 
Substituting these into equation \eqref{eq:unified-bound} yields
\begin{align}
\Pr \left\{ |\hat{C}_{{full}}(x) - C(x)| \geq \varepsilon_{{full}} \right\}\quad\quad\quad\quad\quad\quad\quad\quad\quad\quad&\notag\\
\le\varepsilon_{{full}}^{-1}\;
2^{32}\!\left[\,1-\left(1-2^{-32}\right)^{\,N(x)-C(x)}\right]
\Biggr..
\end{align}

\textbf{Case2: host-address hashing based on common prefix inference}. 
%Let $l$ denote the true subnet prefix length and $G$ the IP segment width under the halved-segment hashing strategy. 
Each of the $\lfloor l/G\rfloor$ segments matches the subnet prefix with probability $1/2$ based on the halved-segment hashing strategy, so the probability that one of the $U$ flows is mistakenly inferred to share an identical prefix with the targeted subnet is 
\begin{equation}
\alpha = \left(\frac{1}{2}\right)^{\lfloor \frac{l}{G}\rfloor}.
\end{equation}
    
Hence, we get $M = 2^{\,32-\lfloor l/G\rfloor G}$, $U = (N(x)-C(x))\alpha = (N(x)-C(x))\left(\tfrac12\right)^{\lfloor l/G\rfloor}$, and $\varepsilon = \varepsilon_{{host}} %\delta \cdot C(x)$.
=  C(x)+U-M\left[\,1 - \bigl(1-M^{-1}\bigr)^{C(x)+U}\,\right]$. 
Substituting these into equation \eqref{eq:unified-bound} gives
\begin{align}
\Pr& \left\{|\hat{C}_{{host}}(x)-C(x)| \geq \varepsilon_{{host}} \right\}
 \notag \\
\le\; &\varepsilon_{{host}}^{-1}\,
2^{\,32-\left\lfloor \frac{l}{G} \right\rfloor  G}
\Biggl[1-\left(1-2^{\left\lfloor \frac{l}{G}\right\rfloor  G-32}\right)^{(N(x)-C(x))\left(\tfrac12\right)^{\left\lfloor \frac{l}{G}\right\rfloor}}
\Biggr].
\end{align}

\end{proof}

\subsection{Comparison between Expected Cardinality Estimation Error of Two Hashing Strategies}
\begin{proof}
Based on Theorem 1, the difference between the expected estimation error of subnet cardinality through hashing the full IP addresses and that obtained by hashing only the host addresses is
\begin{equation}
\begin{aligned}
&\mathbb{E}\!\left[|\hat{C}_{{full}}(x)-C(x)|\right]
-
\mathbb{E}\!\left[|\hat{C}_{{host}}(x)-C(x)|\right] \\
&=
2^{32}\!\left[1-\left(1-2^{-32}\right)^{\,N(x)-C(x)}\right]\\
&-
2^{\,32-\left\lfloor \frac{l}{G} \right\rfloor  G}
\left[
1-
\left(1-2^{\left\lfloor \frac{l}{G}\right\rfloor  G-32}\right)^{(N(x)-C(x))\left(\tfrac12\right)^{\left\lfloor \frac{l}{G}\right\rfloor}}
\right].
\end{aligned}
\end{equation}

Let $Q = N(x) - C(x)$ and $r = \left\lfloor \frac{l}{G} \right\rfloor$, where $r,G\in\mathbb{N}_{+}$ and $rG<32$.
Then,

\begin{align}
\quad&\mathbb{E}\!\left[|\hat{C}_{{full}}(x)-C(x)|\right]
-
\mathbb{E}\!\left[|\hat{C}_{{host}}(x)-C(x)|\right] \notag\\
&= 2^{32 - r G}
\Bigl[
(2^{r G} - 1)
- 2^{r G} \bigl(1 - 2^{-32}\bigr)^{Q}
+ \bigl(1 - 2^{r G-32}\bigr)^{\,Q\left(\tfrac{1}{2}\right)^{r}}
\Bigr].
\label{eq:14}
\end{align}

Based on the Taylor expansion, we have
\begin{align}
&(1 - 2^{-32})^Q 
= 1 - Q \cdot 2^{-32} + \frac{Q(Q-1)}{2}\,(2^{-32})^2+ o\!\left((2^{-32})^2\right), \label{eq:taylo2}\\[1ex]
&(1 - 2^{\,r G-32})^{\,Q\left(\tfrac{1}{2}\right)^{r}}
= 1 - Q\left(\tfrac{1}{2}\right)^{r} \cdot 2^{\,r G-32} \notag\\
&+ \frac{Q\left(\tfrac{1}{2}\right)^{r}\!\Bigl(Q\left(\tfrac{1}{2}\right)^{r}-1\Bigr)}{2}\,(2^{\,r G-32})^2+ o\!\left((2^{\,r G-32})^2\right).
\label{eq:taylo}
\end{align}

Since $o\!\left((2^{-32})^2\right)$ and $o\!\left((2^{\,r\cdot G-32})^2\right)$ are negligible compared with the leading terms, 
we truncate the higher-order terms and approximate equation~\eqref{eq:taylo2} and equation~\eqref{eq:taylo} as
\begin{align}
&(1 - 2^{-32})^Q 
\approx 1 - Q \cdot 2^{-32} + \frac{Q(Q-1)}{2}\,(2^{-32})^2, \label{eq:taylor32}\\
&(1 - 2^{\,r G-32})^{\,Q\left(\tfrac{1}{2}\right)^{r}}
\approx 1 - Q\left(\tfrac{1}{2}\right)^{r} \cdot 2^{\,r G-32} \notag\\
&+ \frac{Q\left(\tfrac{1}{2}\right)^{r}\!\Bigl(Q\left(\tfrac{1}{2}\right)^{r}-1\Bigr)}{2}\,(2^{\,r G-32})^2.
\label{eq:18}
\end{align}

Substituting equations~\eqref{eq:taylor32} and~\eqref{eq:18} into equation~\eqref{eq:14}, and after simplification, we obtain
\begin{align}
&\mathbb{E}\!\left[ |\hat{C}_{{full}}(x) - C(x)| \right] 
- \mathbb{E}\!\left[ |\hat{C}_{{host}}(x) - C(x)| \right] \\ \notag
&\approx Q^{2}\,2^{-33}\!\left(2^{\,rG-2r}-1\right)
\;+\;
Q\!\left(1-2^{-r}-2^{\,rG-r-33}+2^{-33}\right)
> 0.
\end{align}

\begin{figure}[htbp]
    \centering
    \includegraphics[width=0.63\linewidth]{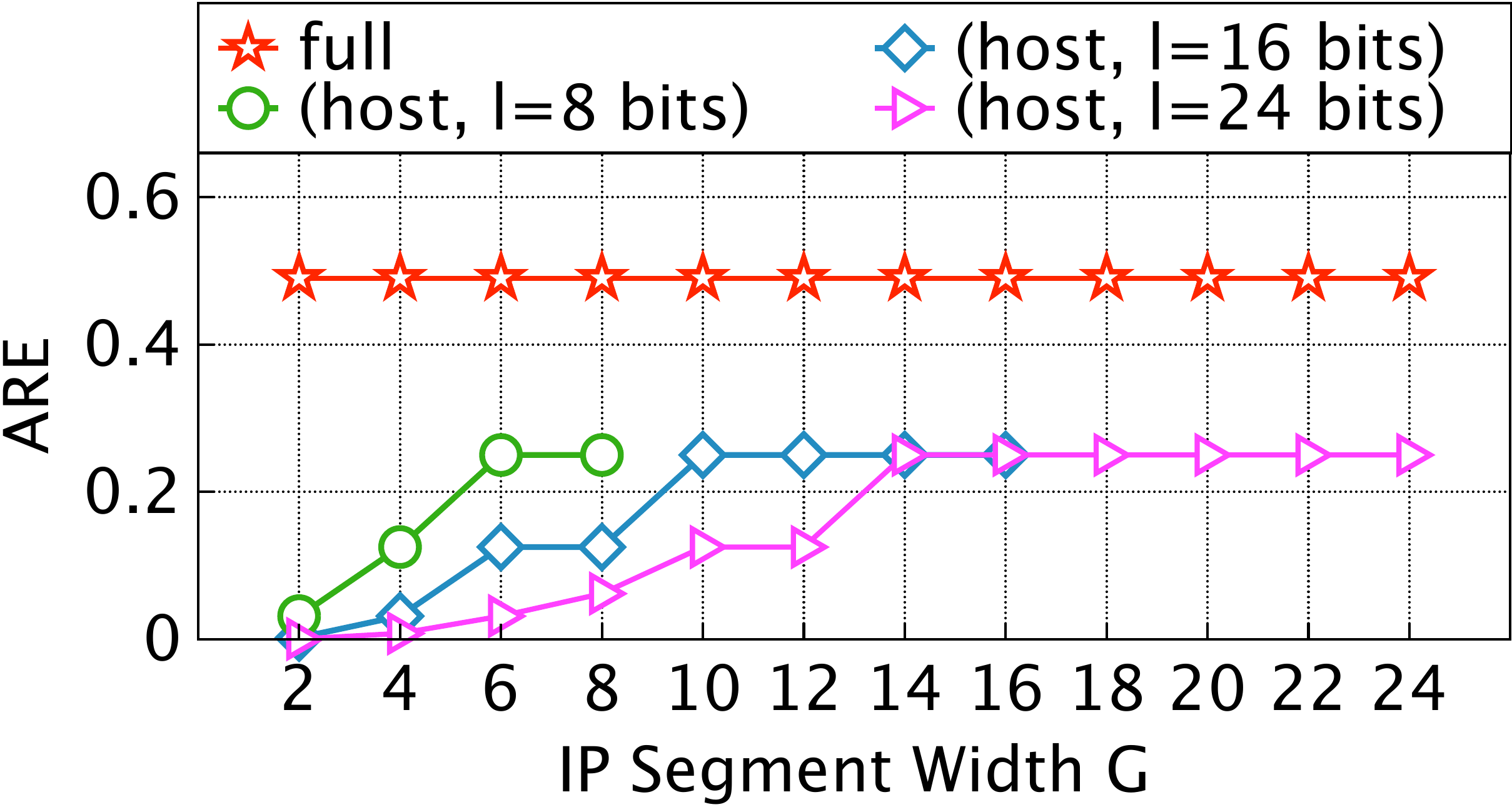}
    \vspace{-1mm}
    \caption{ARE of subnet cardinality estimation through full-address and host-address hashing. $l$ represents subnet length.}
    \vspace{-3mm}
    \label{fig:math}
\end{figure}

Figure~\ref{fig:math} compares the Average Relative Error (ARE) of $\mathbb{E} \left[ {|\hat{C}_{{host}}(x) - C(x)|}/{C(x)} \right]$ under varying IP segment width $G$ with $\mathbb{E} \left[{|\hat{C}_{{full}}(x) - C(x)|}/{C(x)} \right]$. The increase of ARE for host-address hashing as $G$ increases indicates that estimating subnet cardinality though host-address hashing yields lower estimation error with finer IP segmentation. Moreover, for all values of $G$, host-address hashing consistently achieves lower ARE than full-address hashing.

\end{proof}

\section{Experimental Results}

\subsection{Impact of Threshold Scaling Factor $\theta$ on Super Host Detection}

The threshold scaling factor $\theta$ adjusts the detection threshold for identifying super hosts. Specifically, it ensures that the threshold is proportional to the number of distinct connections a host can make within a subnet. Tuning $\theta$ can strike a balance between precision and recall, with the results shown in Table 3.

\begin{table}[ht]
\centering
\caption{Impact of $\theta$ on detection performance.}
\begin{tabular}{ccccc}
\toprule
\textbf{Memory (KB)} & \textbf{$\theta$} & \textbf{Precision} & \textbf{Recall} & \textbf{F1-Score} \\
\midrule
 & 0.35 & 0.45 & 0.36 & 0.40 \\
32   & 0.50 & 0.90 & 0.35 & 0.50 \\
   & 0.65 & 0.91 & 0.34 & 0.50 \\
\midrule
 & 0.35 & 0.500 & 0.75 & 0.60 \\
64   & 0.50 & 1.00  & 0.73 & 0.84 \\
   & 0.65 & 1.00  & 0.38 & 0.55 \\
\midrule
 & 0.35 & 0.53  & 0.91 & 0.67 \\
128    & 0.50 & 1.00  & 0.90 & 0.95 \\
    & 0.65 & 1.00  & 0.52 & 0.68 \\
\midrule
 & 0.35 & 0.57  & 1.00 & 0.73 \\
256    & 0.50 & 1.00  & 1.00 & 1.00 \\
    & 0.65 & 1.00  & 0.66 & 0.80 \\
\midrule
 & 0.35 & 0.59  & 1.00 & 0.74 \\
512    & 0.50 & 1.00  & 1.00 & 1.00 \\
    & 0.65 & 1.00  & 0.68 & 0.81 \\
\bottomrule
\end{tabular}
\end{table}

  \vspace{-2mm}
Based on the above results, we recommend using smaller $\theta$ to reduce false negatives and prioritize recall, and larger $\theta$ to reduce false positives and prioritize precision. Considering these factors, we set $\theta$ as 0.5.

\subsection{Accuracy of Super Host Detection}

\textbf{Super spreader detection}. We further evaluate the detection performance of SegSketch on the mixed dataset constructed from MAWI2021 and UNSW-NB15, with results presented in Figure~\ref{fig:mawiii}. SegSketch again consistently outperforms all other schemes under all memory configurations. Under the smallest memory budget of 32KB, it achieves higher Precision by $3.53\times$, $4.45\times$, and $28.60\times$, higher Recall by $1.60\times$, $1.12\times$, and $2.45\times$, and higher F1-Score by $2.45\times$, $2.57\times$, and $5.29\times$, while reducing ARE by 88.57\%, 18.18\%, and 65.38\%, compared to SpreadSketch, Couper, and RHHH, respectively.

The advantage of SegSketch becomes more pronounced under lower memory due to its efficient subnet-aware design, which enables accurate super-host detection with minimal resource usage. As memory increases, all methods show improved performance. However, SegSketch maintains a clear lead due to its ability to distinguish hosts targeting a single subnet versus those spreading connections across the whole network. In contrast, SpreadSketch and Couper suffer from high false positives, while RHHH performs poorly under limited memory because of its costly multi-layer estimator structure.

    \vspace{-2mm}
\begin{figure}[htbp]
    \centering
    % Case1-1
    \subfigure[Precision]{
        \includegraphics[width=0.465\linewidth]{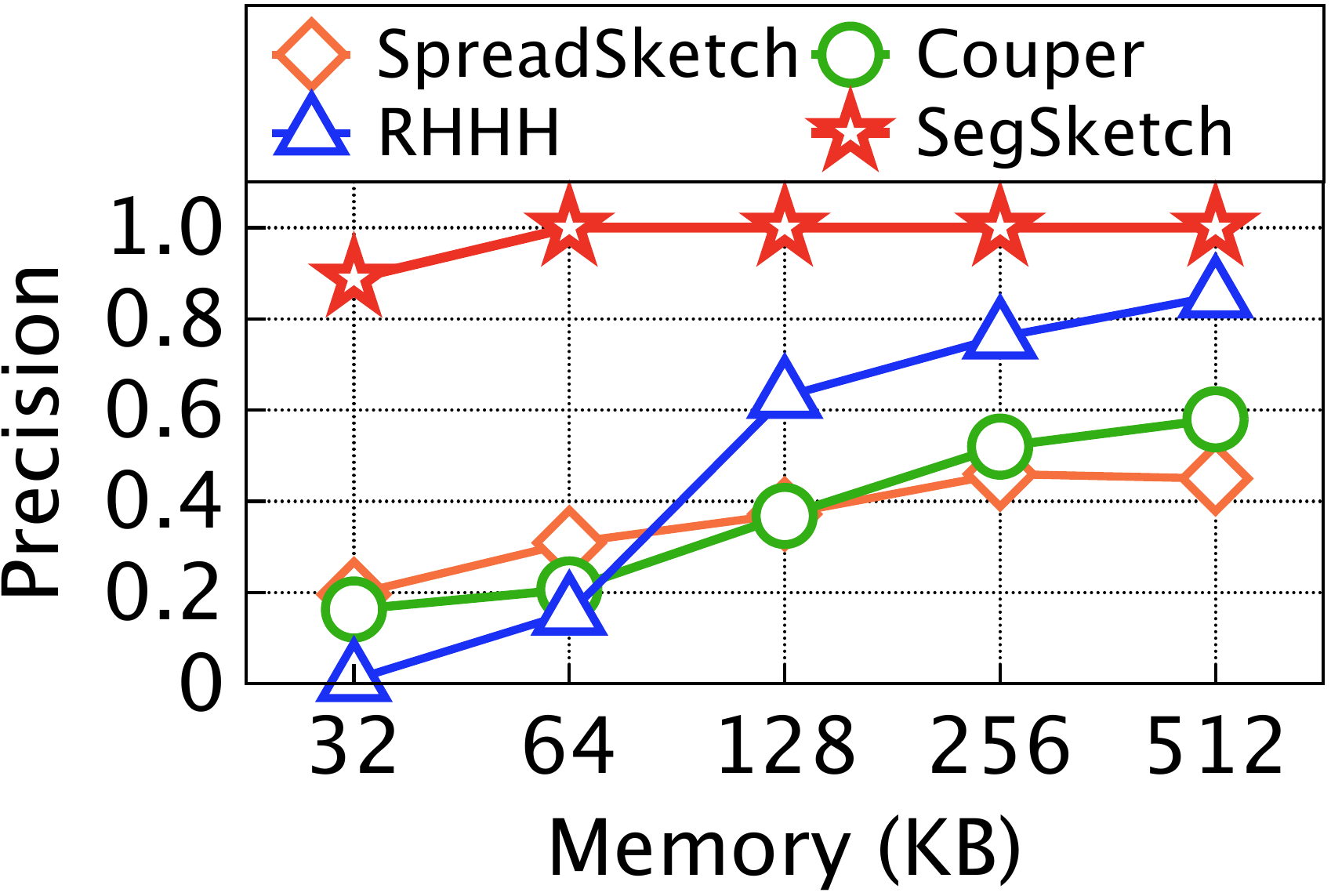}
     \label{fig:a}
    }
    \subfigure[Recall]{
        \includegraphics[width=0.465\linewidth]{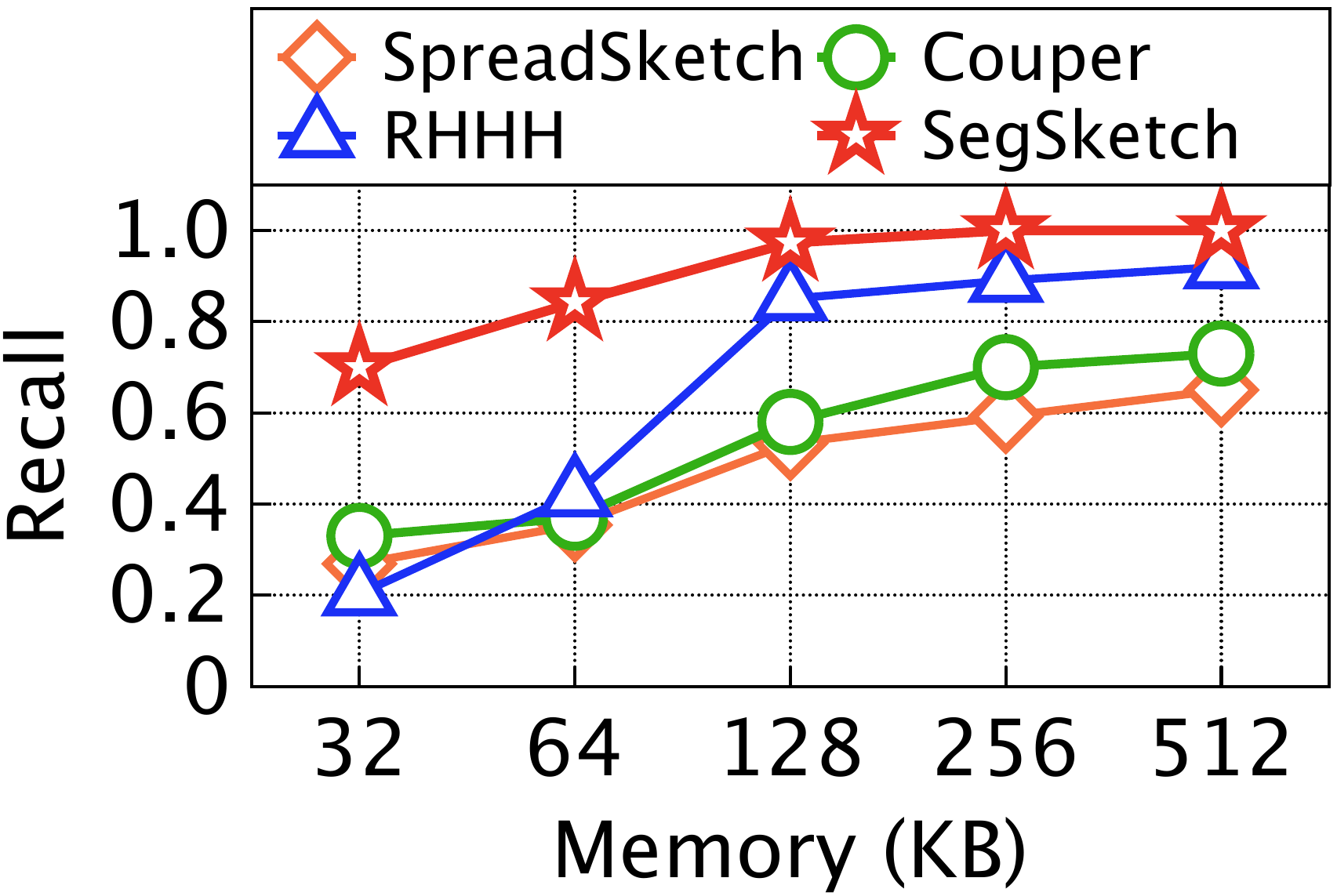}
     \label{fig:b}
    }
    \\
    % Case2
    \subfigure[F1-Score]{
        \includegraphics[width=0.465\linewidth]{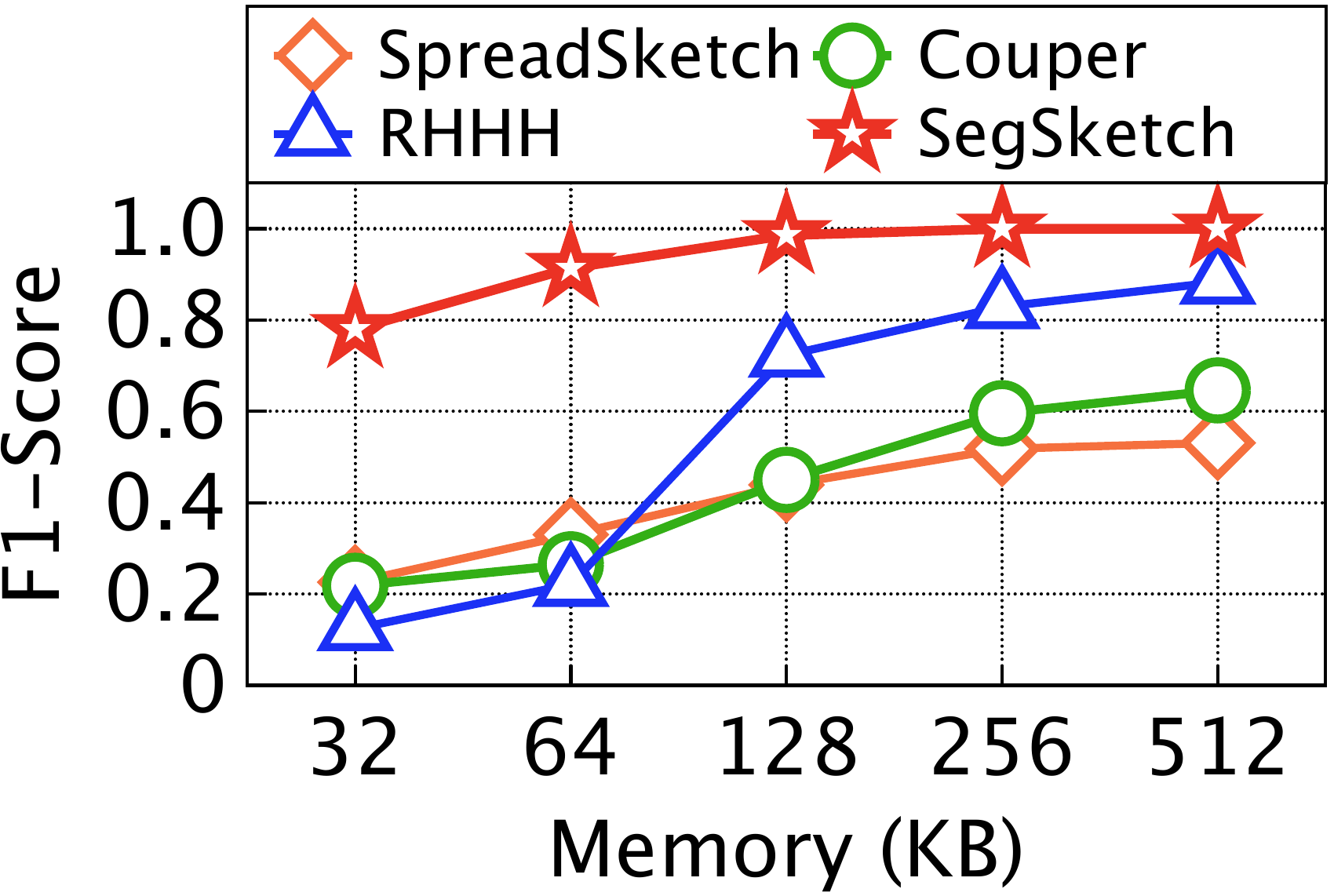}
    \label{fig:c}

    }
    % Case3
    \subfigure[ARE]{
        \includegraphics[width=0.465\linewidth]{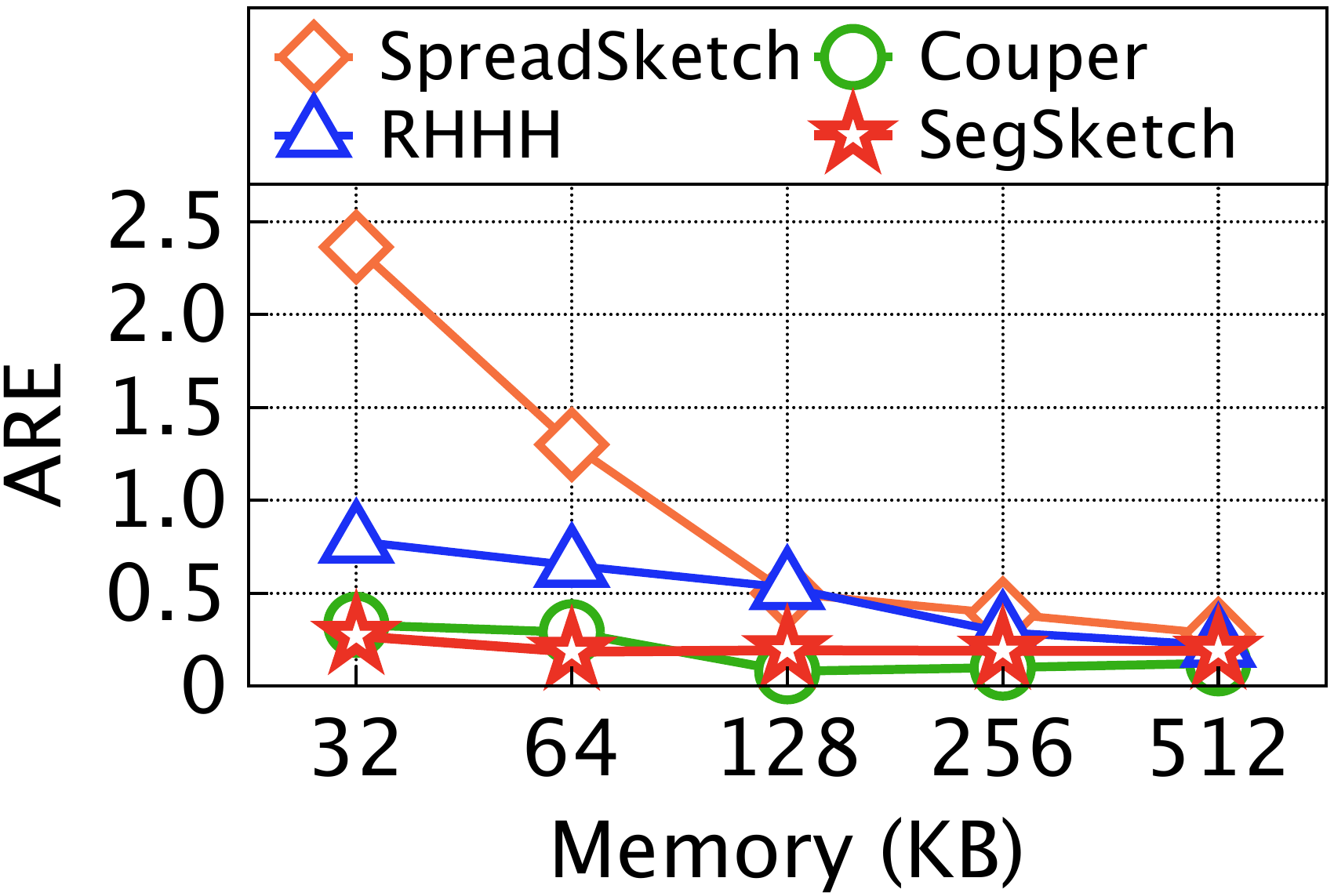}
    \label{fig:d}
    }
    \vspace{-1mm}
    \caption{Super spreader detection performance on the dataset mixing UNSW-NB15 with MAWI2021.}
   %\vspace{-3mm}
    \label{fig:mawiii}
\end{figure}

  \vspace{-3mm}
\begin{figure}[htbp]
    \centering
    % Case1-1
    \subfigure[Precision]{
        \includegraphics[width=0.465\linewidth]{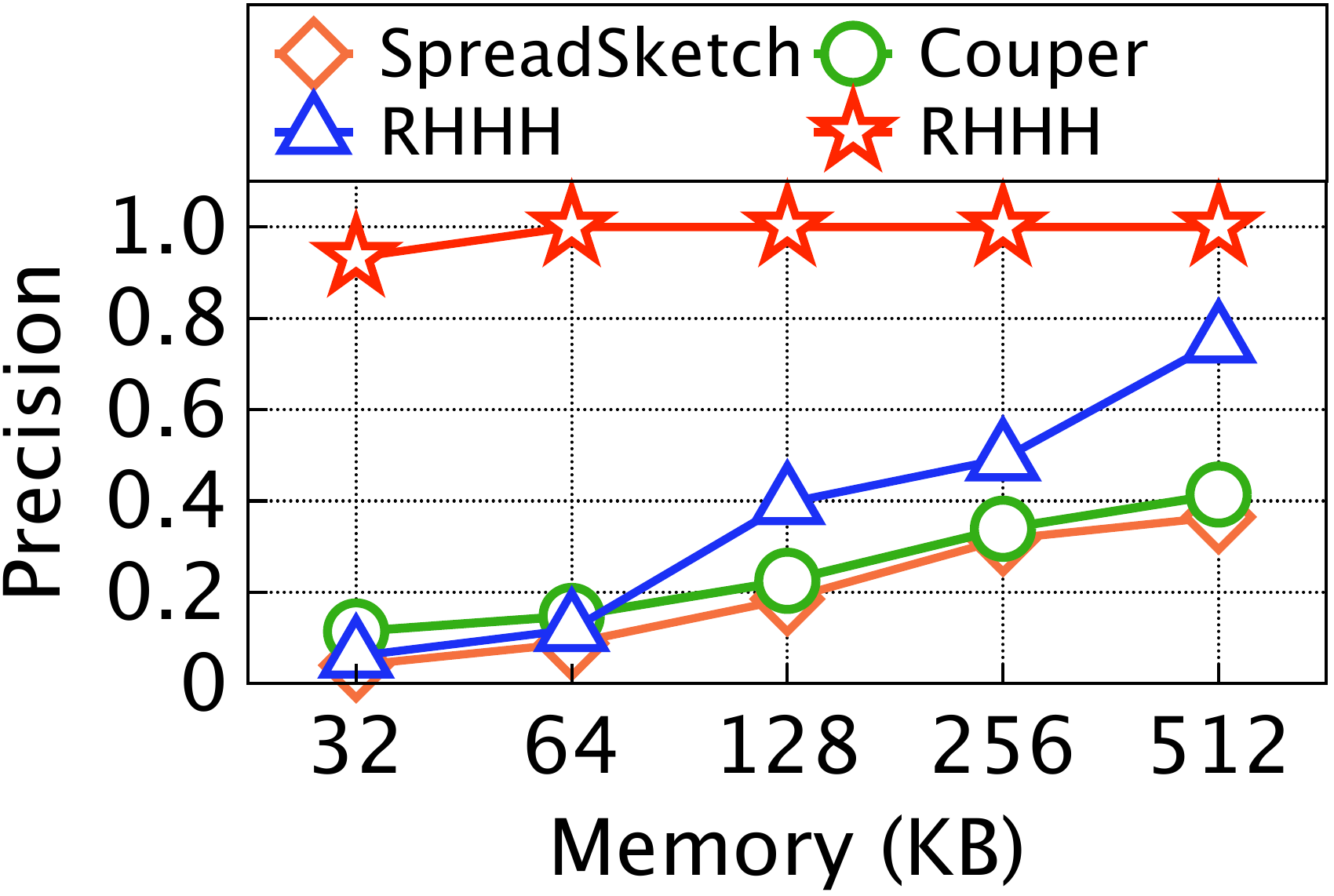}
     \label{fig:a}
    }
    \subfigure[Recall]{
        \includegraphics[width=0.465\linewidth]{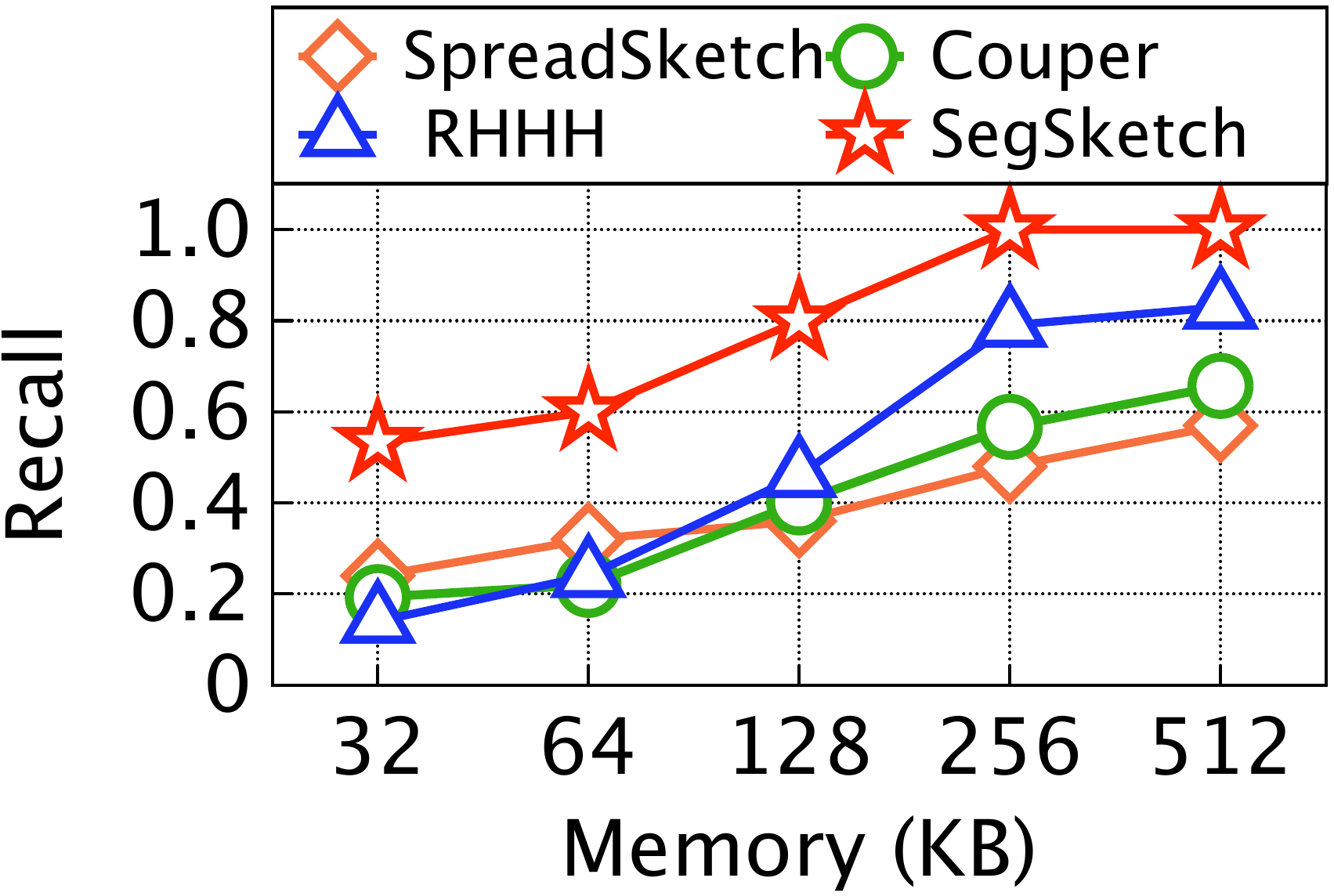}
     \label{fig:b}
    }
    \\
    % Case2
    \subfigure[F1-Score]{
        \includegraphics[width=0.465\linewidth]{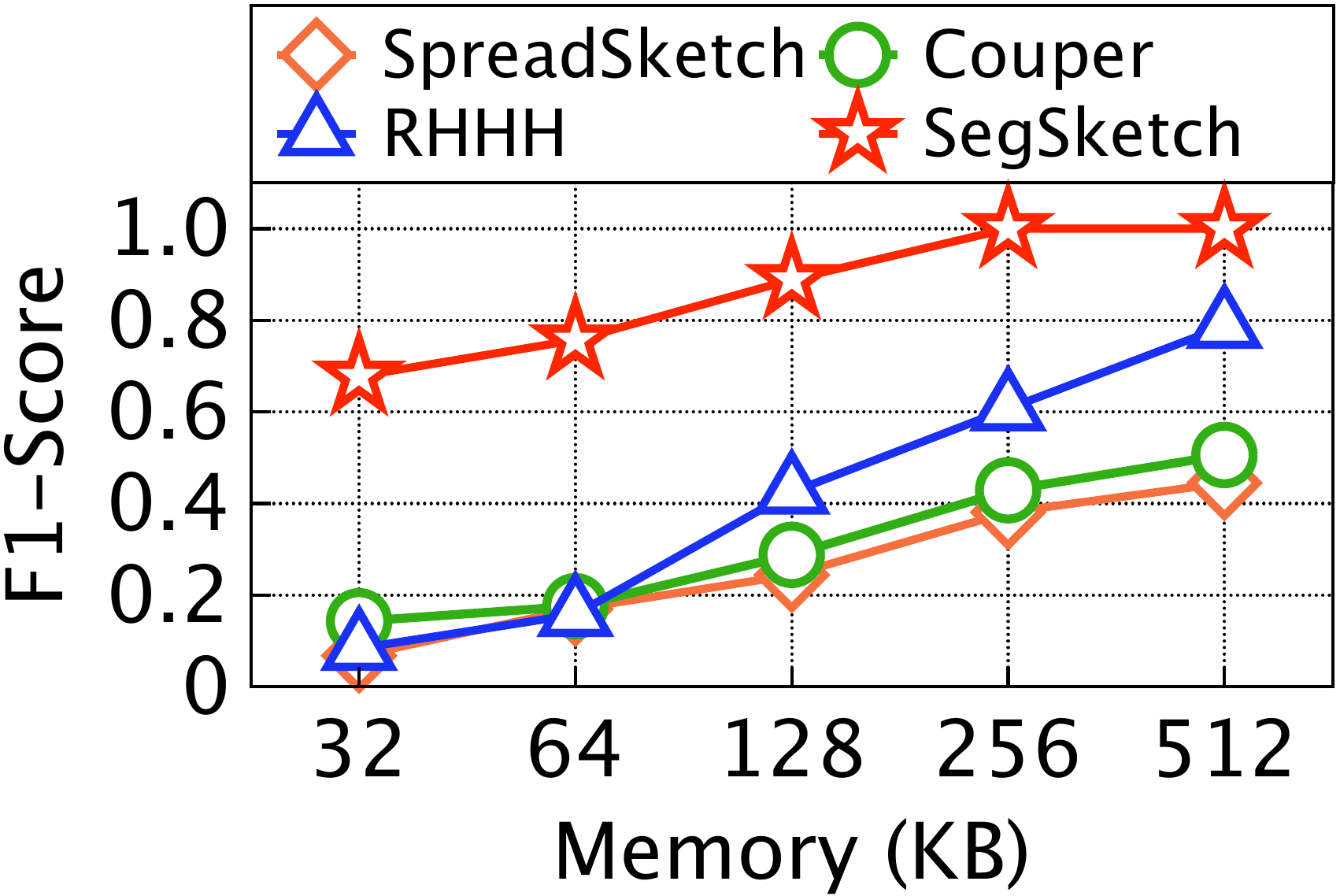}
    \label{fig:c}

    }
    % Case3
    \subfigure[ARE]{
        \includegraphics[width=0.465\linewidth]{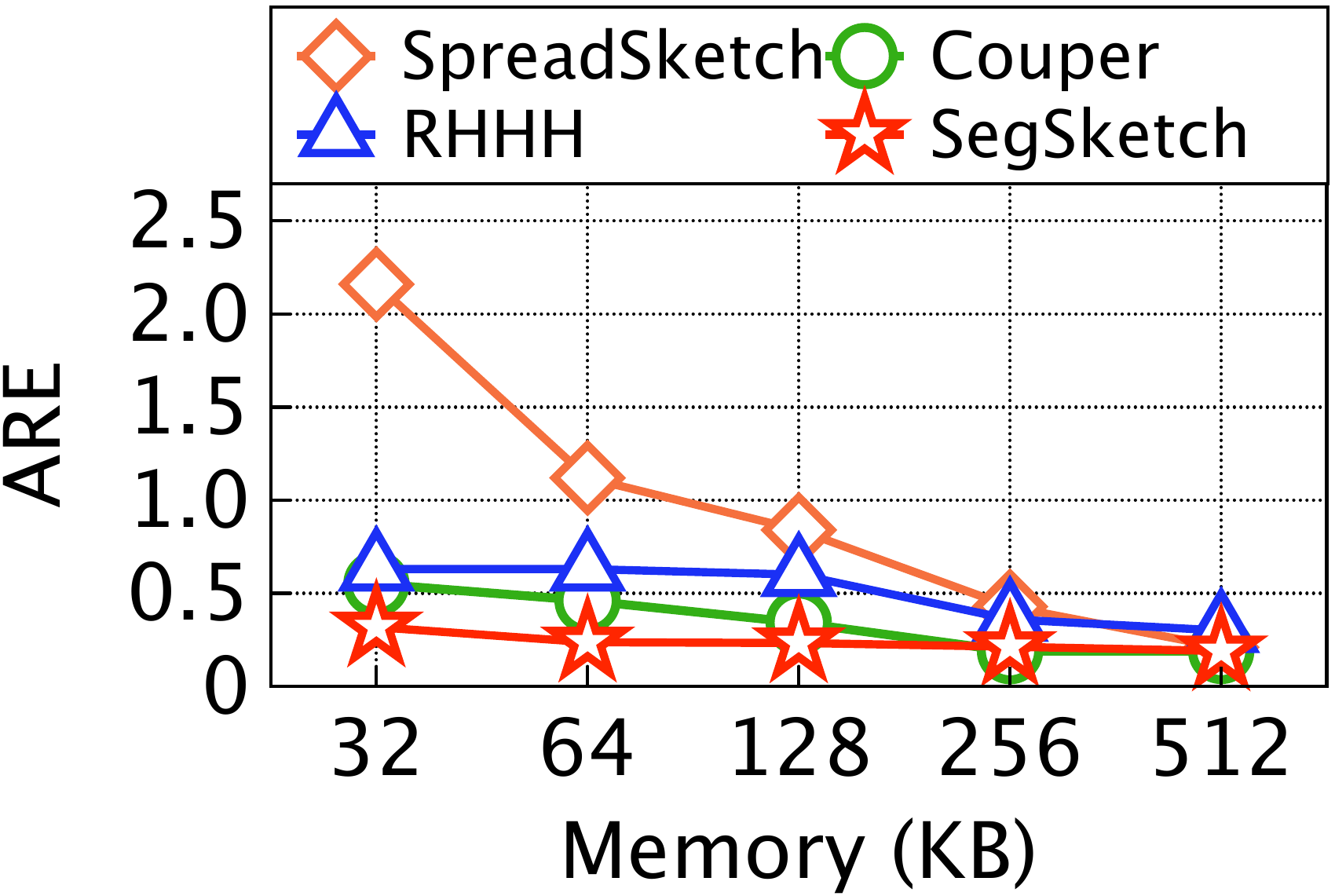}
    \label{fig:d}
    }
    \vspace{-1mm}
    \caption{Super receiver detection performance on the dataset mixing UNSW-NB15 with MAWI2021.}
    \label{fig:MAWI}
\end{figure}

\textbf{Super receiver detection}. For super receiver detection, we switch the host key to the destination IP address and conduct evaluation on the mixed MAWI dataset. Results in Figure~\ref{fig:MAWI} exhibit that, under 32KB memory, SegSketch improves F1-Score by $8.97\times$, $3.74\times$, and $7.07\times$ over SpreadSketch, Couper, and RHHH, respectively. These improvements highlight SegSketch’s advantage.

The primary performance improvement of SegSketch is still attributed to its subnet cardinality estimation mechanism, which reduces false positives by filtering out benign high-cardinality hosts that communicate with many distinct IP addresses across the whole network.

\subsection{Impact of IP Segment Width on Super Host Detection}

As shown in Figure \ref{fig:g1}, smaller $G$ enables finer segmentation and more accurate prefix estimation, leading to improved detection accuracy under limited memory. F1-Score increases as $G$ decreases from 8 to 2. Figure \ref{fig:g2} shows that ARE remains stable across different $G$, indicating that prefix estimation error has limited impact on subnet cardinality estimation. Given the trade-off between accuracy and computational overhead, we set $G=4$.

\vspace{-2mm}
\begin{figure}[htbp]
    \centering
    \subfigure[F1-Score of super host detection]{
        \includegraphics[width=0.72\linewidth]{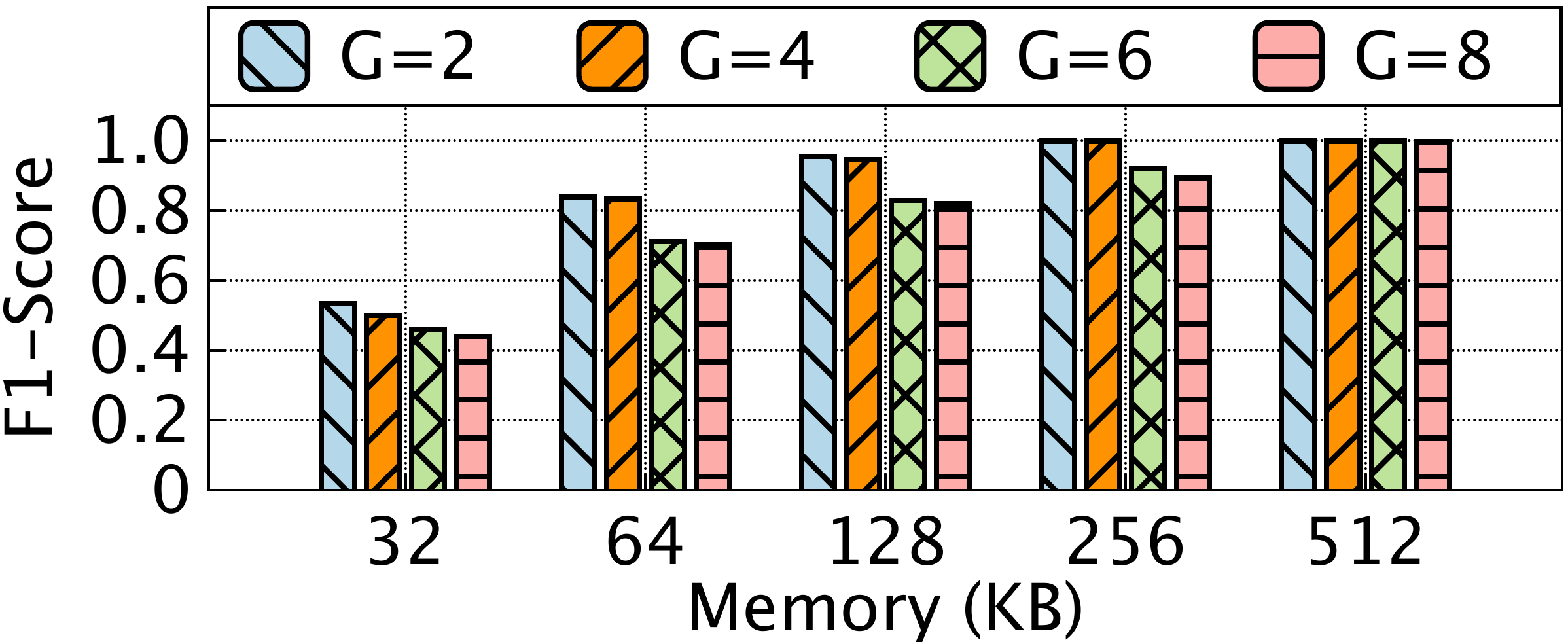}
  
     \label{fig:g1}
    }        
    \\
    % Case2
    \subfigure[ARE of super host detection]{
        \includegraphics[width=0.72\linewidth]{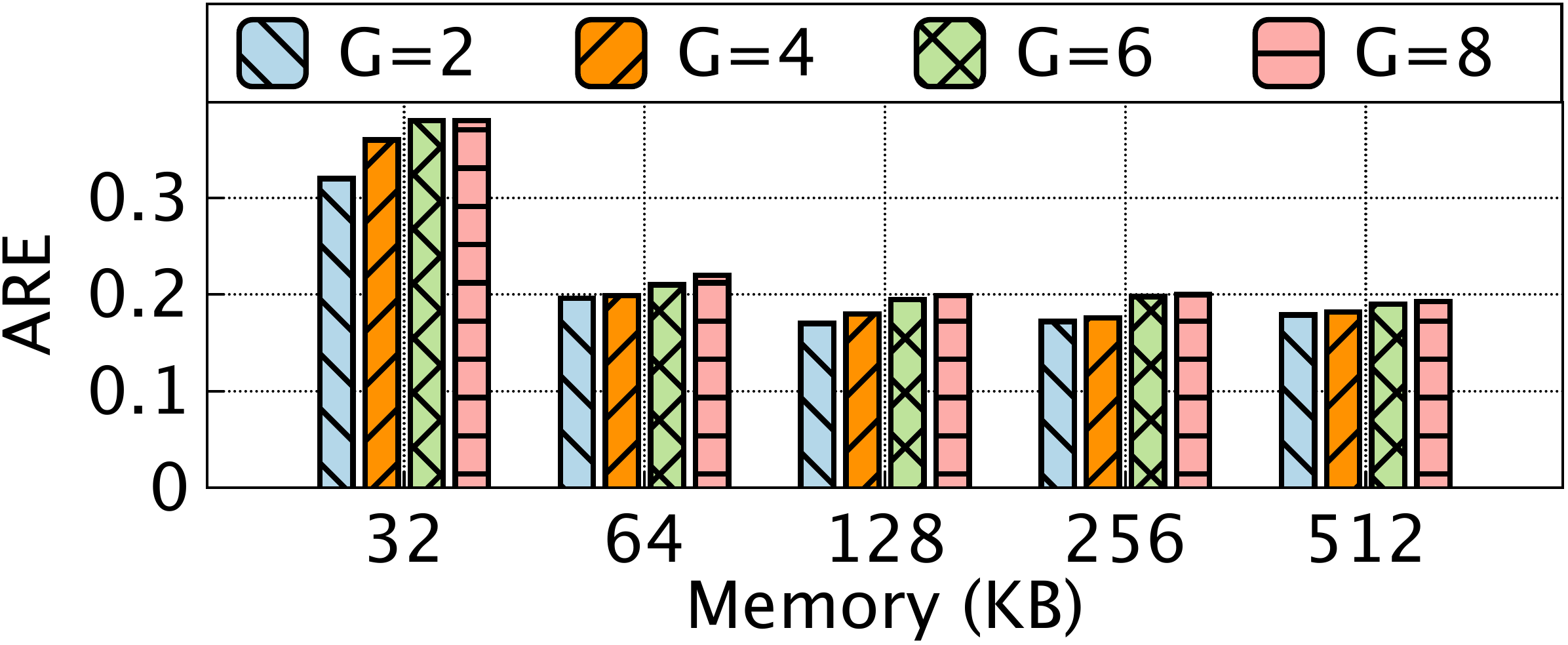}
 
    \label{fig:g2}

    }
  %\vspace{-3mm}
    \caption{Impact of varying IP segment width $G$.}
    %\vspace{-2mm}
    \label{fig:gg}
\end{figure}

\vspace{-2mm}
\subsection{Throughput}

We further evaluate the throughput of all four methods on the mixed MAWI dataset, with results shown in Figure~\ref{fig:th-mawi}. Consistent with the results on the mixed CAIDA dataset, SegSketch achieves the highest throughput across all memory settings. Under the tightest memory constraint of 32KB, it still maintains over 29 Mpps, significantly outperforming SpreadSketch, Couper, and RHHH.

\begin{figure}[htbp]
    \centering  
    \includegraphics[width=0.63\linewidth]{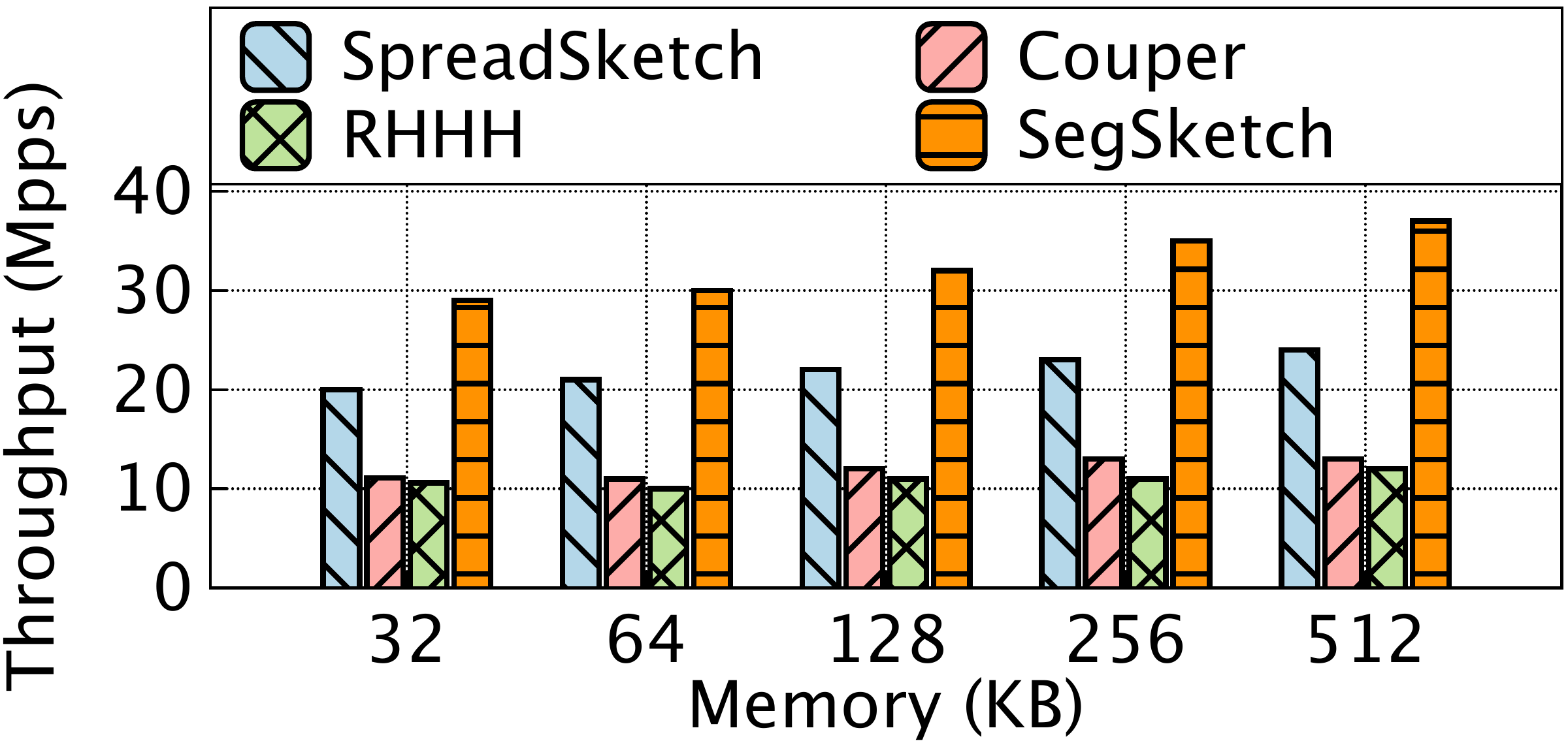}
    \caption{Throughput on the mixed MAWI dataset.}
    \label{fig:th-mawi}
\end{figure}

The efficiency of SegSketch stems from its computationally lightweight design, where the halved-segment hashing and the bitmap-based cardinality estimation enables fast updates. Conversely, SpreadSketch incurs extra cost from operations in Multi-Resolution Bitmaps, Couper requires maintaining dual-layer estimators, and RHHH suffers from complex multi-layer estimator updates. These observations further validate that SegSketch is not only memory-efficient and accurate, but also well-suited for online deployment in high-speed network environments.

%\section{EVALUATION}

%\subsection{Detection with Different Thresholds}

%\subsection{Deep Diving into SegSketch’s Operation}

\end{sloppypar}
\end{document}